\documentclass[A4paper,notitlepage,superscriptaddress,longbibliography,aps,twocolumn]{revtex4-1}
\usepackage{graphicx}
\graphicspath{{figures/}}
\usepackage{amsfonts}
\usepackage{amssymb}
\usepackage{mathrsfs}
\usepackage{soul}
\usepackage[usenames, dvipsnames]{color}
\definecolor{darkdarkgray}{gray}{0.1}
\usepackage[colorlinks=true,citecolor=blue,urlcolor=purple,linkcolor=darkdarkgray]{hyperref}

\usepackage{tikz}
\usepackage{pdfpages}
\makeatletter
\AtBeginDocument{\let\LS@rot\@undefined}
\makeatother

\usepackage[mathcal]{euscript}

\newcommand{\mbb}[1]{\mathbb{#1}}
\newcommand{\mcl}[1]{\mathcal{#1}}

\usepackage{amsmath}

\renewcommand{\sup}{\operatorname{sup}}

\renewcommand{\d}{\mathrm{d}}
\newcommand{\Tr}{\operatorname{Tr}}

\newcommand{\GL}{GL}

\newcommand{\ket}[1]{\left| #1 \right\rangle}

\newcommand{\ketbra}[2]{| #1 \rangle \langle #2 |}

\newcommand{\dprod}[2]{\left\langle #1, #2\right\rangle}

\newcommand{\abs}[1]{\left| #1\right|}

\newcommand{\norm}[1]{\left\| #1 \right\|}

\newcommand{\kommentar}[1]{}

\usepackage{amsthm}
\newtheorem{theorem}{Theorem}
\newtheorem{lemma}[theorem]{Lemma}
\newtheorem{proposition}[theorem]{Proposition}
\newtheorem{corollary}[theorem]{Corollary}

\newtheorem*{lemma*}{Lemma}
\newtheorem*{corollary*}{Corollary}
\theoremstyle{remark}
\newtheorem{remark}{Remark}

\usepackage{graphicx}

\usepackage{float}
\usepackage{color}
\definecolor{npurple}{rgb}{0.3,0,0.6}

\renewcommand{\S}{\mathcal{S}}

\newcommand{\M}{\mcl{M}}

\newcommand{\II}{\openone}
\newcommand{\I}{\openone}

\newcommand{\RR}{\mbb{R}}
\newcommand{\CC}{\mbb{C}}

\renewcommand{\P}{\mcl{P}}

\newcommand{\K}{\mcl{K}}

\renewcommand{\H}{\mcl{H}}

\newcommand{\B}{\mcl{B}}

\newcommand{\n}{\pmb{n}}
\newcommand{\x}{\pmb{x}}

\renewcommand{\a}{\pmb{a}}
\renewcommand{\b}{\pmb{b}}
\renewcommand{\c}{\pmb{c}}
\newcommand{\s}{\pmb{s}}
\renewcommand{\r}{\pmb{r}}
\newcommand{\vv}{\pmb{v}}

\newcommand{\G}{\mathcal{G}}

\newcommand{\U}{\operatorname{U}}
\renewcommand{\GL}{\operatorname{GL}}

\newcommand{\red}[1]{{\color{black}#1}}

\newcommand{\blue}[1]{{\color{black}#1}}

\definecolor{mygray}{gray}{0.6}

\newcommand{\new}[1]{{\color{black} #1}}

\begin{document}

\setstcolor{red}

\title{The geometry of Einstein-Podolsky-Rosen correlations}
\date{\today}

\author{H. Chau Nguyen}
\email{chau.nguyen@uni-siegen.de}
\affiliation{Naturwissenschaftlich-Technische Fakult\"at, Universit\"at Siegen,
Walter-Flex-Stra{\ss}e 3, 57068 Siegen, Germany}

\author{Huy-Viet Nguyen}
\email{nhviet@iop.vast.ac.vn}
\affiliation{Institute of Physics, Vietnam Academy of Science and Technology, 
10 Dao Tan, Hanoi, Vietnam}

\author{Otfried G\"{u}hne}
\email{otfried.guehne@uni-siegen.de}
\affiliation{Naturwissenschaftlich-Technische Fakult\"at, Universit\"at Siegen,
Walter-Flex-Stra{\ss}e 3, 57068 Siegen, Germany}

\begin{abstract}
Correlations between distant particles are central to many puzzles and
paradoxes of quantum mechanics and, at the same time, underpin various 
applications such as quantum cryptography and metrology. Originally in 1935, 
Einstein, Podolsky and Rosen (EPR) used these correlations to argue against 
the completeness of quantum mechanics. To formalise their argument, 
Schr\"odinger subsequently introduced the notion of quantum steering. Still, 
the question which quantum states can be used for EPR steering and which not 
remained open. Here we show that quantum steering can be viewed as an 
inclusion problem in convex geometry. For the case of two spin-$\frac{1}{2}$ 
particles, this approach completely characterises the set of states leading 
to EPR steering. \red{In addition, we discuss the generalisation to higher-dimensional 
systems as well as generalised measurements.} Our results find applications in 
various protocols in quantum information processing, and moreover they are 
linked to quantum mechanical phenomena such as uncertainty relations and the 
question which observables in quantum mechanics are jointly measurable.
\end{abstract}

\maketitle

%%%%%%%%%%%%%%%%%%%%%%%%%%%%%%%%%%%%%%%%%%%%%%%%%%%%%%%%%%%%%%%%%%%%%
%{\it Formalising the EPR argument.---}
%%%%%%%%%%%%%%%%%%%%%%%%%%%%%%%%%%%%%%%%%%%%%%%%%%%%%%%%%%%%%%%%%%%%%
In the simplest setting, the argument can be explained with two 
spin-$\frac{1}{2}$ particles, also called qubits, which are controlled 
by Alice and Bob at different locations \cite{epr, bohm}. The particles 
are in the singlet state, 
\begin{equation}
\ket{\psi}_{AB} = \frac{1}{\sqrt{2}}(\ket{01}-\ket{10}),
\label{eq-singlet}
\end{equation}
where $\ket{0}= \ket{\uparrow}_z$ and $\ket{1}= \ket{\downarrow}_z$
denote the two possible spin orientations in the $z$-direction. If Alice 
measures the spin of her particle in the $z$-direction, then, depending 
on the obtained result, Bob's state will be either in state $\ket{0}$ or 
state $\ket{1}$, due to the perfect anti-correlations of the singlet state. 
On the other hand, if Alice rotates her measurement device to measure the 
spin in the $x$-direction, Bob's conditional states are accordingly rotated 
to states $\ket{\uparrow}_x = \frac{1}{\sqrt{2}} (\ket{0}+\ket{1})$ or 
$\ket{\downarrow}_x = \frac{1}{\sqrt{2}}(\ket{0}-\ket{1})$ (see Figure~\ref{fig:steering}). 
So, by choosing her measurement, Alice can predict with certainty both the 
values of $z$- and $x$-measurements on Bob's side. According to EPR, this means 
that both observables must correspond to ``elements of reality''. As the quantum 
mechanical formalism does not allow one to assign simultaneously definite values to 
these observables, EPR concluded that quantum mechanics is incomplete.
As Schr\"odinger noted, Alice cannot transfer any information to Bob by choosing
her measurement directions, but she can determine whether the wave function on his 
side is in an eigenstate of the Pauli matrix $\sigma_x$ or $\sigma_z$. This {\it 
steering} of the wave function is, in Schr\"odinger's own words, ``magic'', as it 
forces Bob to believe that Alice can influence his particle from a 
distance~\cite{schroedingerletter, schroedingerpaper}. 

\begin{figure}[t!]
\begin{center}
\includegraphics[width=0.35\textwidth]{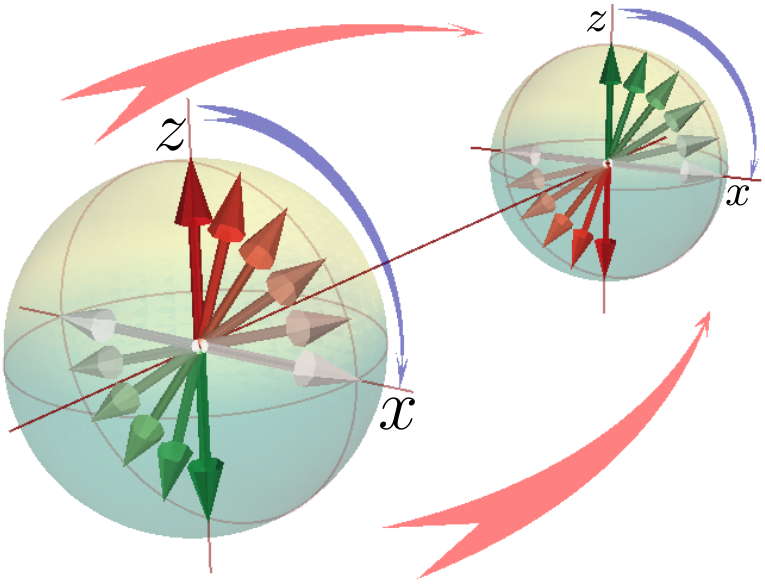}
\end{center}
\caption{Visualisation of the steering phenomenon: Alice (in the 
forefront) measures the spin of her particle in an arbitrary direction. 
Due to the quantum correlations of the singlet state, Bob's state 
(in the back) is projected onto the opposite direction. Bob cannot explain
this phenomenon by assuming pre-existing states at his location, 
so he has to believe that Alice can influence his state from a distance.
}
\label{fig:steering}
\end{figure}

The situation for general quantum states other than the singlet state
can be formalised as follows \cite{wiseman1}: Alice and Bob share a bipartite  
quantum state $\varrho_{AB}$ and Alice performs different measurements. For 
each of Alice's measurement setting $s$ and result $r$, Bob remains with a 
conditional state $\varrho_{r|s}$. These conditional states obey the condition 
$\sum_r \varrho_{r|s} = \varrho_B$, meaning that the reduced state 
$\varrho_B=\Tr_A(\varrho_{AB})$ on Bob's side is independent of Alice's choice 
of measurements. However, after characterising the states $\varrho_{r|s}$, Bob may try to 
explain their appearance as follows: He assumes that initially his particle was 
in some states $\sigma_\lambda$ with probability $p(\lambda)$, parametrised by 
some parameter $\lambda$. Then, Alice's measurement and result just gave him 
additional information on the probability of the states. This leads to states of 
the form \cite{wiseman1}
\begin{equation}
\varrho_{r|s} 
= p(r|s) \!\!\int\!\!d\lambda p(\lambda|r,s) \sigma_\lambda. 
\label{eq-lhsmodel}
\end{equation}
This can be interpreted as if the probability distribution $p(\lambda)$ is 
just updated to $p(\lambda|r,s)$, depending on the classical information about 
the result and setting $r,s$. If a representation as in equation~(\ref{eq-lhsmodel})
exists, Bob does not need to assume any kind of action at a distance to
explain the post-measurement states $\varrho_{r|s}$. Consequently, he does not need 
to believe that Alice can steer his state by her measurements and one also says that 
the state $\varrho_{AB}$ is  \emph{unsteerable} or has a local hidden state (LHS) 
model. If such a model does not exist, Bob is required to believe that Alice can steer 
the state in his laboratory by some action at a distance. In this case, the state 
is said to be \emph{steerable}.

So far, EPR steering has been observed in several experiments
\cite{experiment1, experiment2, experiment3, experiment4, experiment5, 
experiment6, experiment7, experiment8}, but the question which states 
can be used for EPR steering and which not remained, despite considerable 
theoretical effort \cite{wiseman2,quintinoinequivalence, pusey2013, gallegoresource, chaujpa,
steeringsdp, brunnerlhsrecent, bowles2014, chauepl, chaujpa, Jevtic2015a, chaupra, paulsdp, 
gdansk, yu1, yu2}, open. It is known that the set 
of steerable quantum states is strictly smaller than the set of entangled 
states and strictly larger than the set of states leading to a Bell inequality
violation. But both entanglement and Bell nonlocality are not
well understood \cite{brunnerreview, horodeckireview}; only for the case 
of small dimensions \red{or special families of states} the famous Peres-Horodecki 
criterion provides an exact characterisation of the entangled states \cite{peresppt, horodeckippt}. 
In this paper we present a solution to the problem of steerability for 
the case of projective measurements carried out on two qubits.
\red{The generalisation of the technique to higher-dimensional systems 
as well as taking into account generalised measurements is possible.}

%%%%%%%%%%%%%%%%%%%%%%%%%%%%%%%%%%%%%%%%%%%%%%%%%%%%%%%%%%%%%%%%%%%%%%%%%
{\it Conditional states and LHS models.---}
%%%%%%%%%%%%%%%%%%%%%%%%%%%%%%%%%%%%%%%%%%%%%%%%%%%%%%%%%%%%%%%%%%%%%%%%%
Let us characterise the conditional states and possible LHS models. 
For the former, we note that any bipartite quantum state $\varrho_{AB}$ 
defines a map $\Lambda$ from operators on Alice's space to operators on 
Bob's space via
\begin{equation}
\Lambda(X_A) = \Tr_A(\varrho_{AB} X_A \otimes \openone_B). 
\end{equation}
This map characterises the conditional states as follows: A result of 
a measurement setting is described by an effect $E_{r|s}$ which is an 
operator with positive eigenvalues not larger than one. The conditional 
state is then just given by 
$\varrho_{r|s}= \Tr_A(\varrho_{AB} E_{r|s} \red{\otimes \openone_B}) = \Lambda(E_{r|s}).$ 

For our approach it is important that $\Lambda$ has a clear geometrical meaning 
(see Figure~\ref{fig:geometry1}). The set of measurement effects on Alice's side, 
denoted by $\M_A=\{E_{r|s}\;|\;0\leq E_{r|s} \leq \openone_A \}$, is a four-dimensional
double cone, where $0$ and $\openone_A$ correspond to the south- and north pole, and the 
pure effects of the form $E_{r|s}=\ketbra{\psi}{\psi}$ constitute the equator, which is 
nothing but Alice's Bloch sphere. The map $\Lambda$ is linear and maps this
double cone to a  smaller double cone, denoted by $\Lambda(\M_A)$, which we
call the {\it set of steering outcomes}~\cite{chaupra}. For our purposes, we can
assume without loss of generality that the map $\Lambda$ is invertible; the 
proof of this and all forthcoming mathematical statements, can be found in 
the Appendix \cite{appremark}.

\begin{figure}[t]
\begin{center}
\includegraphics[width=0.5\textwidth]{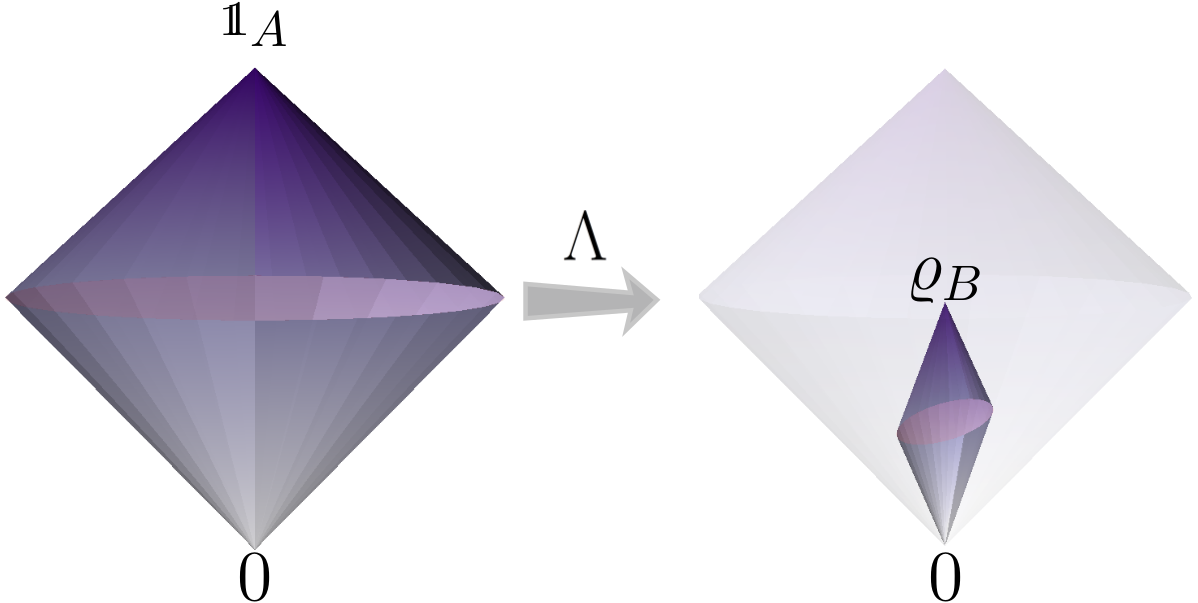}
\end{center}
\caption{Geometrical view of the map $\Lambda$: The set 
of measurement effects $\M_A$ on Alice's side is a four-dimensional 
double cone, where $0$ and $\openone_A$ correspond to the south- and north 
pole and the equator is formed by the Bloch sphere. Under the action of 
the linear map $\Lambda$ this double cone is mapped onto a subset of itself, 
with $\Lambda(0)=0$ and $\Lambda(\openone_A)=\varrho_B$. The resulting set of 
steering outcomes is completely characterised by $\varrho_B$ and the 
image of the equator under $\Lambda$.
}
\label{fig:geometry1}
\end{figure}

Let us now characterise the set of all possible LHS models. 
We first restrict our attention to projective measurements on two qubits, 
later we discuss the general case. Projective 
measurements are described by two orthogonal projectors 
$E_{+|s}$ and $E_{-|s}$ summing up to the identity, 
$E_{+|s} + E_{-|s} = \openone_A.$ It is known that the LHS 
model~\eqref{eq-lhsmodel} can be rewritten as \cite{chaujpa}
\begin{equation}
\varrho_{\pm|s} 
= \Lambda(E_{\pm|s}) 
= \int_{\sigma \in \mathcal{B}_B}  \!\!\!\!\! \d\mu(\sigma) G_{\pm|s}( \sigma) \sigma,
\end{equation}
with an integration over a probability distribution $\mu$ over all pure and mixed 
states in Bob's Bloch ball $\mathcal{B}_B$. The so-called response functions 
$G_{\pm|s}(\sigma)$ are positive and normalised as $G_{+|s}+G_{-|s}=1$, which
implies that they always have to obey the minimal requirement
\begin{equation}
\varrho_B = \Lambda(\openone_A) = \int_{\sigma \in \mathcal{B}_B} \!\!\!\!\!\d \mu (\sigma) \sigma. 
\label{eq:minimal_requirement}
\end{equation}
In this scenario the set of all conditional states $\varrho_{\pm|s}$ 
that can be modelled with an LHS model is characterised by the probability 
distribution $\mu$ only. We call this set the {\it capacity of $\mu$}
and denote it by~\cite{chaupra,chaujpa} 
\begin{equation}
\K(\mu) = \Big\{ K = \int_{\sigma \in \mathcal{B}_B} \!\!\!\!\! \d \mu (\sigma) g(\sigma) \sigma : 
0 \le g(\sigma) \le 1\Big\}.
\end{equation}
%

%%%%%%%%%%%%%%%%%%%%%%%%%%%%%%%%%%%%%%%%%%%%%%%%%%%%%%%%%%%%%%%%%%%%%%%%%
{\it The geometric approach.---}
%%%%%%%%%%%%%%%%%%%%%%%%%%%%%%%%%%%%%%%%%%%%%%%%%%%%%%%%%%%%%%%%%%%%%%%%%
In order to decide steerability, one has to compare the set of steering 
outcomes  with the possible capacities. If one
finds an LHS ensemble $\mu$ for which $\Lambda(\M_A)$ is a subset of $\K(\mu)$, 
then $\varrho_{AB}$ is not steerable. On the other hand, if $\K(\mu)$ does not 
cover $\Lambda(\M_A)$ for all $\mu$, then $\varrho_{AB}$ is steerable.

\begin{figure}[t]
\begin{center}
\includegraphics[width=0.45\textwidth]{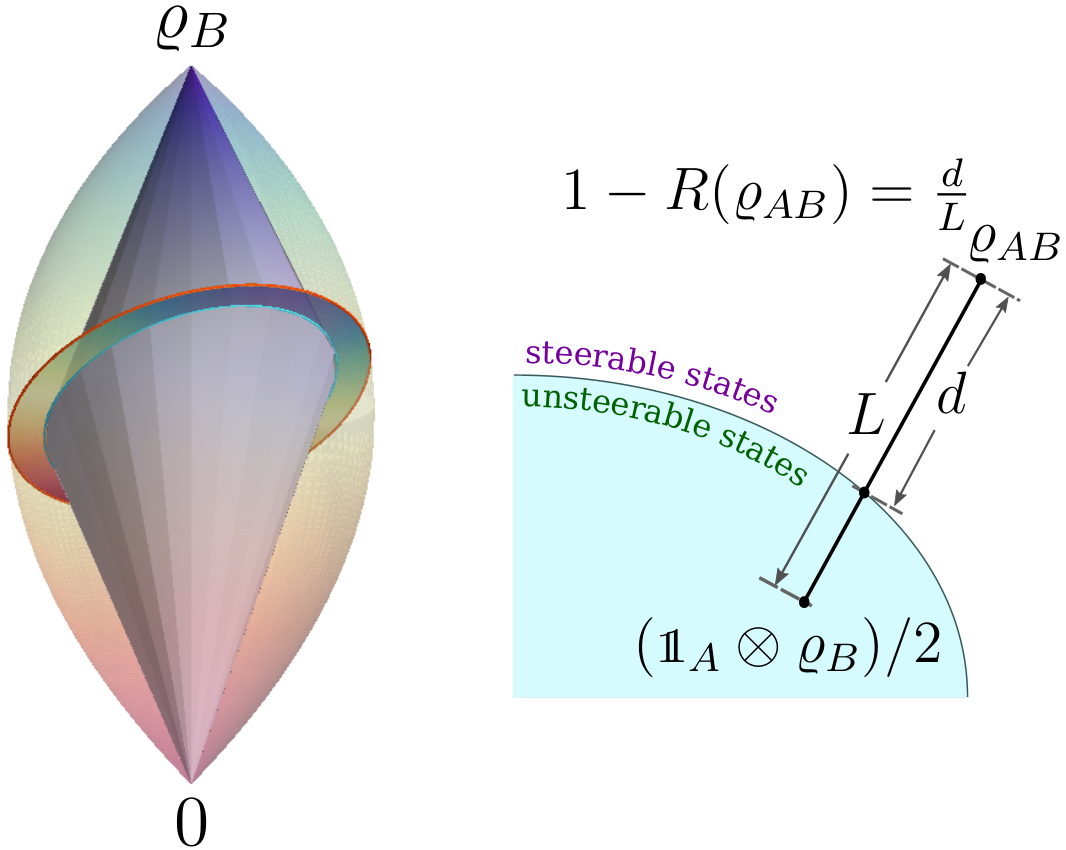}
\end{center}
\caption{(left) The geometrical interpretation of the critical radius:
The capacity $\K(\mu)$ is a convex set containing $0$ and $\varrho_B$. The double
cone {$\Lambda(\M_A)$} has $0$
and $\varrho_B$ as south- and north
pole, so {$\Lambda(\M_A)$} is contained in
$\K(\mu)$ if and only if its equator (cyan)
is in $\K(\mu)$. This can be checked by computing the radius of $\K(\mu)$
in the appropriate plane and metric (red). (right) Operational 
meaning of the critical radius: $1-R(\varrho_{AB})$ measures the 
distance from $\varrho$ to the surface of unsteerable/steerable states 
relatively to $({\openone_A} \otimes \varrho_B)/2$.}
\label{fig:geometry2}
\end{figure}

Checking the inclusion relation between these sets
is simplified by geometry, see Figure~\ref{fig:geometry2}. 
$\K(\mu)$ is a convex set which contains $0$ and $\varrho_B$. 
The double cone $\Lambda(\M_A)$ is contained in this set if and 
only if its equator is contained in $\K(\mu)$. If we choose the metric 
appropriately, the equator of $\Lambda(\M_A)$ is a ball of radius one. 
Whether $\K(\mu)$ contains the ball or not, can thus be determined by 
calculating the {\it principal radius}, defined as the minimal distance 
from the boundary of $\K(\mu)$ in the equator hyperplane to the centre 
of the ball~\cite{chauepl}.

Our first main result is that the principal radius for a given 
probability distribution can be computed as a simple optimisation 
problem, given by
\begin{equation}
r(\varrho_{AB},\mu)= \min_{C} \frac{1}{\sqrt{2}\Vert{\Tr_B[\bar{\varrho} (\openone_A \otimes C)]}\Vert}
\int_{\sigma \in \mathcal{B}_B} \!\!\!\!\!\!\!\!\!\!\!\!\! \d \mu (\sigma)  |\Tr_B({C}{\sigma})|,
\label{eq:simple_r-maintext}
\end{equation}
where $\bar{\varrho}=\varrho_{AB}-(\openone_{A} \otimes \varrho_B)/2$,
the norm is given by $\norm{X}=\sqrt{\Tr(X^\dagger X)}$, and  the
minimisation runs over all single-qubit observables 
$C$ on Bob's space. \red{The proof of Eq.~\eqref{eq:simple_r-maintext}
relies on the Bloch representation and is given in the Appendix 
\cite{appremark}.}

Equation~\eqref{eq:simple_r-maintext} allows us to compute the principal radius 
for a given distribution $\mu$ over states in Bob's Bloch ball. It remains to 
maximise this over all possible probability distributions. This leads to the 
{\it critical radius} 
\begin{equation}
R(\varrho_{AB}) = \max_{\mu} r(\varrho_{AB},\mu).
\label{eq:rcrit-def}
\end{equation}
In this way, we have reduced the characterisation of steering to the
computation of the critical radius and we can formulate: {\it A two-qubit state
can be used for EPR steering, if and only if the critical radius is smaller
than one.} All that remains to be done is to characterise the critical radius and 
to provide efficient methods for computing it. \red{Showing the existence of the 
maximum in Eq.~\eqref{eq:rcrit-def} requires careful continuity arguments as explained 
in the Appendix \cite{appremark}.} 

%%%%%%%%%%%%%%%%%%%%%%%%%%%%%%%%%%%%%%%%%%%%%%%%%%%%%%%%%%%%%%%%%%%%%%%%%
{\it Properties of the critical radius.---}
%%%%%%%%%%%%%%%%%%%%%%%%%%%%%%%%%%%%%%%%%%%%%%%%%%%%%%%%%%%%%%%%%%%%%%%%%
The
first interesting property of the critical radius is its {\it scaling}. Given 
a two-qubit state, we can consider a family of states by mixing 
it with a special kind of separable noise,
\begin{equation}
\varrho_\alpha^{\rm noise} = \alpha \varrho_{AB} 
+ (1-\alpha) \frac{\openone_A}{2} \otimes \varrho_B,
\end{equation}
where $0 \leq \alpha \leq 1$. For these states, we can show that
\begin{equation}
R(\varrho_\alpha^{\rm noise}) = \frac{1}{\alpha} R(\varrho_{AB}).
\end{equation}
This implies that computing the critical radius for $\varrho_{AB}$ also gives 
its values on the entire line in the state space parametrised by 
$\varrho_\alpha^{\rm noise}$. 
This scaling sheds light on the operational meaning of the critical radius: 
$1-R(\varrho_{AB})$ measures the distance from $\varrho_{AB}$ along this line to the border 
between steerable and unsteerable states relatively to 
$({\openone_A}\otimes \varrho_B)/2$.

The second important property is the {\it symmetry} of the critical radius. 
Given a state $\varrho_{AB}$, we consider the family of states
\begin{equation}
\tilde{\varrho}_{AB} = \frac{1}{\mathcal{N}}
(U_A \otimes V_B) \varrho_{AB} (U^\dagger_A \otimes V^\dagger_B), 
\label{eq:localsymm}
\end{equation}
where $U_A$ is a unitary matrix on Alice's side, $V_B$ is an invertible
matrix on Bob's side, and $\mathcal{N}$ denotes
the normalisation. For this family of states one can show that 
$R(\varrho_{AB})=R(\tilde{\varrho}_{AB})$. This symmetry of the critical 
radius thus generalises and formalises quantitatively the early observation 
that the existence of an LHS model is invariant under Alice's local unitary 
and Bob's local filtering {operations}~\cite{gallegoresource, quintinoinequivalence, 
roope1}.
One may ask to which extent a mixed two-qubit state can be simplified with 
transformations as in equation~(\ref{eq:localsymm}). The answer is that any 
entangled state can be brought into a canonical form without changing its 
critical radius. In the canonical form, $\varrho_B = {\openone_B}/{2}$ is 
maximally mixed and, in addition, all two-body correlations  vanish, up to the diagonal ones, 
$s_i = \Tr(\varrho_{AB}\sigma_i \otimes \sigma_i)$ for $i=x,y,z$. So the critical 
radius of a state is uniquely determined by six parameters, coming from the reduced 
state of Alice, parametrised by $a_i=\Tr(\varrho_{AB}\sigma_i \otimes \openone_B)$ 
and by a diagonal $3 \times 3$-matrix $T$.

Some facts about steering follow directly from the two properties mentioned 
above. First, as any pure entangled state $\ket{\psi}$ is equivalent to a Bell state in the sense of equation~(\ref{eq:localsymm}), one can easily show that $R(\ketbra{\psi}{\psi})={1}/{2}$. Second, 
the previous properties allow for characterising the convex sets 
$Q_t= \{\varrho_{AB}: R(\varrho_{AB}) \ge t \}$ and one can, for some cases, compute the
tangent hyperplanes, resulting in optimal steering inequalities. Finally, generalising 
equation~(\ref{eq:localsymm}), $R$ is also invariant under the inversion of the Bloch 
sphere of either of the parties. This is rather surprising as entanglement of two-qubit 
states is equivalent to the occurrence of negative eigenvalues after partial transposition 
\cite{peresppt,horodeckippt}, which can be seen as a local inversion of
the Bloch sphere. So, entanglement and quantum steering are, in fact, types of quantum 
correlations with fundamentally different mathematical structures.

\begin{figure}[t]
%\begin{center}
%\begin{minipage}[c]{0.25\textwidth}
%\includegraphics[width=\textwidth]{polytopes.png}
%\end{minipage}
%\begin{minipage}{0.01\textwidth}
%\mbox{ }
%\end{minipage}
%\begin{minipage}[c]{0.16\textwidth}
%    \begin{tikzpicture}
%    \node at (0,0) {\includegraphics[width=\textwidth]{{ku.png}}};
%    \node at (0.07,-2.6) {{\Large $0$}};
%    \node at (0.07,2.6) {{\Large $\varrho_B$}};
%    \end{tikzpicture}
%\end{minipage}
%\end{center}
\begin{center}
\includegraphics[width=0.45\textwidth]{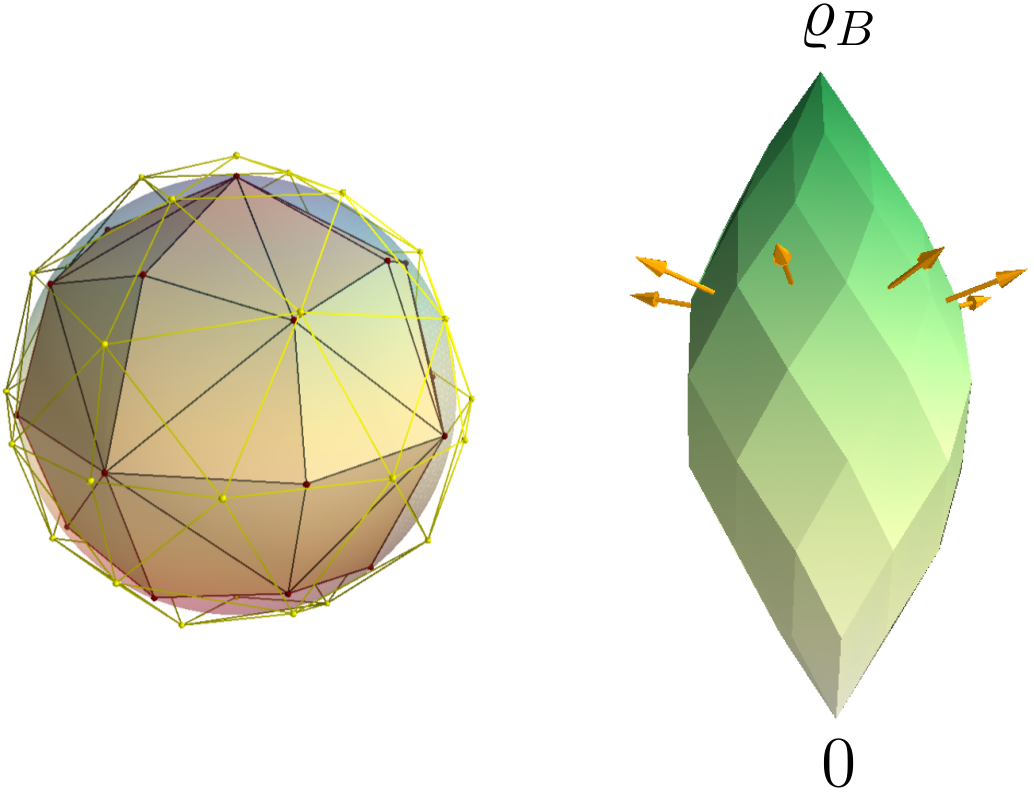}
\end{center}
\caption{(left) 
In order to characterise all probability distributions 
on the Bloch sphere, one can use inner and outer approximations of the sphere 
by polytopes. For the polytopes and the optimisation problem in
equation~(\ref{eq:simple_r-maintext}) it suffices to consider probability 
distributions supported at the extremal points. 
(right) For a given polytope, the capacity $\K(\mu)$ is a polytope
again. Consequently, when computing the principal radius it suffices
to consider the (finite) set of directions corresponding to the 
faces of the capacity polytope. 
}
\label{fig:geometry3}
\end{figure}

%%%%%%%%%%%%%%%%%%%%%%%%%%%%%%%%%%%%%%%%%%%%%%%%%%%%%%%%%%%%%%%%%%%%%%%%%
{\it Computation of the critical radius.---}
%%%%%%%%%%%%%%%%%%%%%%%%%%%%%%%%%%%%%%%%%%%%%%%%%%%%%%%%%%%%%%%%%%%%%%%%%
For practical convenience, the calculation of the critical radius of a generic 
state is carried out starting from its canonical form. Then, in order to evaluate 
equation~(\ref{eq:rcrit-def}) one needs to characterise the possible distributions 
$\mu$. Instead of maximising over all probability distributions on the Bloch ball, 
we approximate the ball by inner or outer polytopes as illustrated in 
Figure~\ref{fig:geometry3}. 
Crucially, for the special function in equation~(\ref{eq:simple_r-maintext}) one 
can show that optimising over probability distributions supported at the {vertices} 
of the outer (inner) polytope leads to an upper (lower) bound $R_{\rm out}$ ($R_{\rm in}$) for the critical radius. 
One may even simplify the calculation: If the inner polytope is chosen to have inversion 
symmetry, one has $R_{\rm in} \leq R(\varrho_{AB}) \leq R_{\rm in}/r_{\rm in}$, 
where $r_{\mathrm{in}}$ is the inscribed radius of the polytope. Then the relative
difference between the bounds depends on the polytope only and not on details of 
the state. This bound also shows that as $r_{\mathrm{in}}$ converges to $1$ one obtains 
an asymptotically exact value for $R(\varrho_{AB})$.

For a given polytope with $N$ vertices, the calculation of the critical 
radius proceeds as follows: The capacity $\mathcal{K}(\mu)$ is a polytope in the
four-dimensional space with $O(N^3)$ facets. When computing the critical
radius, it suffices to consider the finite set of operators $C$ that
correspond to normal vectors of these facets, and these operators do 
not depend on the probability distribution on the polytope. As a consequence, 
the optimisation over probability distributions is formulated as a linear
program of finite size.

To illustrate the power of the method, we show examples of two-dimensional 
random cross-sections of the set of two-qubit states, see 
Figure~\ref{fig:cross_sections}. We observe that the computed upper and lower 
bounds for the critical radius are very tight even when a polytope with 
$252$ vertices was used. A detailed
discussion including further examples of states is given in the Appendix \cite{appremark}.

Prior to our work, certain necessary and sufficient conditions for steering were proposed~\cite{yu1,yu2}, however their computability cannot be generally illustrated. 
There have been also attempts in estimating the boundary of 
the set of unsteerable states for special families of states with 
semidefinite programming (SDP) \cite{paulsdp, cavalcantisdp, steeringsdp, 
brunnerlhsrecent}. However, the SDP size increases exponentially 
with the number of measurements used to approximate the set of all measurements. 
This limitation hinders the accurate locating of the boundary even for special choices of 
cross-sections. Contrary to that, here we obtained a linear program, 
of which the size increases cubically with the number of approximated points. 
Both lower bound and upper bound with a pre-defined difference less that $1\%$ 
for the critical radius of a generic state can be easily obtained in
a reasonable computational time. Our implementation is available at a public repository~\cite{gitlab}.

\begin{figure}
\begin{center}
\includegraphics[width=0.21\textwidth]{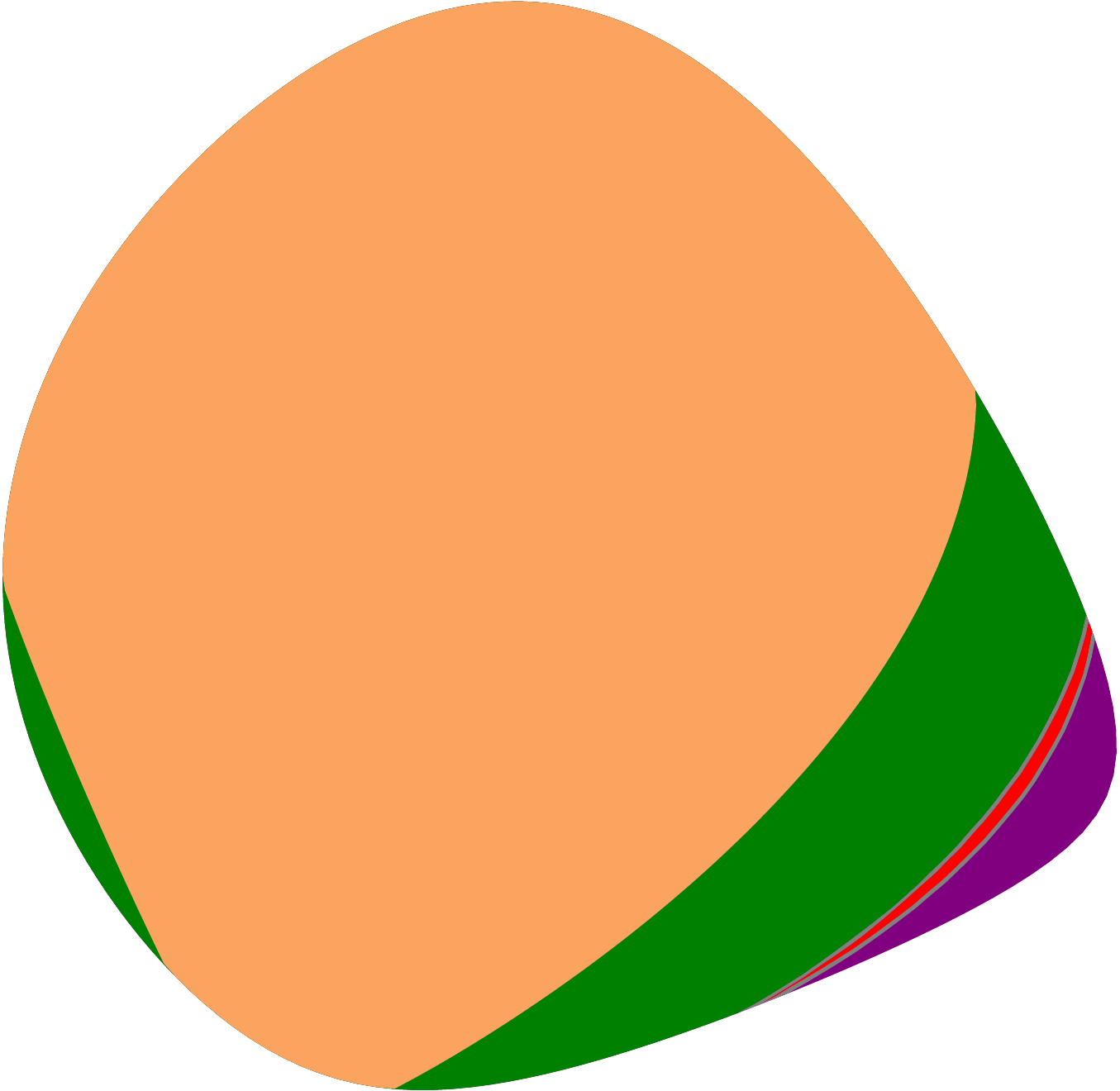}
\hspace{0.3cm}
\includegraphics[width=0.21\textwidth]{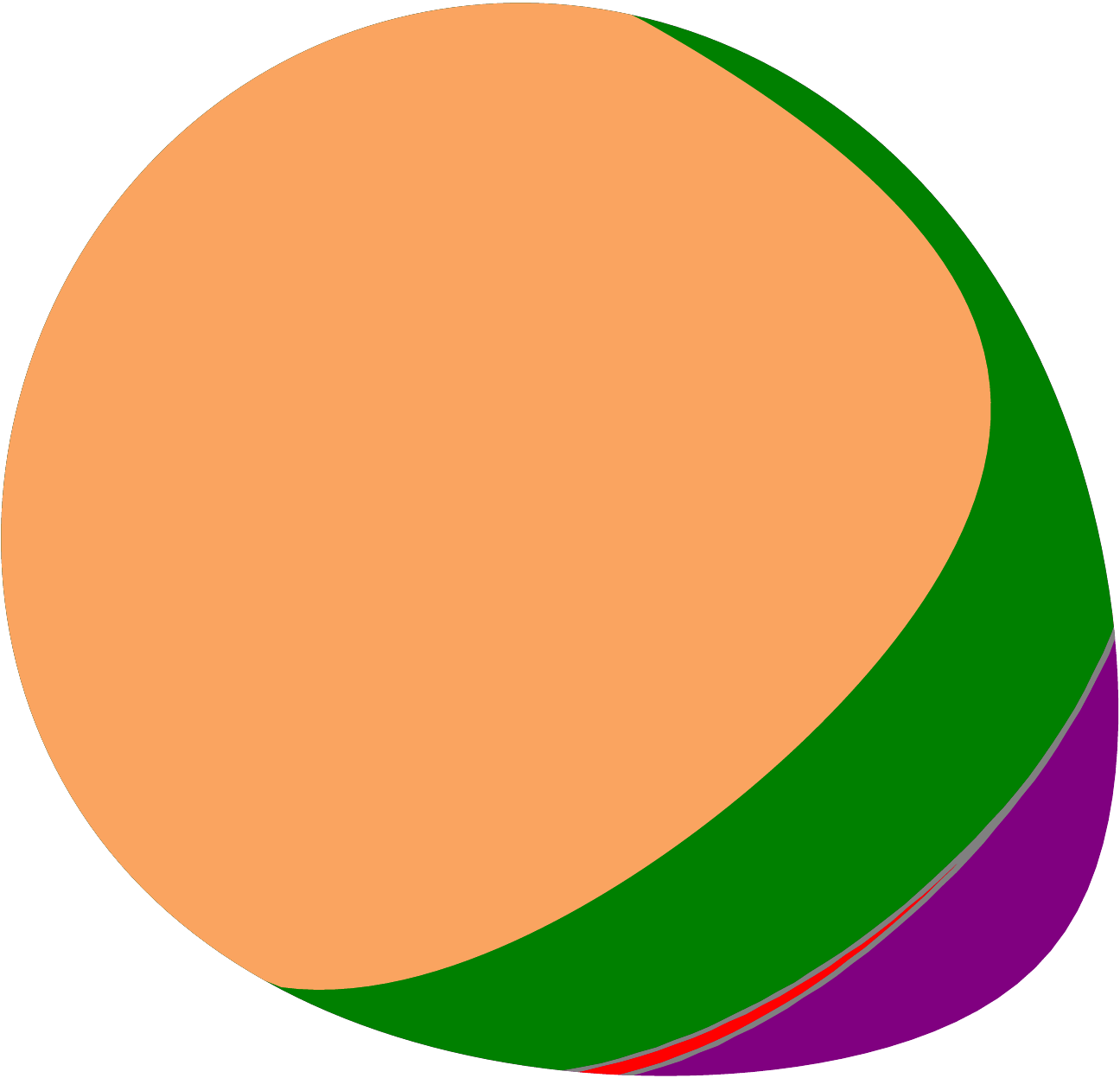}
\\ {\ } \\
\includegraphics[width=0.47\textwidth]{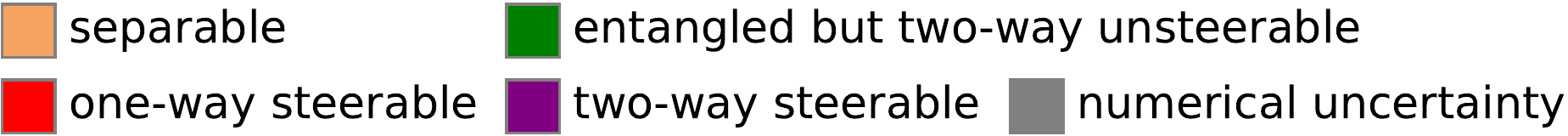}
\end{center}
\caption{Two two{-dimensional} random cross-sections of the set 
of all two-qubit states. As EPR steering is not symmetric under 
the exchange of Alice and Bob, one can distinguish different classes
of steerable states. The colours denote the set of separable states
(characterised by the partial transposition \cite{peresppt,horodeckippt}), 
entangled states that are unsteerable, one-way steerable states (Alice to 
Bob or vice versa), and two-way steerable states (Alice to Bob and vice 
versa). The very thin grey areas denote the states where the used numerical precision
was not sufficient to make an unambiguous decision.}
\label{fig:cross_sections}
\end{figure}

Finally, we note that certain analytical bounds for the critical radius 
can also be derived from our approach. For example, for a state in the 
canonical form, it can be shown that 
\begin{equation}
2 \pi N_T \abs{\det(T)} \ge R(\varrho_{AB}) \ge \frac{2 \pi N_T \abs{\det(T)}}{1+ \norm{T^{-1}\vec{a}}},
\end{equation} 
where $\vec a = (a_x, a_y, a_z)$ is the Bloch vector of Alice's reduced state,
$N_T^{\red{-1}}=\int \d S(\vec{n}) [\vec{n}^T T^{-2} \vec{n}]^{-2}$ and the integration runs 
over the surface of the unit sphere. If $\vec{a}=0$, these bounds {recover}
the exact formula for the critical radius of Bell diagonal states \cite{Jevtic2015a,chauepl}.

\red{
%%%%%%%%%%%%%%%%%%%%%%%%%%%%%%%%%%%%%%%%%%%%%%%%%%%%%%%%%%%%%%%%%%%%%%%%%
{\it Generalised measurements and higher-dimensional systems.---}
%%%%%%%%%%%%%%%%%%%%%%%%%%%%%%%%%%%%%%%%%%%%%%%%%%%%%%%%%%%%%%%%%%%%%%%%%
A similar formula for the principal and critical radius can be derived 
for generalised measurements (i.~e., positive operator-valued measures--POVMs) 
and higher-dimensional systems, despite their more complicated geometry.  
As we explain in the Appendix \cite{appremark}, many properties of the critical 
radius, such as its scaling and its symmetry still hold. The fundamental question 
arises whether generalised measurements are more useful for steering than the 
standard projective measurements considered so far. For two qubits, numerical 
estimation of the principal radii for POVMs provides clear evidence that, 
for a generic probability distribution $\mu$, the principal radius for POVMs 
is the same as that for projective measurements. This encourages us to conjecture 
that POVMs do not give any advantage in EPR steering for the case of two-qubit 
states.}

%%%%%%%%%%%%%%%%%%%%%%%%%%%%%%%%%%%%%%%%%%%%%%%%%%%%%%%%%%%%%%%%%%%%%%%%%
{\it Discussion.---}
EPR steering is an asymmetric phenomenon where Bob, contrary to Alice, has 
well characterised measurements. Consequently the underlying correlations 
find applications in non-symmetric scenarios of quantum information processing, 
such as one-sided device-independent quantum key distribution \cite{branciard} 
or sub-channel discrimination \cite{piani}. Clearly, our solution to the steering 
problem helps to understand and optimise these applications and their experimental 
realisations.

In addition, there are far-ranging consequences. First, it has 
been established that steering is in one-to-one correspondence with 
the  question which measurements in quantum mechanics can be jointly 
measured \cite{quintinojm, roope1, roope2, heinosaari}. Second, recent works 
established close connections between quantum steering and entropic uncertainty 
relations \cite{anaentropic, frowisrecent}. Joint measurability and entropic 
uncertainty relations are central for many applications of quantum physics, such
as the security of quantum key distribution \cite{colesuncertainty}. We expect 
that our results and methods presented here may shed new light on these topics 
in the near future. 

%%%%%%%%%%%%%%%%%%%%%%%%%%%%%%%%%%%%%%%%%%%%%%%%%%%%%%%%%%%%%%%%%%%%%%%%%
{\it Acknowledgements.---}
%%%%%%%%%%%%%%%%%%%%%%%%%%%%%%%%%%%%%%%%%%%%%%%%%%%%%%%%%%%%%%%%%%%%%%%%%
%\begin{acknowledgements}
We are grateful to Fabian Bernards, Francesco Buscemi, Shuming Cheng, Ana C.~S.~Costa, Michael J.~W. Hall, 
Teiko Heinosaari, C. Jebaratnam, Sania Jevtic, X.~Thanh Le, Antony Milne, V.~Anh Nguyen, Q.~Dieu Nguyen, 
Jiangwei Shang, Roope Uola, Howard M.~Wiseman, Xiao-Dong Yu, and particularly N.~Duc Le for helpful comments and discussions. 
We thank the authors of Ref.~\cite{brunnerlhsrecent} for kindly providing 
us with their SDP data.  
This work was supported by the DFG and the ERC (Consolidator Grant 683107/TempoQ).
\red{CN also acknowledges the support by the Vietnam National Foundation for 
Science and Technology Development (NAFOSTED) under grant number 103.02-2015.48.
HVN acknowledges the financial support of the International Centre of Physics at the Institute of Physics, Vietnam Academy of Science and Technology.}
%\end{acknowledgements}

%%%%%%%%%%%%%%%%%%%%%%%%%%%%%%%%%%%%%%%%%%%%%%%%%%%%%%%%%%%%%%%%%%%%%%%%%
%{\bf Methods.---}

%\clearpage
%\newpage
%\input{appendix-critical-radii}
%%%%%%%%%%%%%%%%%%%%%%%%%%%%%%%%%%%%%%%%%%%%%%%%%%%%%%%%%%%%%%%%%%%%%%%%%
\let\oldaddcontentsline\addcontentsline% Store \addcontentsline
\renewcommand{\addcontentsline}[3]{}% Make \addcontentsline a no-op

\let\addcontentsline\oldaddcontentsline% Restore \addcontentsline
%%%%%%%%%%%%%%%%%%%%%%%%%%%%%%%%%%%%%%%%%%%%%%%%%%%%%%%%%%%%%%%%%%%%%%%%%

\clearpage
\newpage

\appendix

\begin{center}
\textbf{ \Large Supplementary Material}
\end{center}
\tableofcontents

%====================================================================================

\section{The geometry of the state space}
\label{sec:EPR_map}

%In this subsection, we explore the details of the notion of EPR map $\Lambda$, the measurement outcomes $\M_A$, and the steering outcomes $\Lambda(\M_A)$ introduced in the main text. 
To fix the notation, we consider a state of two qubits $AB$, that is, a 
positive (semi-definite)  unit-trace operator $\varrho$ over $\H_A \otimes
\H_B$, where $\H_A$ and $\H_B$ are $2$-dimensional (2D) Hilbert spaces. The
spaces of hermitian operators acting on $\H_A$ and $\H_B$ are denoted by
$B^H(\H_A)$ and $B^{H}(\H_B)$, respectively, with the identity operators $\II_A$
and $\II_B$. 

Note that $B^H(\H_A)$ is a $4$-dimensional (4D) Euclidean space with the
Hilbert-Schmidt inner product, $\dprod{X}{Y}=\Tr(XY)$ for $X,Y \in B^H(\H_A)$.
If one chooses an orthonormal basis for $\H_A$, and uses the Pauli matrices
$\{\sigma^A_i\}_{i=0}^{3}$ (with $\sigma^A_0= \II_A$) as the basis of
$B^H(\H_A)$, any operator $X$ of $B^H(\H_A)$ can be written as 
\begin{equation}
X=\frac{1}{2}\sum_{i=0}^{3} x_i \sigma_i^A,
\end{equation}
where $x_i = \Tr (X \sigma_i^A)$. We will refer to this basis as the Pauli basis.

One can also use the Pauli basis for $B^H(\H_B)$. With these two coordinate
systems, a density operator $\varrho$ can then be written in terms of the Bloch
tensor, 
\begin{equation}
\varrho= \frac{1}{4}\sum_{i,j=0}^{3} \Theta_{ij}  \sigma_i^A \otimes \sigma_j^B,
\end{equation}
where $\Theta_{ij} =\Tr [\varrho (\sigma_i^A \otimes \sigma_j^B) ]$.
The Bloch tensor is usually written as a matrix
\begin{equation}
\Theta= 
\begin{pmatrix}
1 & \b^T \\
\a & T
\end{pmatrix},
\label{eq:bloch_tensor}
\end{equation}
where $\a$ and $\b$ are Alice's and Bob's Bloch vectors, and $T$  is their
correlation matrix.

%For a subset or an element $X$ of $B^H(\H_A)$, its image under Alice's EPR map is denoted by $X'$, $X'=\Lambda (X)$. 
The map $\Lambda$ from Alice's side, $\Lambda: B^H(\H_A) \to B^H(\H_B)$, is defined by 
\begin{equation}
\Lambda (X)=\Tr_A \left[ \varrho (X \otimes \II_B)\right]
\end{equation}
for $X \in B^H(\H_A)$. 
The Bloch tensor also allows for a direct representation of Alice's map
$\Lambda$ as a $(4 \times 4)$ matrix,
\begin{equation}
\Lambda \equiv \frac{1}{2} \Theta^T.
\end{equation}

We say a state $\varrho$ is \emph{degenerate} if the map $\Lambda$ is
degenerate, i.~e., non invertible. Otherwise it is said to be
\emph{non-degenerate}. We note that degenerate states are zero-measured in the
set of all states. Moreover, they are separable~\cite{saniaemail} (see
Section~\ref{sec:computation_degenerate}). As separable states cannot be used
for steering, we can, without loss of generality, always assume states to be
non-degenerate. We do often make side remarks on how to cope with degenerate
states for completeness.  
%------------------------------------------------------------------
\section{Measurement outcomes and steering outcomes}
\label{sec:outcomes}
The set of Alice's \emph{measurement outcomes} is defined by $\M_A= \{X \in
B^H(\H_A): 0 \le X \le \II_A\}$. Under the map $\Lambda$, Alice's measurement
outcome set is mapped to the set of Alice's \emph{steering outcomes},
$\Lambda(\M_A) \subseteq B^H(\H_B)$. \blue{Note that our considerations start
with a given state $\varrho_{AB}$, so that it is clear that the assemblage
of steering outcomes can be generated by a suitable set of measurements on
Alice's side. In general, any non-signalling assemblage can be realised by
suitable measurements on a suitable state \cite{erwin, nicolas, william}.}

For convenience, we will also consider
Alice's Bloch hyperplane, $\P_A= \{X  \in B^H(\H_A) : \Tr (X)=1\}$, and Alice's
Bloch ball, $\B_A= \P_A \cap \M_A$. The boundary of Alice's Bloch ball is
referred to as Alice's Bloch sphere, denoted by $\S_A$. The same notations with
super/subscripts $B$ apply to Bob's side. 

In the Pauli coordinates, the positive cone is presented as the forward light
cone at the origin. The set of measurement outcomes $\M_A$ is a double cone,
formed by intersecting the positive cone at $0$ with the negative cone at
$\II_A$; see Figure 2 (left) in the main text. The double cone $\M_A$ has two
vertices, $0$ and $\II_A$, and an `equator' of extreme points, which is the
Bloch sphere $\S_A$. 

Note that the steering outcome set $\Lambda(\M_A)$ is simply a linear image of
$\M_A$, thus just a deformed double cone; see Figure 2 (right) in the main text.
The set of steering outcomes has two vertices at $0$  and
$\varrho_B=\Lambda(\II_A)$. It also has an equator which is the image of the
Bloch sphere $\Lambda(\S_A)$. Being a linear image of $\S_A$, this equator is in
fact an ellipsoid  if $\Lambda$ is non-degenerate. 

%-----------------------------------------------------------------------------------
\section{The capacity of an LHS ensemble}
\label{sec:capacity}
An LHS ensemble $\mu$ is a probability measure on Bob's Bloch ball. For a LHS
ensemble, we define its capacity as the set of conditional states that Alice can
simulate,
\begin{equation}
\K(\mu) = \left\{K= \int \d \mu (\sigma) g(\sigma) \sigma :  0 \le g(\sigma) \le 1 \right\}.
\end{equation}
Note that for the case of
two qubits this simplified capacity is sufficient for studying steering with
projective measurement and with positive operator valued measures of $2$ outcomes ($2$-POVM) as well since
they are equivalent. For steering with more general POVMs or steering of
systems in higher dimension, one would need the concept of $n$-capacity of
$\mu$; see Ref.~\cite{chaujpa} for more details.

Now it is clear that a state $\varrho$ is unsteerable with $2$-POVMs 
(hereafter always considered from $A$ to $B$, unless stated otherwise) if any
only if there exists an LHS ensemble $\mu$ such that $\Lambda(\M_A) \subseteq
\K(\mu)$~\cite{chaupra,chauepl,chaujpa}. 
%-----------------------------------------------------------------------------------
\section{The minimal requirement and the principal radius}
\label{sec:capacity}  
Fixing a choice of LHS ensemble $\mu$, we can find an easy criterion for this
nesting problem. Indeed, for $\K(\mu)$ to contain $\Lambda (\M_A)$, it is
sufficient for it to contain all extreme points of $\Lambda (\M_A)$. If we impose the \emph{minimal requirement} for the LHS ensemble
\begin{equation}
\varrho_B= \int \d \mu(\sigma) \sigma,
\label{eq:minimal_requirement}
\end{equation}
then two vertices $0$ and $\varrho_B$ are automatically contained in $\K(\mu)$. 

As described in the main text, it is left to check the inclusion in
$\K(\mu)$ of the equator of the steering outcomes $\Lambda (\S_A)$. Recall that
$\Lambda$ is assumed to be invertible, so instead of working in Bob's space as
described in the main text we can reverse the transformation to work in Alice's
space; see Figure~\ref{fig:capacity}. More precisely, the inclusion of 
$\Lambda(\S_A)$ in $\K(\mu)$ is equivalent to the condition $\S_A \subseteq
\Lambda^{-1} [\K(\mu)]$.
The principal radius $r[\varrho,\mu]$ is then the minimal distance (in the
normal Euclidean metric) from the centre of the Bloch sphere to the boundary of
$\Lambda^{-1} [\K(\mu)]$ constrained to the Bloch hyperplane. Then $\S_A
\subseteq \Lambda^{-1} [\K(\mu)]$ if and only if $r[\varrho,\mu] \ge 1$. 

%Note that for now, the state is assumed to be non-degenerate. This non-degeneracy is in fact not crucial and will be relaxed later. 

\begin{figure}[t]
%\begin{center}
%\begin{minipage}[c]{0.27\textwidth}
%\includegraphics[width=\textwidth]{{capacity_inverse}.png}
%\end{minipage}
%\begin{minipage}[c]{0.06\textwidth}
%\begin{center}
%\includegraphics[width=0.6\textwidth]{Lambda_inversed.png} \\
%\includegraphics[width=0.9\textwidth]{arrow_reversed.pdf}\hspace{0.1cm} \\
%\vspace{0.5cm}
%\end{center}
%\end{minipage}
%\begin{minipage}[c]{0.13\textwidth}
%\includegraphics[width=\textwidth]{{{capacity}.png}}
%\end{minipage}
%\end{center}
\includegraphics[width=0.47\textwidth]{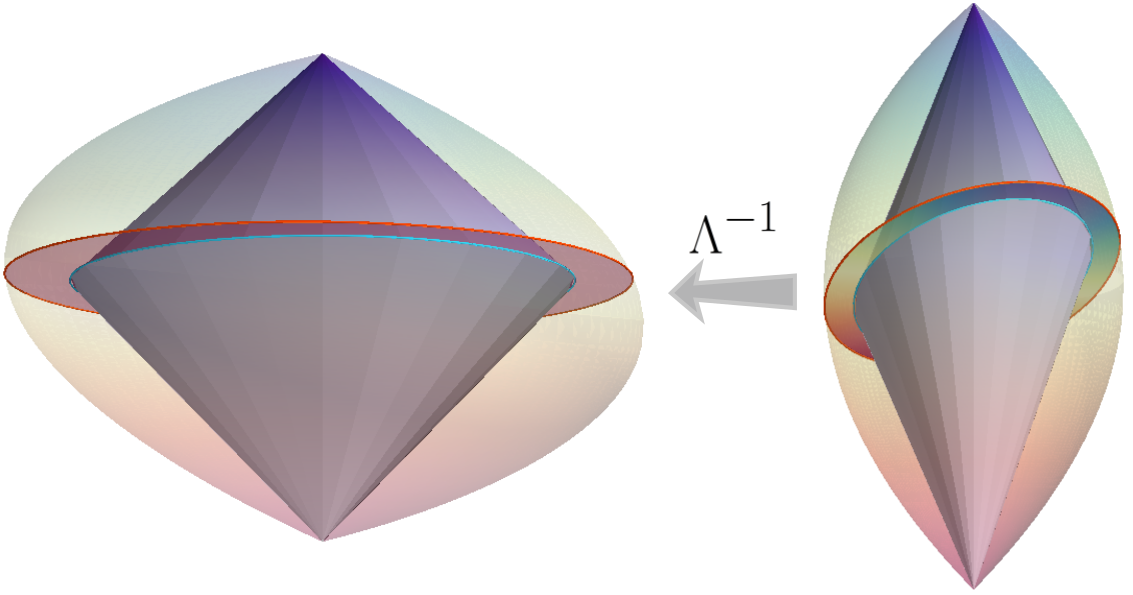}
\caption{Schematic representation of a capacity $\K(\mu)$ that contains the set
of steering outcomes in Bob's space (left) and their images in Alice's
space (right) via the action of $\Lambda^{-1}$.}
\label{fig:capacity}
\end{figure}

%====================================================================================
\section{A simple formula for the critical radius of two qubits}
\label{sec:geometrical_critical_radius}
In this section, with geometrical description of the principal radius above as the starting point, we give a proof for the formula equation (7) in the main text for the principal radius. 

\begin{theorem}
For a given (non-degenerate) state $\varrho$ and for a given LHS ensemble $\mu$ satisfying the minimal requirement, $\varrho_B = \int \d \mu(\sigma) \sigma$, the principal radius is given by
\begin{equation}
r[\varrho,\mu]= \inf_{C} \frac{\int \d \mu (\sigma)  |\Tr(C\sigma)|}{\sqrt{2}\norm{\Tr_B[\bar{\varrho} (\II_A \otimes C)]}}
\label{eq:simple_r}
\end{equation}
where $\bar{\varrho}= \varrho - \frac{\II_A}{2}  \otimes \varrho_B$ and the minimisation is taken over all operators $C$ on Bob's space.
\label{th:simple_r}
\end{theorem}
We refer to the function under the infimum~\eqref{eq:simple_r} as the \emph{fraction function} (inspired by the gap function in Ref.~\cite{chaujpa}),
\begin{equation}
F[\varrho,\mu,C]=\frac{\int \d \mu (\sigma)  |\dprod{C}{\sigma}|}{\sqrt{2}\norm{\Tr_B[\bar{\varrho} (\II_A \otimes C)]}},
\label{eq:fraction_function}
\end{equation}
where we also use the Hilbert-Schmidt product notation $\dprod{C}{\sigma}= \Tr(C \sigma)$.
The fraction function is defined with the \emph{denominator-dominated} convention, namely it is $+\infty$ whenever the denominator vanishes, regardless of the numerator. 
Using the Pauli basis defined in equation~\eqref{eq:bloch_tensor}, Section~\ref{sec:EPR_map}, we represent operators by $4$-vectors, $C \equiv \begin{pmatrix}c_0 \\ \c\end{pmatrix}$, $\sigma \equiv \begin{pmatrix} 1 \\ {\vv} \end{pmatrix}$, and the bipartite state $\rho$ by its Pauli tensor, $\rho \equiv \begin{pmatrix}1 & \b^T \\ \a & T \end{pmatrix}$. In these coordinates, the fraction function can be written expressively, 
\begin{equation}
F[\varrho,\mu,C]= \frac{\int \d \mu ({\vv})  |c_0 + \c^T {\vv}|}{\norm{c_0 \a +  T \c}},
\label{eq:fraction_function_coordinate}
\end{equation}
where $\vv$ runs over vectors in Bob's Bloch ball. 
Being explicit, this formula is very convenient for direct computation. We will refer to both definitions~\eqref{eq:fraction_function} and~\eqref{eq:fraction_function_coordinate} interchangeably.

\begin{proof}
%As we discussed above, by means of the Pauli basic matrices, we can identify an operator with the $4$D vector space. Note that the Pauli basis is orthonomal, but not normalised, $\Tr (\sigma_i \sigma_j)= 2 \delta_{ij}$, thus there will be some factor of $\frac{1}{2}$ when converting the inner product to the canonical one in the $4$D coordinates, $\dprod{X}{Y}=\frac{1}{2} \sum_{i=0}^{3} x_i y_i$.

To derive the formula~\eqref{eq:simple_r}, we proceed as follows. As $\K(\mu)$ is a compact convex object in the $4$D space of Bob's operators, we can define it by a set of linear inequalities, which are easy to determine. Transforming it back to Alice's operator space, we obtain a set of inequalities that define $\Lambda^{-1} [\K(\mu)]$. Constraining this set of inequalities to the Bloch hyperplane $x_0=1$, we obtain a set of inequalities that define the cross-section of  $\Lambda^{-1} [\K(\mu)]$ at $x_0=1$. Note that each inequality in this set corresponds to a $3$D half-space, and the principal radius as described in Section~\ref{sec:capacity} is simply the minimal distance from the corresponding separating $2$D planes to the origin.

We start with finding the set of inequalities that define $\K(\mu)$. These inequalities can be found rather easily~\cite{chaupra,chaujpa}. Let $Y$ be a point in the set $\K(\mu)$, then for any operator $C$
\begin{equation}
\dprod{C}{Y} \le \max_{K \in \K(\mu)} \dprod{C}{K}.
\label{eq:bob_linequalities_symbolic}
\end{equation}
The left-hand side can be solved rather easily,
\begin{align}
\max_{K \in \K(u)} \dprod{C}{K} &= \max_{0 \le g(\sigma) \le 1} \int \d \mu (\sigma) g(\sigma) \dprod{C}{\sigma} \nonumber \\
&= \int \d \mu (\sigma) \max \{ \dprod{C}{\sigma},0\}.
\label{eq:bob_linequalities}
\end{align}
This should be viewed as a family in inequalities parametrised by $C$ that defines $\K(\mu)$.

The inequalities that define $\Lambda^{-1} [\K(\mu)]$  can be found by replacing the operator $Y \in \K(\mu)$ by $Y=\Lambda (X)$ for $X \in \Lambda^{-1} [\K(\mu)]$ in~\eqref{eq:bob_linequalities_symbolic}. Using the explicit coordinates, $C \equiv \begin{pmatrix}c_0 \\ \c\end{pmatrix}$, $\sigma \equiv \begin{pmatrix} 1 \\ {\vv} \end{pmatrix}$, $X\equiv \begin{pmatrix} x_0 \\ \x \end{pmatrix}$, $\Lambda \equiv \frac{1}{2}\begin{pmatrix}1 & \a^{T} \\ \b & T^{T} \end{pmatrix}$, these inequalities can be written as
\begin{equation}
\frac{1}{2}
\begin{pmatrix}c_0 & \c^T\end{pmatrix}
\begin{pmatrix}1 & \a^{T} \\ \b & T^{T} \end{pmatrix}
\begin{pmatrix}x_0 \\ \x \end{pmatrix} 
\le
\int \d \mu ({\vv})  \max \{c_0 + \c^T {\vv},0\}.
\end{equation}
More explicitly, we have 
\begin{equation}
(c_0 \a^T + \c^T T^T) \x \le 2 \int \d \mu ({\vv})  \max \{c_0 + \c^T {\vv},0\} - x_0 (c_0 + \c^T \b).
\label{eq:alice_inequality}
\end{equation}
This should be viewed as a family of inequalities parametrised by $(c_0,\c)$ that defines $\Lambda^{-1} [\K(\mu)]$ consisting of points $X\equiv \begin{pmatrix} x_0 \\ \x \end{pmatrix}$. 

To check if $\Lambda^{-1} [\K(\mu)]$ contains $\M_A$, we only need to check the condition at the equator $\S_A$ (since $\mu$ satisfies the minimal requirement). Since $\S_A$ belongs to the Bloch hyperplane $\P_A$, we can fix $x_0=1$ and~\eqref{eq:alice_inequality} becomes
\begin{equation}
(c_0 \a^T + \c^T T^T) \x \le \int \d \mu ({\vv})  \abs{c_0 + \c^T {\vv}},
\label{eq:alice_linequalities}
\end{equation}
where we have also used the minimal requirement $\b= \int \d \mu (\vv) \vv$ to simplify the right-hand side. 
Then~\eqref{eq:alice_linequalities} is a family of $3$D half-spaces with normal vectors $(c_0 \a^T + \c^T T^T)$ and offsets $\int \d \mu ({\vv})  \abs{c_0 + \c^T {\vv}}$. The distance of each of the separating planes to the origin is
\begin{equation}
\frac{\int \d \mu ({\vv}) \abs{c_0 + \c^T {\vv}}}{\norm{c_0 \a + T \c}}.
\end{equation}
By definition, 
\begin{equation}
r[\varrho,\mu]= \inf_{c_0,\c} \frac{\int \d \mu ({\vv})  |c_0 + \c^T {\vv}|}{\norm{c_0 \a +  T \c}}.
\end{equation}
This is precisely the formula for the principal radius in the coordinate form equation~\eqref{eq:fraction_function_coordinate}. 
\end{proof}
%====================================================================================
\section{Relaxing the non-degeneracy condition of the state}
Strictly, the definition for the principal radius applies only to non-degenerate states.  
One however can take equation~\eqref{eq:simple_r} as the primary definition of the principal radius, which works also for degenerate states. The above proof can be easily adapted to show that a state is unsteerable with a specific choice of LHS ensemble $\mu$ if and only if $r[\varrho,\mu] \ge 1$ with the principal radius as defined by equation~\eqref{eq:simple_r}. %This generalisation with the proof will be obtained explicitly as a by-product of the generalisation of the critical radius to systems of arbitrary dimensions and POVMs of arbitrary number of outcomes in Section~\ref{sec:generalisation}.  
%====================================================================================
\section{Defining domain of the principal radius}
\label{sec:defining_domain}
From the formula~\eqref{eq:simple_r}, one can easily see that the principal radius is well-defined even when $\varrho$ is not a proper state. In the following, when referring to a \emph{state}, we do not impose positivity on it. When imposing positivity on a state, we refer to it as a \emph{proper state}. 

For the principal radius to be well-defined, it is prerequisite that $\varrho_B$ is inside Bob's Bloch ball. This is to guarantee that the minimal requirement does not result in an empty-set of probability measures. It is easy to see that the set of states that have Bob's reduced states inside Bob's Bloch ball is convex and closed. This set is the (most general) defining domain we consider.

%====================================================================================
\section{Concavity of the principal radius}
\begin{proposition}
The principal radius $r[\varrho,\mu]$ is concave in $\mu$.
\label{th:principal_concavity}
\end{proposition}
\begin{proof}
Since $r[\varrho,\mu]$ is an infimum of a family of linear, thus concave, functions in $\mu$, $r[\varrho,\mu]$ must be itself concave in $\mu$. 
\end{proof}
Although not mandatory in the following, it is worth noting that the convexity of $r^{-1}[\varrho,\mu]$ is somewhat better behaved.
\begin{proposition}
The inverse principal radius $r^{-1}[\varrho,\mu]$ is convex either in $\mu$ or $\varrho$, if $\varrho$ is constrained by $\Tr_A[\varrho]= \varrho_B$.
\label{th:inverse_principal_concavity}
\end{proposition}
\begin{proof}
The convexity in $\varrho$ is limited to decompositions which respect the (affine) constraint $\Tr_A[\varrho]=\varrho_B$ (so that $\mu$ satisfies the minimal requirement for all states under consideration). 
We write
\begin{equation}
r^{-1}[\varrho,\mu]=\sup_C \frac{\sqrt{2}\norm{\bar{\varrho} (\II_A \otimes C)}}{\int \d \mu (P) \abs{\dprod{C}{P}}}.
\end{equation}
Now the function under the supremum is convex either in $\mu$ or $\varrho$. Therefore $r^{-1}[\varrho,\mu]$ is convex either in $\mu$ or $\varrho$.
\end{proof}
%====================================================================================
\section{Upper-semicontinuity of the principal radius}
To study in detail the topological properties of the principal radius, we need a weaker notion of continuity, namely semicontinuity.

Recall that $\bar{\RR}=\RR \cup \{- \infty\} \cup \{+\infty\}$. Consider a sequence $\{u_n\}_{n=1}^{+\infty}$ in $\bar{\RR}$. The limit of a subsequence of $\{u_n\}$ is called an accumulation point. The set of accumulation points is closed; its maximum is called the limit superior of $\{u_n\}$, denoted by $\overline{\lim}_{n \to \infty } {u_n}$, and the minimum is called the limit inferior of $\{u_n\}$, denoted by $\underline{\lim}_{n \to \infty } {u_n}$. 

Below we assume that $X$ is a metric space. 

A function $f:X \to \bar{\RR}$ is said to be upper-semicontinuous at $x \in X$ if for any sequence $\{x_n\} \to x$, one has $\overline{\lim}_{n \to \infty} f(x_n) \le f(x)$. An upper-semicontinuous function on a compact metric space attains its maximum. 

Similarly, a function $f:X \to \bar{\RR}$ is said to be lower-semicontinuous at $x \in X$ if for any sequence $\{x_n\} \to x$, one has $\underline{\lim}_{n \to \infty} f(x_n) \ge f(x)$. A lower-semicontinuous function on a compact metric space attains its minimum. 

A function $f:X \to \bar{\RR}$ is continuous at $x \in X$ if and only if it is both upper-semicontinuous and lower-semicontinuous at $x$. 

To study the upper-semicontinuity of $r[\varrho,\mu]$, we need the following lemma. 
\begin{lemma}
Consider $f:X \times Y \to \bar{\RR}$ where $X$ is a metric space and $Y$ is an arbitrary set. Define $g:X \to \bar{\RR}$ by $g(x)= \inf_{y \in Y} f(x,y)$. Suppose all the functions $f(\cdot,y): X \to \bar{\RR}$ with $y\in Y$ are upper-semicontinuous at a certain $x$, then $g$ is upper-semicontinuous at $x$.
\label{lem:upper_semicontinuity}
\end{lemma}
\begin{proof}
We would like to show that for any converging sequence $\{x_n \} \to x$, we have 
\begin{equation} 
\overline{\lim}_{n \to \infty} g(x_n) \le g(x).
\label{eq:to_be_proved}
\end{equation}

Because $g(x)=\inf_{y \in Y} f(x,y) $, there exists $y_0 \in Y$ such that
\begin{equation}
 f(x,y_0) \le g(x).
\end{equation}
Now because for all $x_n$,
\begin{equation}
g(x_n)= \inf_{y \in Y} f(x_n,y) \le f(x_n,y_0),
\end{equation}
we have that
\begin{equation}
\overline{\lim}_{n \to \infty} g(x_n) \le \overline{\lim}_{n \to \infty} f(x_n,y_0).
\end{equation}
And because $f(\cdot,y_0)$ is upper-semicontinuous at $x$, 
\begin{equation}
\overline{\lim}_{n \to \infty} f(x_n,y_0) \le f(x,y_0) \le g(x).
\end{equation}
So we indeed have~\eqref{eq:to_be_proved}.
\end{proof}
The space of probabilistic Borel measures over the Bloch sphere is metrizable and weakly compact; see, e.g., Ref.~\cite[Theorem 6.3.5, 7.2.2, 8.3.2, 8.9.3, 8.9.4]{Bogachev2007a}. Then its intersection with the minimal constraint (which is weakly closed) is also metrizable and weakly compact. From Theorem~\ref{th:simple_r} and Lemma~\ref{lem:upper_semicontinuity}, we can then easily prove the following  proposition.
\begin{proposition}
The principal radius $r[\varrho,\mu]$ is weakly upper-semicontinuous in $\mu$.
\label{th:principal_upper_semiconitinuity}
\end{proposition}
\begin{proof}
It is easy to check that for fixed $C$, the fraction function $F[\varrho,\mu,C]$ in equation~\eqref{eq:fraction_function} is upper-semicontinuous in $\mu$ with respect to the weak topology. To be more precise, for all $C$ such that $\norm{\bar{\varrho} (\II_A \otimes C)} \ne 0$, $F[\varrho,\mu,C]$ is continuous in $\mu$. For $\norm{\bar{\varrho} (\II_A \otimes C)} = 0$, $F[\varrho,\mu,C]$ is $+\infty$, thus trivially upper-semicontinuous. Therefore $r[\varrho,\mu]$ is weakly upper-semicontinuous as a consequence of Lemma~\ref{lem:upper_semicontinuity}.
\end{proof}
In contrast to upper-semicontinuity, the lower-semicontinuity of the principal radius is rather subtle. We postpone this study until we have discussed the canonical form of a state; see Section~\ref{sec:principal_lowersemicontiunuity}.
%==================================================================================== 
\section{Existence of an optimal LHS ensemble}
As the principal radius $r[\varrho,\mu]$ is upper-semicontinuous over the compact space of probabilistic Borel measures $\mu$ satisfying the minimal requirement, it attains its maximum.
The critical radius of $\varrho$ is then defined by
\begin{equation}
R[\varrho]= \max_\mu r[\varrho,\mu],
\label{eq:critical_radius}
\end{equation}
where the maximum is taken over probabilistic Borel measures satisfying the minimal requirement~\eqref{eq:minimal_requirement}.   
Physically, we have proved the following statement:

\begin{theorem}[Existence of optimal LHS ensemble]
For any two-qubit state $\varrho$, there exists an optimal LHS ensemble for steering given by $\mu^{\ast}=\arg \max r[\varrho,\mu]$.
\end{theorem}

Note that this concept of optimal LHS ensemble is similar (but not identical) to that of optimal LHS \emph{model} defined in the original paper by Wiseman {\it et al.}~\cite{wiseman1}. There, an optimal LHS model consists of an LHS ensemble \emph{and} a choice of response functions which is also optimal in a certain sense. It is still unknown whether or not one can construct optimal response functions. Here we prove that an optimal choice of LHS ensemble does exist. The existence of an optimal LHS ensemble changes the perspective on the problem of determining the steerability of a state. Now, instead of checking every LHS ensemble, we search for a specific LHS ensemble. Moreover, instead of checking all possible choices of response functions, the single value of the critical radius is enough to tell about the steerability of the state. All is then about how to compute the critical radius. We discuss the practical computation of the critical radius in Section~\ref{sec:computation}. 

%The concept of the critical radius is useful beyond than checking the steerability a single state. It characterises quantum steering in a more global level. In particular, the scaling and invariance of the critical radius implies structure of 
%==================================================================================== 
\section{Implication of symmetry on the optimal LHS ensemble}
We say a state $\varrho$ is $(\G,U,V)$-symmetric with a compact group $\G$ with its two actions $U$ on $\H_A$ and $V$ on $\H_B$ if $U(g) \otimes V(g) \varrho U^{\dagger} (g) \otimes V^{\dagger}(g)$ for all $g \in \G$. Recall that the action $V$ on $\H_B$ induces an action on the measures on $\B_B$, defined by $R_V(g)[\mu](X)= \mu[V(g) X V^{\dagger}(g)]$ for all measurable subsets $X$ of $\B_B$. 
\begin{theorem}[Symmetry of LHS ensemble]
If $\varrho$ is $(\G,U,V)$-symmetric, then there exists an optimal ensemble $\mu^{\ast}$ which is $(\G,V)$-invariant, $R_V(g) [\mu^{\ast}] = \mu^{\ast}$.
\label{th:symmetry_LHS}
\end{theorem}
\begin{proof}
This theorem is a simple consequence of the concavity of $r[\varrho,\mu]$ in $\mu$. We will only sketch the proof. 
Let $\mu^{\ast}$ be an optimal LHS ensemble for $\varrho$, namely, $R[\varrho] = r[\varrho,\mu^\ast]$.
From the formula of the principal radius, and with the symmetry of $\varrho$, one can easily verify that also $r[\varrho,R_V(g)[\mu^\ast]]=R[\varrho]$. 
Define a measure $\bar{\mu}^\ast$ on $\B_B$ by $\bar{\mu}^{\ast}= \int \d \omega (g) R_V(g) [\mu^\ast]$, where $\omega$ is the Haar measure of $\G$. It is easy to see that $\mu^{\ast}$ is invariant under $\G$. We show that it is an optimal LHS ensemble. Due to the concavity of $r[\varrho,\mu]$ in $\mu$, we have $r[\varrho,\bar{\mu}^\ast] \ge \int \d \omega (g) r[\varrho,R_V(g)[\mu^\ast]]=R[\varrho]$. On the other hand, by the definition of the critical radius, $r[\varrho,\bar{\mu}^\ast] \le R[\varrho]$. We therefore have $r[\varrho,\bar{\mu}^\ast] = R[\varrho]$, or $\bar{\mu}^\ast$ is an optimal LHS ensemble.
\end{proof}
One may observe from the above proof that the symmetry of LHS ensembles is determined only by the action $V$ (and not $U$). In fact, the notion of the $(\G,U,V)$-symmetric state seems a bit stronger than necessary. This is indeed the case. In fact, the theorem can be formulated as: when the set of steering outcomes $\Lambda (\M_A)$ is $(\G,V)$-symmetric, then LHS ensembles can be assumed to be $(\G,V)$-symmetric.

%====================================================================================
\section{Scaling of the critical radius}
\begin{theorem}[Scaling of the critical radius]
For any state $\varrho$ and any $\lambda \ge 0$, we have
\begin{equation}
R[\varrho]= \lambda R[\lambda \varrho + (1- \lambda) \frac{\I_A}{2} \otimes \varrho_B].
\label{eq:scaling}
\end{equation}
Note that the theorem applies as well if $\lambda \varrho + (1- \lambda) \frac{\I_A}{2} \otimes \varrho_B$ is not a proper state.
\label{th:scaling}
\end{theorem}
\begin{proof}
The proof is trivial given formula~\eqref{eq:simple_r}. One simply inserts $\lambda \varrho + (1- \lambda) \frac{\I_A}{2} \otimes \varrho_B$ to find that $r[\lambda \varrho + (1- \lambda) \frac{\I_A}{2} \otimes \varrho_B,\mu]=\frac{1}{\lambda} r[\varrho,\mu]$, which implies equation~\eqref{eq:scaling}.

%Suppose the Pauli tensor for $\varrho$ is $\Theta(\varrho)=\begin{pmatrix}1 & \b^{T} \\ \a & T\end{pmatrix}$, then the Pauli tensor for $\lambda \varrho + (1- \lambda) \frac{\I}{2} \otimes \varrho_B$  is given by $\Theta[\lambda \varrho + (1- \lambda) \frac{\I}{2} \otimes \varrho_B]=\begin{pmatrix}1 & \b^{T} \\ \lambda \a  & \lambda T \end{pmatrix}$. Then from the formula for $r_u(\varrho)$~\eqref{eq:simple_r}, one finds $r[\varrho,\mu]= \lambda r[\lambda \varrho + (1- \lambda) \frac{\I}{2} \otimes \varrho_B,\mu]$. Therefore $R[\varrho]= R[\lambda \varrho + (1- \lambda) \frac{\I}{2} \otimes \varrho_B]$.
\end{proof}

Note that by setting $\lambda=0$, we find $R[\frac{\I_A}{2} \otimes \varrho_B]=+\infty$. Thus the critical radius can be infinite; we will see below that it is infinite only at states of this form. Geometrically, along the scaling line, the equator of the steering outcomes $\Lambda(\S_A)$ is  uniformly rescaled by the factor $\lambda$.  At $\lambda=0$, the equator degenerates to a single point.

%==================================================================================== 
\section{Continuous symmetry of the critical radius}

The Bloch hyperplane for two qubits, denoted by $\P$, is the linear manifold of hermitian trace-$1$ operators acting on $\H_A \otimes \H_B$. For $U \in \U(2) $, $V \in \GL(2)$, consider the affine transformation from the Bloch hyperplane of the joint system into itself $\varphi_{(U,V)}:\P \to \P$, defined by 
\begin{equation}
\varphi_{(U,V)} (X) = \frac{(U \otimes V) X (U^\dagger \otimes V^\dagger)}{\Tr [(U \otimes V) X (U^\dagger \otimes V^\dagger)]}.
\end{equation}
for $X \in \P$. Note that this is a group action of $\U(2) \times \GL(2)$ on $\P$. Moreover $\varphi_{(U,V)}$ conserves the positivity, thus also maps the set of (bipartite) proper states into itself. 

Accepting a bit of ambiguity in notation for the sake of simplicity, for $V \in \GL(2)$, we also denote $\varphi_{V}: \B_B \to \B_B$ defined by
\begin{equation}
\varphi_V (\sigma)= \frac{V \sigma V^{\dagger}}{\Tr (V \sigma V^{\dagger})}.
\end{equation} 
\begin{lemma}
Consider a given state $\varrho$ and a given probability measure (LHS ensemble) $\mu$ satisfying the minimal requirement $\int \d \mu (\sigma) \sigma = \varrho_B$. For $U \in \U(2)$ and $V \in \operatorname{GL} (2)$, we denote $\tilde{\varrho}=\varphi_{(U,V)} (\varrho)$.
Note that there exists a unique probability measure $\tilde{\mu}$ on $\B_B$ defined by
\begin{equation}
\int \d \tilde{\mu} (\sigma ) f(\sigma ) = \frac{1}{\Tr (V \varrho_B V^{\dagger})} \int  \frac{\d \mu \circ \varphi_{V^{-1}} ( \sigma ) }{ \Tr [V^{-1} \sigma (V^{-1})^{\dagger}]} f( \sigma )
\end{equation}
for all continuous functions $f$. Then $\tilde{\mu}$ satisfies the minimal requirement for $\tilde{\varrho}$ and $r[\varrho,\mu] = r[\tilde{\varrho},\tilde{\mu}]$.
\label{lem:principal_invariance}
\end{lemma}

\begin{proof}
(i) To prove that  $\tilde{\mu}$ satisfies the minimal requirement for $\tilde{\varrho}$, we need to show that
\begin{equation}
\frac{V \varrho_B V^{\dagger}}{\Tr (V \varrho_B V^{\dagger})} = \int \d \tilde{\mu} ( \sigma ) \sigma.
\label{eq:to_prove}
\end{equation}
Using the definition of $\tilde{\mu}$, we have
\begin{equation}
\int \d \tilde{\mu} (\sigma) \sigma = \frac{1}{\Tr (V \varrho_B V^{\dagger})} \int \frac{\d \mu \circ \varphi_{V^{-1}} ( \sigma )}{\Tr [V^{-1} \sigma (V^{-1})^{\dagger}]} \sigma. 
\label{eq:def_of_mu_tilde}
\end{equation}
By changing the variable of integral, $\sigma = \varphi_V ( \tau )$,
\begin{align}
\int \frac{\d \mu \circ \varphi_{V^{-1}} ( \sigma )}{\Tr [V^{-1} 
\sigma (V^{-1})^{\dagger}]} \sigma  &= 
 \int \d \mu  (\tau) \frac{\varphi_V(\tau) }{\Tr (V^{-1} \varphi_V(\tau) (V^{-1})^{\dagger})}  \nonumber \\
&= \int \d \mu  (\tau) V \tau V^{\dagger} \nonumber \\
&= V\varrho_BV^{\dagger},
\label{eq:change_variable}
\end{align}
where we have used the minimal requirement for $\mu$, $\int \d \mu (\tau) \tau = \varrho_B$. From~\eqref{eq:change_variable} and~\eqref{eq:def_of_mu_tilde}, we obtain~\eqref{eq:to_prove}.

(ii) Now we prove that $r[\varrho,\mu] = r[\tilde{\varrho},\tilde{\mu}]$. Using the definition~\eqref{eq:simple_r}, we have $r[\tilde{\varrho},\tilde{\mu}]$ as
\begin{equation}
\inf_{C} 
\frac{1}{\sqrt{2}\norm{\Tr_B[\bar{\tilde{\varrho}} (\II_A \otimes C)]} \Tr (V \varrho_B V^{\dagger}) } \int  \frac{ \d \mu \circ \varphi^{-1}_V (\sigma) |\dprod{C}{\sigma}|}{ \Tr (V^{-1} \sigma (V^{-1})^{\dagger}) },
\end{equation}
where $\bar{\tilde{\varrho}}= \tilde{\varrho}-\frac{\II_A}{2} \otimes \varrho_B$.
In the numerator, we make a change of the integration variable $\sigma=\varphi_V(\tau)$,
\begin{align}
\int  \frac{ \d \mu \circ \varphi^{-1}_V (\sigma) |\dprod{C}{\sigma}|}{ \Tr (V \sigma V^{\dagger}) } &= \int \frac{ \d \mu (\tau) \abs{\dprod{C}{\varphi_V(\tau)}}}{\Tr (V^{-1} \varphi_V(\tau) (V^{-1})^{\dagger})} \nonumber \\
&= \int \mu (\tau) \abs{\dprod{V^\dagger CV}{\tau}}. 
\end{align}

In the denominator, we have
\begin{align}
\Tr_B[\bar{\tilde{\varrho}} (\II_A \otimes C) ] &= \frac{ \Tr_B[(U \otimes V) \bar{\varrho} (U^{\dagger} \otimes V^{\dagger}) (\II_A \otimes C)] }{\Tr (V \varrho_B V^{\dagger})} \nonumber \\
&= \frac{\Tr_B [\bar{\varrho} (\II_A \otimes V^\dagger C V) ]}{\Tr (V \varrho_B V^{\dagger})}. 
\end{align}

So we have
\begin{align}
r[\tilde{\varrho},\tilde{\mu}] &= \inf_C \frac{\int \d \mu (\sigma) \abs{\dprod{V^\dagger CV}{\sigma}}}{ \sqrt{2}\norm{\Tr_B [\bar{\varrho} (\II_A \otimes V^\dagger C V) ]]}} \nonumber \\
&= \inf_C \frac{\int \d \mu (\sigma) \abs{\dprod{C}{\sigma}}}{ \sqrt{2}\norm{\Tr_B[\bar{\varrho} (\II_A \otimes C)]}}, 
\end{align}
where we have used the fact that $C \mapsto VCV^{\dagger}$ is bijective.
The last expression then coincides with $r[\varrho,\mu]$.
\end{proof}

The invariance of the principal radius also has a simple geometrical interpretation. Under the local unitary transformation $U$ on Alice's side, the set of steering outcomes $\Lambda(\M_A)$ is invariant. On the other hand, under the (so-called) local filtering $V$ on Bob's side, $\Lambda(\M_A)$ and $\K(\mu)$ transform covariantly; depending only on the relative geometry of $\Lambda(\M_A)$ and $\K(\mu)$, the principal radius is invariant.  

\begin{theorem}[Continuous symmetry of the critical radius]
For any state $\varrho$ and $U \in \operatorname{U} (2)$, $V \in \operatorname{GL} (2)$, we have $R[\varrho] = R [\varphi_{(U,V)} \varrho]$. 
\label{th:invariant}
\end{theorem}
\begin{proof}
Let $\mu^\ast$ be an optimal LHS ensemble for $\varrho$, then $R [\varrho]= r [\varrho,\mu^\ast]$. Let $\tilde{\mu^{\ast}}$ be defined as in Lemma~\ref{lem:principal_invariance}, then $r [\varrho,\mu^\ast] = r [\tilde{\varrho},\tilde{\mu^\ast}] \le R [\tilde{\varrho}]$, thus $R[\varrho] \le R[\tilde{\varrho}]$. 
Applying the Lemma for the reversed transformation from $\tilde{\varrho}$ to $\varrho$, we find $R[\tilde{\varrho}] \le R[\varrho]$. 
It then follows that $R[\varrho]= R[\tilde{\varrho}]$.
\end{proof}

%It has the following interesting consequence:
%\begin{corollary}
%Suppose an optimal LHS ensemble for a state $\varrho$ is generated by a distribution $u$ (with respect to the Haar measure), then the state $\varphi_{(U,V)} (\varrho)$ admits the measure generated by the following distribution as an optimal LHS ensemble:
%\begin{equation}
%\hat{\varphi}_{V} [u] (P)= \frac{1}{2 \Tr (V \varrho_B V^\dagger)} \frac{\det [V^{-1}(V^{-1})^{\dagger}]}{\Tr^3 [V^{-1} P  (V^{-1})^{\dagger}]} u [ \varphi_V^{-1} (P)].
%\end{equation}
%\end{corollary}
%\begin{proof}
%\blue{
%Note that the transformation transforms the Boch sphere into itself (projective Lorentz transformations). Therefore it transforms the distribution $u$ to a new one $\tilde{u}$. In order to get a formula for $\tilde{u}$, the main problem is thus just to compute the Jacobian of the automorphism of the Bloch sphere. We have an old trick to do this (I will write the details later; it is quite fun). This is precisely the factor that appears in front of $u$ in the expression for $\tilde{u}$. In fact my old formula was derived only for the case where $V$ is hermitian (which is the case we are mostly interested in); the above one an extrapolation to the case non-hermitian $V$. Nevertheless, I believe it is not so difficult to extend the computation.
%}
%\end{proof}
%==================================================================================== 
\section{Time-reversal symmetry of the critical radius}
Consider Alice's Bloch hyperplane $\P_A$ and fix the Pauli basis. The time-reversal transformation on Alice's Bloch hyperplane is the transformation $T_A: \P_A \to \P_A$, $X \mapsto T_A(X)= X^{\ast}$, where $X^{\ast}$ is the complex conjugation of $X$. Geometrically, $T_A$ is the reflection along $\sigma_y$. Therefore $T_A$ maps Alice's Bloch ball to itself. Upto a unitary transformation, $T_A$ is also equivalent to the inversion of $\P_A$ through $\frac{\II_A}{2}$. \new{In fact, we will not distinguish different implementations of the time-reversal transformation which are equivalent upto some unitary transformations.} 

On a bipartite state $\varrho$, $T_A$ is extended to partial time-reversal transformation $T_A \otimes \mathcal{I}_B$, where $\mathcal{I}_B$ is the identity map on Bob's space. The same notation is applied to the time-reversal transformation on Bob's side. Upto local unitary transformations, the partial time-reversal transformation is equivalent to the partial transposition. Note that on the bipartite Bloch hyperplane, $T_A \otimes \mathcal{I}_B$ does not map the set of \emph{proper} states into itself. In fact, the subset of proper states that is invariant under $T_A \otimes \mathcal{I}_B$ are separable states---by the Peres--Horodecki criterion of the partial transposition~\cite{peresppt,horodeckippt}. Somehow unexpectedly, for steerability, the following theorem tells that the critical radius $R$ is invariant under the partial time-reversal transformations. 
\begin{theorem}[Time-reversal symmetry of the critical radius]
For any state $\varrho$, we have $R[\varrho]= R[(T_A \otimes \mathcal{I}_B) \varrho]=R[(\mathcal{I}_A \otimes T_B) \varrho]$.
\label{th:time_reversal}
\end{theorem}
While quantum steering is asymmetric between two parties, this theorem has a rather symmetric form between the time-reversals on either of the parties. The proof, however, seems to suggest that this symmetry is perhaps rather accidental. 
\begin{proof}
(i) We start with proving $R[\varrho]= R[(T_A \otimes \mathcal{I}_B) \varrho]$. In fact we can show the \emph{invariance} of the principal radius, $r[\varrho,\mu]= r[(T_A \otimes \mathcal{I}_B) \varrho, \mu]$. This is easily seen because the numerator of~\eqref{eq:simple_r} is invariant under the transformation, the denominator is also invariant since the time-reversal is isometric. 

(ii) The proof that $R[\varrho]= R[(\mathcal{I}_A \otimes T_B) \varrho]$ is only slightly different. It follows from the \emph{covariance} of the principal radius, $r[\varrho,\mu]= r[(\mathcal{I}_A \otimes T_B) \varrho, \mu \circ T_B^{-1}]$.
Clearly the minimal requirement is covariant, namely, 
\begin{equation}
T_B(\varrho_B) = \int \d \mu \circ T_B^{-1} (\sigma) \sigma.
\end{equation} 
Moreover, we have
\begin{equation}
r[(\mathcal{I}_B \otimes T_B) \varrho, \mu \circ T_B^{-1}]= \inf_{C} \frac{\int \d \mu \circ T_B^{-1} (\sigma)  |\dprod{C}{\sigma}|}{\sqrt{2}\norm{\Tr_B[(\mathcal{I}_A \otimes T_B) \bar{\varrho} (\II_A \otimes C)]}},
\end{equation}
where $\bar{\varrho}= \varrho - \frac{\II_A}{2}  \otimes \varrho_B$.
Now we note that $T_B$ is symmetric on the hermitian operators, $\dprod{X}{T_B(Y)}=\dprod{T_B(X)}{Y}$. In the numerator, changing the integration variable and applying the symmetry of $T_B$, we arrive at
\begin{equation}
\int \d \mu \circ T_B^{-1} (\sigma)  |\dprod{C}{\sigma}|=\int \d \mu (\sigma)  |\dprod{T_B(C)}{\sigma}|.
\end{equation}
In the denominator, we have the identity
\begin{equation}
\Tr_B[(\mathcal{I}_A \otimes T_B) \bar{\varrho} (\II_A \otimes C)]= \Tr_B[\bar{\varrho} (\II_A \otimes T_B(C))],
\end{equation}
which can be proved by expanding $\bar{\varrho}$ in product operators and verifying it for every product operator term. 
Collecting both the numerator and the denominator, we then have
\begin{equation}
r[(\mathcal{I}_B \otimes T_B) \varrho, \mu \circ T_B^{-1}]= \inf_{C} \frac{\int \d \mu (\sigma)  |\dprod{T_B(C)}{\sigma}|}{\sqrt{2}\norm{\Tr_B[\bar{\varrho} (\II_A \otimes T_B(C))]}},
\end{equation}
which is expressively the same as $r[\varrho,\mu]$ since $T_B$ is bijective.
\end{proof}

From the above proof, one may find a similar geometrical signature as the continuous symmetry of the critical radius: the local time-reversal transformation on Alice's space leaves $\Lambda (\M_A)$ invariant, while the local time-reversal transformation on Bob's side acts covariantly on $\Lambda (\M_A)$ and $\K (\mu)$. 
%===================================================================================
\section{The canonical form and normal states}
\label{sec:canonical}
A generic state is fully characterised by Alice's and Bob's Bloch vectors $\a$ and $\b$ and their correlation matrix $T$. We therefore sometimes identify $\varrho$ with a triple $(\a,T,\b)$, $\varrho \equiv (\a,T,\b)$.

If Bob's reduced state is pure, the bipartite state is called \emph{abnormal}. In this case, there exists only a single measure that satisfies the minimal requirement, namely the one supported only at Bob's reduced state. 
The critical radius then reads,
\begin{equation}
R[(\a,T,\b)]= \inf_{c_0,\c}\frac{\abs{c_0 + \c^T \b}}{\norm{c_0 \a + T \c}},
\label{eq:critical_radius_non_normal}
\end{equation}
where we have used the Bloch parameters $\a$, $\b$, $T$ to denote the state. To find this infimum, we change the variable $c_0'= c_0+\c^T \b$ and have
\begin{equation}
R[(\a,T,\b)]= \inf_{c_0',\c}\frac{\abs{c_0'}}{\norm{c_0' \a + (T-\a\b^T) \c}}.
\label{eq:critical_radius_non_normal_new}
\end{equation}
One then finds that for abnormal states, we have
\begin{equation}
R[(\a,T,\b)]=
\left\{ 
\begin{array}{ll}
\frac{1}{\norm{\a}} & \mbox{ if $T=\a \b^T$ (product states),} \\
0 & \mbox{ if $T \ne \a \b^T$}.
\end{array}
\right.
\label{eq:critical_non_normal}
\end{equation}
Note that if the abnormal state is a proper state (i.e., positive), it must be a product state and thus unsteerable.

If Bob's reduced state is not pure, the state is said to be \emph{normal}. By the continuous symmetry of the critical radius Theorem~\ref{th:invariant}, a normal state can always be brought into the \emph{canonical form} without changing the critical radius,
\begin{equation}
\Theta = 
\begin{pmatrix}
1 & \pmb{0}^T \\
\a & T
\end{pmatrix},
\label{eq:canonical_theta}
\end{equation}
where $\a$ is Alice's reduced state, and $T$ is the correlation matrix, which can also be assumed to be diagonal $T=\operatorname{diag} (\pmb{s})$. \new{Note that the canonical parameters, that is, its Alice's reduced state and the correlation diagonal in the canonical form, vary continuously as functions of $\varrho$ limited to the set of normal states.} Moreover if a normal states is non-degenerate, its canonical form is also non-degenerate (and vice versa).  

For a state in the canonical form, we also identify the notation $\varrho \equiv (\a,T)$. \new{In fact, the importance of the canonical form to studying quantum steering cannot be over-emphasised. Let us note immediately some interesting properties of the canonical form.} 

First, for canonical states, the invariance of the critical radius under partial time-reversal transformation implies that $R[(\a,T)]=R[(-\a,-T)]=R[(\a,-T)]=R[(-\a,T)]$. 

Second, the minimum requirement $\int \d \mu (\vv) \vv=0$ is independent of the canonical state $\varrho=(\a,T)$. This is in fact a every important technical point, which renders studying of general properties of the critical radius such as its continuity possible at all.

And third, the operator $C$ in the fraction function can be limited to some simple constraints: 
\begin{lemma}
For a two-qubit canonical state $\varrho=(\a,T)$, the critical radius can be found by
\begin{equation}
r[\varrho,\mu]= \inf_{c_0,\c} \frac{\int \d \mu ({\vv})  |c_0 + \c^T {\vv}|}{\norm{c_0 \a +  T \c}}.
\end{equation}
where $c_0$ and $\c$ can be subjected to canonical constraints $\norm{\c}=1$ and $-1 \le c_0 \le +1$.
\label{lem:cononical_constraint}
\end{lemma}
\begin{proof}
Let us recall the fraction function 
\begin{equation}
F[\varrho,\mu,c_0,\c]=\frac{\int \d \mu ({\vv})  |c_0 + \c^T {\vv}|}{\norm{c_0 \a +  T \c}}.
\end{equation}
We first note that we can assume $\c \ne 0$. This is because  $F[\varrho,\mu,c_0,\c=0]=\frac{1}{\norm{\a}} \ge \lim_{c_0 \to \infty} F[\varrho,\mu,c_0,\c \ne 0] \ge \inf_{c_0, \c \ne 0}  F[\varrho,\mu,c_0,\c \ne 0]$.

Now since for all $\lambda > 0$, $F[\varrho,\mu,c_0,\c]=F[\varrho,\mu,\lambda c_0, \lambda \c]$, we can out the constraint $\norm{\c}=1$ by choosing an appropriate $\lambda$.  

We next show that $F[\varrho,\mu,c_0,\c]$ with $\norm{\c}=1$ and $\abs{c_0} \ge 1$ attains the infimum at $\abs{c_0}=1$. To see this, note that for $\norm{\c}=1$ and $\abs{c_0} \ge 1$, we have either $c_0 + \c^T {\vv} \ge 0$ or $c_0 + \c^T {\vv} \le 0$ for all ${\vv}$ in Bob's Bloch ball. Therefore $\int \d \mu ({\vv})  |c_0 + \c^T {\vv}|=\abs{\int \d \mu ({\vv})  (c_0 + \c^T {\vv})}=\abs{c_0}$. Thus, for $\norm{\c}=1$ and $\abs{c_0} \ge 1$, we have
\begin{equation}
F[\varrho,\mu,c_0,\c]=\frac{1}{\norm{\a +  T \frac{\c}{c_0}}},
\end{equation}
which clearly attains the infimum at $c_0=\pm 1$. 

To summarise, we therefore can limit the infimum in computing the principal radius from the fraction function to $\norm{\c}=1$ and $-1 \le c_0 \le 1$. 
\end{proof} 

%We thus have shown that the domain of $c_0$ and $\c$ can be limited to $\norm{\c}=1$ and $-1 \le c_0 \le +1$.
%==================================================================================
\section{Lower-semicontinuity of the principal radius}
\label{sec:principal_lowersemicontiunuity}
For the sake of convenience, we will limit our analysis to non-degenerate states in the canonical form only. 
This is sufficient to decide steerability. 

\begin{lemma}
Consider $f:X \times Y \to \bar{\RR}$ where $X$ is a metric space and $Y$ is a compact metric space. Define $g:X \to \bar{\RR}$ by $g(x)= \inf_{y \in Y} f(x,y)$.
Suppose at a certain $x$, the function $f$ is jointly lower-semicontinuous at $(x,y)$ for all $y \in Y$. Moreover suppose the function $f(x',\cdot): Y \to \bar{\RR}$  attains its infimum over $Y$ for all $x'$ in a neighbourhood $V(x)$ of $x$. Then $g$ is lower-semicontinuous at $x$.
\label{lem:lower_semicontinuity}
\end{lemma}
\begin{proof}
We would like to show that for any converging sequence $\{x_n \} \to x$, we have 
\begin{equation} 
\underline{\lim}_{n \to \infty} g(x_n) \ge g(x).
\label{eq:to_be_proved_lower}
\end{equation}
Without loss of generality, we can assume that $x_n \in V(x)$ for all $n$. 

Letting $\{x_{k}\}$ be a subsequence of $\{x_n\}$ such that $\{ g(x_{k}) \}$ converges to $\underline{\lim}_{n \to \infty} g(x_n)$, we have
\begin{equation}
\lim_{k \to \infty} g(x_{k})= \underline{\lim}_{n \to \infty} g(x_n).
\end{equation}
Now because for every $x_k \in V(x)$, $g(x_k,\cdot)$ attains its infimum, there exists $y_k$ such that
\begin{equation}
g(x_k)= f(x_k,y_k).
\end{equation}
and thus in particular
\begin{equation}
\lim_{k\to \infty} g(x_k)= \lim_{k \to \infty} f(x_k,y_k).
\end{equation}
Because $Y$ is compact, there exists a subsequence $\{y_p\}$ of $\{y_k\}$ that converges to certain point $y_0$ of $Y$.
We also have
\begin{equation}
\lim_{k\to \infty} g(x_k) = \lim_{p \to \infty} f(x_p,y_p).
\end{equation}
Note that $\{(x_p,y_p) \} \to (x,y_0)$ and because $f$ is jointly lower-semicontinuous at $(x,y_0)$ by assumption, we have
\begin{equation}
\lim_{p \to \infty} f(x_p,y_p) \ge f(x,y_0).
\end{equation}
On the other hand, $f(x,y_0) \ge g(x)=\inf_{y \in Y} f(x,y)$, thus
\begin{equation}
\lim_{p \to \infty} f(x_p,y_p) \ge g(x).
\end{equation}
Then equation~\eqref{eq:to_be_proved_lower} follows directly.
\end{proof}

%\begin{lemma}
%Let $Y$ be a compact topological space, then for any other topological space $X$, the projection $\pi_X: X \times Y \to X$ is closed, i.e., it maps closed sets to closed sets. 
%\label{lem:close_projection}
%\end{lemma}
%\begin{proof}
%Let $A$ be a closed set of $X \times Y$, we would like to show that $\pi_X(A)$ is closed. Let $\{x_n\}_{n=0}^{+\infty}$ be a sequence of $\pi_X(A)$ converging to $x$, we would like to show that $x \in A$. Because $x_n \in \pi_X(A)$, there exists $y_n \in Y$ such that $(x_n,y_n) \in A$. Now because $Y$ is compact, the sequence $\{y_n\}_{n=0}^{+\infty}$ contains a subsequence $\{y_{n_k}\}_{k=0}^{+\infty}$ converging to some point of $Y$; called it $y$. Therefore the sequence $\{x_{n_k},y_{n_k}\}_{k=0}^{+\infty}$ converge to $(x,y)$, which is in $A$ because of its closeness. By definition of projection, we have $x \in \pi_X (A)$.
%\end{proof}
%\begin{corollary}
%Let $Y$ be a compact topological space, and $X$ be a topological space. If a function $f: X \times Y \to \bar{\RR}$ is lower-semicontinuous, then the function $g: X \to \bar{\RR}$ defined by $g(x) = \inf_y f(x,y)$ is lower-semicontinuous. 
%\end{corollary}
%\begin{proof}
%To show that $g$ is upper-semicontinuous, we show that the set $\{x \vert g(x) \le t\}$ is closed for all $t$. One notes that $\{x \vert g(x) \le t\} = \pi_X \{(x,y) \vert f(x,y) \le t \}$. According to Lemma~\ref{lem:close_projection}, $\pi_X$ is closed. Moreover  $\{(x,y) \vert f(x,y) \le t \}$ is closed by the lower-semicontinuity of $f$. It then follows that $\{x \vert g(x) \le t\}$ is closed.
%\end{proof}

\begin{corollary}
Consider $f:X \times Y \to \bar{\RR}$ where $X$ is a metric space and $Y$ is a compact metric space. Define $g:X \to \bar{\RR}$ by $g(x)= \inf_{y \in Y} f(x,y)$. Suppose for $x \in X$, the function $f$ is jointly continuous on $V(x) \times Y$ for some neighbourhood $V(x)$ of $x$, then $g$ is continuous at $x$.
\label{th:topo_continuity}
\end{corollary}
\begin{proof}
The upper-semicontinuity of $g$ follows from Lemma~\ref{lem:upper_semicontinuity}. Its lower-semicontinuity follows from Lemma~\ref{lem:lower_semicontinuity}. These two results imply its continuity.
\end{proof}

We are now ready to prove the following important result.
\begin{proposition}
For a non-degenerate state in the canonical form, the principal radius $r[\varrho,\mu]$ is also lower-semicontinuous in $\mu$.
\label{th:principal_lower_semicontinuity}
\end{proposition}
\begin{proof}
 With $\varrho=(\a,T)$ and $C=(c_0,\c)$, we recall the fraction function
\begin{equation}
F[\varrho,\mu,C]=\frac{\int \d \mu ({\vv})  |c_0 + \c^T {\vv}|}{\norm{c_0 \a +  T \c}},
\end{equation}
where $C=(c_0,\c)$ is subject to the canonical constraint $-1 \le c_0 \le +1$ and $\norm{\c}=1$. 

Our purpose is to show that $F[\varrho,\mu,C]$ is jointly continuous in $\mu$ and $C$. In fact, the fraction function $F[\varrho,\mu,C]$  is continuous almost everywhere (including those where the denominator vanishes but the numerator is strictly positive). The only  points we have to inspect are those where both the numerator and the denominator vanish. These points, however, do not exist for non-degenerate canonical states.

Indeed, the numerator vanishes, i.e.,  $\int \d \mu ({\vv})  |c_0 + \c^T {\vv}|=0$, implies that  $c_0 + \c^T {\vv}=0$ is of measure $1$. However the minimal requirement imposes that $\int \d \mu ({\vv}) {\vv} = 0$. This is only possible when $c_0=0$. 
However when $c_0=0$, the denominator never vanishes if $T$ is non-degenerate.

Thus we have shown that $F[\varrho,\mu,C]$ is jointly continuous in $\mu$ and $C$. By Corollary~\ref{th:topo_continuity}, $r[\varrho,\mu]= \inf_C F[\varrho,\mu,C]$ is also continuous in $\mu$. (Note that we have shown the upper-semicontinuity of $r[\rho,\mu]$ more generically in Proposition~\ref{th:principal_upper_semiconitinuity}; here the conclusion on continuity only adds the information on its lower-semicontinuity.) 

\end{proof}

%\begin{proof}[tentative for the degenerate states]
%(i) We see immediately when the state is non-degenerate, $\Ker (T) = O$, the denominator never vanishes. 
%
%
%(ii) Now suppose the state is degenerate, thus so is $T$. The we have to investigate those LHS ensemble $\mu$ such that $\Ker (\mu) \cap \Ker (T)= Q \ne 0$. The idea is that, in this case, $\c$ can be limited to the unit sphere in $Q^{\perp}$. Indeed, every vector $\c$ can be decomposed as $\c=\c_1 + \c_2$, where $\c_1 \in Q$ and $\c_2 \in Q^{\perp}$. The numerator then reduces to
%$\int \d \mu ({\vv})  |\c^T {\vv}|=\int \d \mu ({\vv})  |\c^T_2 {\vv}|$.
%And the denominator reduces to $\norm{T \c} =\norm{T\c_2}$. 
%
%Therefore, for such a state $\varrho$, and for such a probability measure, we can replace the fraction function as
%\begin{equation}
%F[\varrho,\mu,C]=\frac{\int \d \mu ({\vv})  |\c^T_2 {\vv}|}{\norm{T \c_2}},
%\end{equation}
%with constraint $\norm{\c_2} \le 1$, which can also be replaced by $\norm{\c_2}=1$.
%\end{proof}

\begin{remark}
The lower-semicontinuity of the principal radius of canonical states on degenerate states perhaps also holds. The detailed analysis is however tedious. To support what follows, it is sufficient for us to restrict to non-degenerate states; but see also Section~\ref{sec:computation_degenerate}.
\end{remark}

\section{Finiteness of the critical radius}
\begin{proposition}
The critical radius is finite except for states of the form $\frac{\II}{2} \otimes \varrho_B$. 
\label{th:singularity}
\end{proposition}
\begin{proof}
To show that $R[\varrho]$ is finite, we will show that $r[\varrho,\mu]$ is bounded.
It is obvious that $r [\varrho,\mu]$ is lower-bounded by $0$. To show that $r[\varrho,\mu]$ is upper-bounded, we observe that
\begin{equation}
\frac{\int \d \mu (\sigma)  |\dprod{C}{\sigma}|}{\sqrt{2}\norm{\Tr_B[\bar{\varrho} (\II_A \otimes C)]}} \le
\frac{\int \d \mu (\sigma)  \norm{\sigma} \norm{C}}{\sqrt{2}\norm{\Tr_B[\bar{\varrho} (\II_A \otimes C)]}}, 
\end{equation}
for any probability measure $\mu$ by Cauchy--Schwarz inequality, or
\begin{equation}
\frac{\int \d \mu (\sigma)  |\dprod{C}{\sigma}|}{\sqrt{2}\norm{\Tr_B[\bar{\varrho} (\II_A \otimes C)]}} \le
\frac{\norm{C}}{\norm{\Tr_B[\bar{\varrho} (\II_A \otimes C)]}}.
\end{equation}
So
\begin{equation}
R[\varrho] = \max_\mu r[\varrho,\mu] \le \inf_{C} \frac{\norm{C}}{\norm{\Tr_B[\bar{\varrho} (\II_A \otimes C)]}},
\end{equation}
where the right-hand-side is certainly upper-bounded except for $\bar{\varrho}=0$, or $\varrho=\frac{\II}{2} \otimes \varrho_B$. 
\end{proof}

%====================================================================================
\section{Continuity of the critical radius}
While it is desirable to have some feeling of the continuity of the critical radius, this section is technically only needed to demonstrate the closeness of the set of unsteerable states. Readers who are more interested in the practical computation of the critical radius can thus safely skip this section.  

The continuity of the critical radius is a bit subtle. In this section, we will have to consider non-degenerate states more explicitly. We will study the continuity of the critical radius when restricted to certain subsets of the defining domain of the critical radius (c.~f., Section~\ref{sec:defining_domain}): starting from canonical non-degenerate and general canonical states, then extending to non-degenerate normal states and normal states. Within each subset, we will use the notion of relative continuity, which is the continuity with respect to the topology of the subset. Note that when the subset under consideration is not open, this is different from the notion of continuity at every point of the subset when considering the function over the whole defining domain. 

As the topology of the considered subsets matters, we note also that the set of normal states is convex and inherits a natural topology of the defining domain of the critical radius (which inherits the topology of the operator space). The set of abnormal states is closed (since the constraint is closed), and thus the set of normal states is open. The set of canonical states is convex and closed. %These facts are useful when we study the continuity of the critical radius. 

\begin{proposition}
The critical radius function $R$ is upper-semicontinuous relatively in the set of (degenerate and non-degenerate) canonical states. 
\label{th:critical_upper_semicontious}
\end{proposition}
\begin{proof}
Because the fraction function $F[\varrho,\mu,C]$ is upper-semicontinuous jointly in $(\varrho,\mu)$, we have that $r[\varrho,\mu]$ is jointly upper-semicontinuous in $(\varrho,\mu)$ by Lemma~\ref{lem:upper_semicontinuity}. Applying Lemma~\ref{lem:lower_semicontinuity} (with lower-semicontinuity replaced by upper-semicontinuity and infimum replaced by supremum), we then find that $R[\varrho]$ is  upper-semicontinuous. Note that the requirement that $\varrho$ is in the canonical form is indispensable: only in this case the minimal requirement for $\mu$ is independent of $\varrho$ and one can apply Lemma~\ref{lem:lower_semicontinuity}.
\end{proof}

\begin{proposition}
The critical radius function $R$ is continuous relatively in the set of non-degenerate canonical states. 
\label{th:critical_continous}
\end{proposition}
\begin{proof}
The proof is similar to the above proof. Here we note that $F[\varrho,\mu,C]$ is continuous in all variables when $\varrho$ is limited to non-degenerate canonical states (for the same reason as in the proof of Proposition~\ref{th:principal_lower_semicontinuity}). This guarantees that $r[\varrho,\mu]$ is jointly continuous in $(\varrho,\mu)$ by Corollary~\ref{th:topo_continuity}. Applying this corollary again for $r[\varrho,\mu]$, we find that $R[\varrho]$ is continuous relatively in the set of non-degenerate canonical states. 
\end{proof}

\begin{proposition}
The critical radius is upper-semicontinuous relatively in the set of normal states and continuous relatively in the set of non-degenerate normal states. 
\end{proposition}
\begin{proof}
On (non-degenerate or general) normal states, the critical radius function can be considered as a composition of a map from (non-degenerate or general) normal states to (non-degenerate or general) canonical states, and the map from the canonical states to their critical radius values. The former map (i.e., the map from normal states to canonical states) is continuous, and the latter is continuous relatively in the set of non-degenerate canonical states or upper-semicontinuous relatively in the set of canonical states due to the above propositions. Their composition is thus also continuous or upper-semicontinuous, respectively.   
\end{proof}

\begin{remark}
It perhaps also holds that the critical radius is continuous relatively in the set of all normal states, including the degenerate ones. The analysis is again tedious. 
\end{remark}

The continuity of the critical radius breaks down at abnormal product states. It is easy to see that the critical radius is discontinuous at pure product states. Indeed, for all pure entangled states, the critical radius is $\frac{1}{2}$, but it jumps to $1$ at pure product states.

Nevertheless, the upper-semicontinuity still holds at abnormal product states:

\begin{proposition}
The critical radius is upper-semicontinuous at states in the union of normal states and abnormal product states. 
\label{th:critical_upper_semicontinuous_full}
\end{proposition}
Note that here we can use the notion of continuity instead of relative continuity. 
\begin{proof}
Note that the set of normal states is open in the defining domain of the critical radius. Upper-semicontinuity relatively in the open set of normal states implies its upper-semicontinuity at normal states when considering the function over the whole defining domain. Now we consider abnormal product states, $\varrho=(\a,\a\b^T,\b)$. For any sequence $(\a_n,T_n,\b_n) \to (\a,\a\b^T,\b)$ (note that states in the sequence can be normal or abnormal), we have
\begin{equation}
r[(\a_n,T_n,\b_n),\mu]=\inf_{c_0,\c} \frac{\int \d \mu ({\vv})  |c_0 + \c^T {\vv}|}{\norm{c_0 \a_n +  T_n \c}} \le \frac{1}{\norm{\a_n}}.
\end{equation}
This upper-bound is obtained by limiting the infimum to $\c=0$. Therefore we also have $R[(\a_n,T_n,\b_n)] \le \frac{1}{\norm{\a_n}}$. So $\overline{\lim}_{n \to \infty} R[(\a_n,T_n,\b_n)] \le  \overline{\lim}_{n \to \infty} \frac{1}{\norm{\a_n}}=\frac{1}{\norm{\a}} = R[(\a,\a\b^T,\b)]$. This implies that $R$ is upper-semicontinuous at $\varrho=(\a,\a\b^T,\b)$.  
\end{proof}

From the above proof, one may find that the robustness of the upper-semicontinuity of the critical radius is somewhat surprising. It in particular implies that the critical radius is upper-semicontinuous relatively in the entire set of proper states. Nevertheless, we see shortly below that this upper-semicontinuity underlies the closeness of the set of unsteerable states---something we would naturally expect. It is reasonable to expect that the upper-semicontinuity of the critical radius eventually breaks down at abnormal, non-product states. However, these pathological states are improper states and of no physical interest. 

\section{Levels of the critical radius}
For $t \in \RR$, we define $C_t= \{\varrho: R[\varrho] \ge t\}$. Note that here $C_t$ contains also improper states by our convention, c.~f. Section~\ref{sec:defining_domain}. The intersection of $C_t$ with the set of proper states is denoted by $Q_t$ as in the main text. In particular, $Q_1$ is the set of all unsteerable proper states. 
\begin{proposition}
For any $t > 0$, the level set $C_t$ is bounded.
\end{proposition}
\begin{proof}
The boundedness of $C_t$ is obtained by an upper-bound for $R$. First, we notice that in the fraction function, we can assume $C$ is bounded. Therefore the numerator of the fraction function is bounded. Therefore
\begin{equation}
R[\varrho] \le \frac{A}{\sup_{c_0,\c}\norm{c_0 \a + T \c}},
\end{equation}
for some constant $A$.

So $R[\varrho] \ge t > 0$ implies $\sup_{c_0,\c}\norm{c_0 \a + T \c} \le A/t < + \infty$, which implies both $\a$ and $T$ are bounded. (Note that $\b$ is always bounded within Bob's Bloch sphere.)  

\end{proof}
\begin{proposition}
For any $t > 0$, the level set $C_t$ is closed relatively in the union of normal states and abnormal product states.
\end{proposition}
\begin{proof}
This is a direct consequence of the upper-semicontinuity of $R$ over the union of normal states and non-normal product states, c.~f. Proposition~\ref{th:critical_upper_semicontinuous_full}.
\end{proof}
\begin{remark}
When considered in the whole defining domain of the critical radius (or in set of all states), the set $C_t$ may not be closed at the (non-physical) abnormal non-product states.  
\end{remark}
\begin{proposition}
For any $t > 0$, the level set $C_t$  is convex.
\label{th:level_set}
\end{proposition}
\begin{proof}
The proposition is vacuous when $C_t$ is empty, so we assume that it is not empty. Suppose $R[\varrho_1] \ge t$ and $R[\varrho_2] \ge t$, we want to prove that for all $\lambda_1,\lambda_2 \ge 0$, $\lambda_1+\lambda_2=1$ we have $R[\varrho_0] \ge t$ with $\varrho_0=\lambda_1 \varrho_1 + \lambda_2 \varrho_2$. Let $\mu_1$ and $\mu_2$ be two optimal LHS ensemble for $\varrho_1$ and $\varrho_2$, respectively. Then for $i=1,2$, we have $R[\varrho_i]= r[\varrho_i,\mu_i]$. From the definition, we have
\begin{equation}
\inf_C \frac{\int \d \mu_i (\sigma) \abs{\dprod{C}{\sigma}}}{\sqrt{2}\norm{\Tr_B[\bar{\varrho}_i (\II_A \otimes C)]}} \ge t,
\end{equation}
with $\bar{\varrho}_i= \varrho_i - \frac{\II_A}{2} \otimes \varrho_B$.
Since the denominator is positive, this is equivalent to
\begin{equation}
\int \d \mu_i (\sigma) \abs{\dprod{C}{\sigma}} \ge t  \sqrt{2}\norm{\Tr_B[\bar{\varrho}_i (\II_A \otimes C)]}
\end{equation}
for all $C$. Multiplying the two sides with $\lambda_i$ and summing over $i$, we have
\begin{equation}
\int \d \mu_0 (\sigma) \abs{\dprod{C}{\sigma}} \ge t  \sum_{i=1}^{2} \lambda_ i \sqrt{2}\norm{\Tr_B[\bar{\varrho}_i (\II_A \otimes C)]},
\end{equation}
where $\mu_0=\lambda_1 \mu_1+\lambda_2 \mu_2$. Then using the triangular inequality, we have
\begin{equation}
\sum_{i=1}^{2} \lambda_ i \sqrt{2}\norm{\Tr_A[\bar{\varrho}_i (\II_A \otimes C)]} \ge \sqrt{2}\norm{\Tr_B[\bar{\varrho}_0 (\II_A \otimes C)]},
\end{equation}
with $\varrho_0= \lambda_1 \varrho_1 + \lambda_2 \varrho_2$ and $\bar{\varrho}_0= \varrho_0 - \frac{\II_A}{2} \otimes \Tr_B(\varrho_0)$.
Therefore 
\begin{equation}
\int \d \mu_0 (\sigma) \abs{\dprod{C}{\sigma}} \ge t \sqrt{2}\norm{\Tr_A[\bar{\varrho}_0 (\II_A \otimes C)]},
\end{equation}
or 
\begin{equation}
r [\varrho_0,\mu_0] = \inf_C \frac{\int \d \mu_0 (\sigma) \abs{\dprod{C}{\sigma}}}{\sqrt{2}\norm{\Tr_A[\bar{\varrho}_0 (\II_A \otimes C)]}} \ge t.
\end{equation}
Thus $R[\varrho_0] \ge r [\varrho_0,\mu_0] \ge t$.
\end{proof}
%\begin{remark}
%We note that the proof of the convexity of $C_t$ is rather subtle. The convexity of $r[\varrho,\mu]$ (or $r^{-1}[\varrho,\mu]$) is not enough, but rather the convexity of $C_t$ requires the linearity in $\mu$ of the fractional function $F[\varrho,\mu,C]$.
%\end{remark}
\begin{corollary}
For two states $\varrho_1$ and $\varrho_2$, we have $\min\{R[\varrho_1],R[\varrho_2]\} \le R[\lambda \varrho_1 + (1-\lambda) \varrho_2]$ for all $0 \le \lambda \le 1$.
\label{th:semi_concavity}
\end{corollary}
\begin{proof}
Let $t=\min\{ R[\varrho_1],R[\varrho_2]\}$, then $\varrho_1$ and $\varrho_2$ are both in $C_t$. Therefore $\lambda \varrho_1 + (1-\lambda) \varrho_2$ is also in $C_t$ due to its convexity. It follows by definition that $R[\lambda \varrho_1 + (1-\lambda) \varrho_2]\ge t=\min\{ R[\varrho_1],R[\varrho_2]\}$.
\end{proof}
\begin{remark}
If you start to wonder: we do not expect to have $\max\{R[\varrho_1],R[\varrho_2]\} \ge R[\lambda \varrho_1 + (1-\lambda) \varrho_2]$; in particular, if $\rho_1$ and $\rho_2$ are steerable, it certainly can be the case that $\lambda \varrho_1 + (1-\lambda) \varrho_2$ is unsteerable. 
\end{remark}
As a result of these properties of $C_t$, its intersection with the set of proper states, i.e., $Q_t$, is convex and compact. In particular, the set of unsteerable proper states $Q_1$ is convex and compact.

For the following proposition, let us define $S_t= \{\varrho : R[\varrho] = t\}$.
\begin{proposition}
For any $t > 0$, $[C_t \cap \operatorname{ext} (C_t)] \subseteq S_t \subseteq \partial C_t$. Here $\operatorname{ext} (C_t)$ is the set of extreme points of $C_t$ and $\partial C_t$ is the relative boundary of $C_t$. 
\end{proposition}
Note that we have not shown that $C_t$ is closed in the Bloch hyperplane of bipartite states. Therefore in principle $\operatorname{ext} (C_t)$ may not be a subset of $C_t$. Yet, as we mentioned, if $\operatorname{ext} (C_t) \setminus C_t$ is non-empty, it contains only spurious abnormal, non-product states, which are unphysical. 
\begin{proof}
(i) We start with showing that $[C_t \cap \operatorname{ext} (C_t)] \subseteq S_t$. Suppose $\varrho \in C_t$ but $\varrho \not\in S_t$, we show that $\varrho \not\in \operatorname{ext} (C_t)$. 

If $0 < t < R[\varrho] < +\infty$, then we let $\tilde{\varrho}= (R[\varrho]/t) \varrho + (1-R[\varrho]/t) \frac{\II_A}{2} \otimes \varrho_B$. Because $R[\tilde{\varrho}]=t$, it is in $C_t$. On the other hand, we have $\varrho= (t/R[\varrho]) \tilde{\varrho} + (1-t/R[\varrho]) \frac{\II_A}{2} \otimes \varrho_B$, which gives an explicit non-trivial convex decomposition of $\varrho$ in terms of $\tilde{\varrho}$ and $\frac{\II_A}{2} \otimes \varrho_B$, which are both in $C_t$.  Therefore $\varrho \not\in \operatorname{ext} (C_t)$. 

Now we consider the case $R[\varrho]=+\infty$. This implies that $R[\varrho]=\frac{\II_A}{2} \otimes \varrho_B$. We can then make a convex decomposition $\frac{\II_A}{2} \otimes \varrho_B=\frac{1}{2} [(\frac{\II_A}{2} + \epsilon  \sigma_z) \otimes \varrho_B] + \frac{1}{2} [(\frac{\II_A}{2} - \epsilon \sigma_z) \otimes \varrho_B]$. Each of the states in this decomposition has critical radius $1/\epsilon$ (see Section~\ref{sec:analytical_formulae_product_states}), which is larger than $t$ if $\epsilon$ is sufficiently small. Thus for sufficiently small $\epsilon$, both states are in $C_t$. Therefore also in this case $\varrho$ cannot be an extreme point of $C_t$.

(ii) Now we show that $S_t \subseteq \partial C_t$. Suppose $\varrho \not\in \partial C_t$, that is, $\varrho$ is in the relative interior of $C_t$, we show that $\varrho \not\in S_t$. By Theorem 6.4 in Ref.~\cite{Rockafellar1970a}, take $\frac{\II_A}{2} \otimes \varrho_B \in C_t$, there exists $\epsilon > 0$ such that $ (1+\epsilon) \varrho - \epsilon \frac{\II_A}{2} \otimes \varrho_B$ is in $C_t$. So $R[(1+\epsilon) \varrho - \epsilon \frac{\II_A}{2} \otimes \varrho_B]=1/(1+\epsilon) R[\varrho] \ge t$. It then follows that $R[\varrho] \ge  (1+ \epsilon )t > t$, thus $\varrho \not\in S_t$.
\end{proof}

%====================================================================================
%====================================================================================
\section{Analytic formula of the critical radius for certain states}
\label{sec:analytical_formulae}
%-----------------------------------------------------------------------------------
\subsection{Product states}
\label{sec:analytical_formulae_product_states}
A product state is of the form $\varrho_A \otimes \varrho_B$. If $\varrho_B$ is pure, the state is abnormal. In this case, we however have shown in Section~\ref{sec:canonical} that its critical radius is simply $R[\varrho_A \otimes \varrho_B]=\frac{1}{\norm{\a}}$.
 
When a product state $\varrho_A \otimes \varrho_B$ is normal, one can bring it to the canonical form $\varrho_A \otimes \frac{\II_B}{2}$. Now this state is in fact $\U(2)$ invariant, where $\U(2)$ acts trivially on $\H_A$ and acts as conjugation on $\H_B$. Thus an optimal choice for LHS ensemble would be the uniform distribution. The lower bound~\eqref{eq:uniform_ansatz_bound} below is tight. Direct computation then also gives $R[\varrho_A \otimes \varrho_B]= R[\varrho_A \otimes \frac{\II_B}{2}]=\frac{1}{\norm{\a}}$. [In more details: this is nothing but the the lower bound~\eqref{eq:uniform_ansatz_bound}, which is tight as the uniform distribution is optimal; also note that the correlation matrix here vanishes so the infimum can be found easily.] 

%-----------------------------------------------------------------------------------
\subsection{T-states}
In the canonical form, if $\a=0$, we have a $T$-state, also known as a Bell-diagonal state. The $T$-states form the most interesting class of normal states where an analytical formula for the critical radius has been found~\cite{Jevtic2015a,chauepl}. 
The central simplicity of $T$-state is that it carries a time-reversal symmetry on both parties. As a result, the optimal LHS ensemble can be chosen to be central symmetric on the Bloch sphere~\cite{chauepl}. Therefore, we can set $c_0=0$ and the critical radius becomes
\begin{equation}
R[(0,T)]=\max_\mu \inf_{\c} \frac{\int \d \mu ({\vv}) \abs{\c^T {\vv}}}{\norm{T\c}}.
\label{eq:c}
\end{equation} 
It can be shown that $\mu$ can be taken to be supported only on the Bloch sphere; see Section~\ref{sec:computation}. It was recognised by Jevtic and her collaborators~\cite{Jevtic2015a} that, for a $T$-state with correlation matrix $T$, the LHS ensemble generated by 
\begin{equation}
J(\n)= \frac{N_T}{[\n^T T^{-2} \n]^2},
\label{eq:jevtic}
\end{equation} 
as a distribution on the Bloch sphere with
\begin{equation}
N_T^{-1}= \int \d S(\n) \frac{1}{[\n^T T^{-2} \n]^2},
\label{eq:nt_tstates}
\end{equation} 
has some rather special property. Namely, the boundary of the simulated states exactly resembles the so-called steering ellipsoid~\cite{jevtic2014,Jevtic2015a}. This leads to the conjecture that the LHS ensemble is optimal for Alice to simulate steering on Bob's system, which was later proven in Ref.~\cite{chauepl}. 

Translated into our current language, for the distribution~\eqref{eq:jevtic}, the fraction function is in fact independent of $\c$,
\begin{equation}
\frac{\int \d S (\n) J(\n) \abs{\c^T \n}}{\norm{T\c}}=2 \pi N_T \abs{\det(T)}.
\label{eq:anderson_integral}
\end{equation} 
It was then proven that any deviation from $J(\n)$ leads to a decrease in the principal radius~\cite{chauepl}. 
This gives rise to an analytical formula for the critical radius of $T$-states as
\begin{equation}
R[(0,T)]= 2 \pi N_T \abs{\det(T)}.
\label{eq:t_states}
\end{equation} 
For the case where the correlation matrix $T$ has axial symmetry, e.g., $T=\operatorname{diag}(s,s,t)$, $R$ can be given in a closed form,
\begin{equation}
R[(0,T)]= \frac{1}{\abs{t}} \frac{1}{1+(1+x^2) \frac{\operatorname{artg} (x)}{x}}, 
\end{equation} 
with $x=\sqrt{s^2/t^2-1}$, which can take purely imaginary values when $\abs{s/t}<1$.
 
%\blue{Would be good if give a state we can quickly check if it can be brought to a $T$-states.}

%\blue{We can also provide a `neater' trick to calculate of the critical radius than doing in coordinate as in Sania's paper. I can add this, if it is not so distracted.}
\begin{remark}
In Ref.~\cite{Jevtic2015a}, the integral of the form~\eqref{eq:anderson_integral} was performed using direct computation in coordinates. Here we give a coordinate-independent computation of the integral. 
This is done by relaxing the dimension of the integral. Namely, we consider the integral,
\begin{equation}
I= \int \d V(\r) \abs{\c^T \r} e^{-\r^T T^{-2} \r},
\label{eq:dimensional_relaxation}
\end{equation}
which is taken with respect to the volume measure $V$ over the whole $3$D space of $\r$. The relation to the integral~\eqref{eq:anderson_integral} can be realised by separating the integral over the radials $r$ and the unit vector directions $\n$, $\r = r \n$, namely 
\begin{align}
I &= \int \d S(\n) \int_0^{\infty} \d r r^3  \abs{\c^T \n} e^{-r^2 \n^T T^{-2} \n} \nonumber \\
&= \int \d S(\n) \frac{\abs{\c^T \n}}{[\n^T T^{-2} \n]^2} \int_0^{\infty} \d x x^3 e^{-x^2} \nonumber \\
&= \frac{1}{2} \int \d S(\n) \frac{\abs{\c^T\n}}{[\n^T T^{-2} \n]^2}.
\label{eq:radial_separation}
\end{align}

To perform the integral~\eqref{eq:dimensional_relaxation}, we make a variable transformation $\r=T\tilde{\r}$. This gives
\begin{equation}
I= \abs{\det(T)} \norm{T\c} \int \d V(\tilde{\r}) \abs{\tilde{\c}^T \tilde{\r}} e^{-\tilde{r}^2},
\end{equation} 
where $\tilde{\c}=T\c/\norm{T\c}$ is a unit vector. The latter integral can be performed directly in spherical coordinates with the $z$-axis along $\tilde{\c}$,
\begin{align}
\int \d V(\tilde{\r}) \abs{\tilde{\c}^T \tilde{\r}} e^{-\tilde{r}^2} &= \int_0^{2 \pi} \d \phi \int_0^{\pi} \d \theta \sin \theta \abs{\cos \theta} \int_{0}^{+\infty} \d r r^3 e^{-r^2} \nonumber \\
&= \pi.
\end{align}
Thus $I= \pi \abs{\det(T)} \norm{T\c}$, which, by~\eqref{eq:radial_separation} leads to
\begin{equation}
\int \d S(\n) \frac{\abs{\c \cdot \n}}{[\n^T T^{-2} \n]^2}= 2 \pi \abs{\det(T)} \norm{T\c}.
\end{equation}
This in turn directly leads to~\eqref{eq:anderson_integral}. 
\end{remark}
%---------------------------------------------------------------------------------
\subsection{Some analytical bounds for the critical radius}

\begin{theorem}
For a non-degenerate canonical state, we have
\begin{equation}
2 \pi N_T \abs{\det(T)} \ge R[\varrho] \ge \frac{2 \pi N_T \abs{\det(T)}}{1+ \norm{T^{-1}\a}},
\label{eq:analytical_bounds}
\end{equation}
where $N_T^{-1}= \int \d S (\n) [\n^T T^{-2} \n]^{-2}$.
\label{th:analytical_bounds}
\end{theorem}
Note that when we set $\a=0$, the state becomes a $T$-state and the lower bound and upper bound meet at $2 \pi N_T \abs{\det(T)}$, recovering the formula for the critical radius for $T$-states.
\begin{proof}
The upper bound is actually obvious, since limiting the domain of infimum by setting $c_0=0$ always increases the infimum. We therefore only need to prove the lower bound. 

To find the lower bound, we find a minimal factor $\lambda$ such that $\norm{c_0 \a +  T \c} \le \lambda \norm{ T \c}$ for all $\norm{\c}=1$ and $-1 \le c_0 \le +1$. It is easy to show that $\lambda=1+ \norm{T^{-1} \a}$ should work.  
This can be seen as follows. To show that $\norm{c_0 \a +  T \c} \le \lambda \norm{ T \c}$, we show that $\{c_0 \a +  T \c: -1 \le c_0 \le +1, \norm{\c}=1\} \subseteq \{\lambda T\c: \norm{\c} \le 1\}$. By applying $T^{-1}$ to both sets, the latter is equivalent to $\{c_0 T^{-1} \a + \c: -1 \le c_0 \le +1, \norm{\c}=1\} \subseteq \{\lambda \c: \norm{\c} \le 1\}$. This is the case if $\lambda \ge \max \{ \norm{c_0 T^{-1} \a + \c}: -1 \le c_0 \le +1, \norm{\c}=1\}=1+\norm{T^{-1} \a}$.  

We therefore see that 
\begin{equation}
r[\varrho,\mu]= \frac{\int \d \mu ({\vv})  |c_0 + \c^T {\vv}|}{\norm{c_0 \a +  T \c}} \ge \frac{\int \d \mu ({\vv})  |c_0 + \c^T {\vv}|}{(1+ \norm{T^{-1} \a}) \norm{T \c}},
\end{equation}
for all $\mu$, $\norm{\c}=1$ and $-1 \le c_0 \le +1$. So, when taking the infimum over $c_0$ and $\c$ and the maximum over $\mu$, we obtain 
\begin{equation}
R[\varrho] \ge \frac{2 \pi N_T \abs{\det(T)}}{1+ \norm{T^{-1}\a}},
\end{equation}
where the left-hand-side is obtained using the solution of the critical radius for $T$-states. 
\end{proof}

%\begin{remark}

%Perhaps the optimal $T$-states bound, for the case $T$ is diagonal and $\a$ has only $z$-component is
%\begin{equation}
%2 \pi N_{\tilde{T}} \abs{\det(\tilde{T})}
%\end{equation}
%where $\tilde{T}= T \operatorname{diag}(\sqrt{1+\norm{T^{-1} \a}},\sqrt{1+\norm{T^{-1} \a}},1+\norm{T^{-1} \a})$
%\end{remark}

\begin{corollary}
For a state in the canonical form $\varrho=(\a,T)$, we have $R[(\a,T)] \le R[(p\a,T)]$ for all $0 \le p \le 1$. 
\label{th:depolarising_alice}
\end{corollary}
In other word, depolarising Alice's state keeping the bipartite correlations intact decrease steerability. 
\begin{proof}
Note that $(p\a,T)= p(\a,T)+(1-p) (0,T)$. According to Corollary~\ref{th:semi_concavity}, we have $R[(\a,T)]=\min\{R[(\a,T)],R[(0,T)]\} \le R[(p\a,T)]$.
\end{proof}

Despite the fact the lower bound in equation~\eqref{eq:analytical_bounds} is tight for $T$-states, it is often far from tight when $\a \ne 0$. Although we can improve the lower bound, it is perhaps only of theoretical interest. For the practical purpose, the lower bound discussed below is often better. 

\begin{theorem}
For a state given in the canonical form,
\begin{equation}
R[\varrho] \ge \frac{1}{2} \inf_{c_0,\c}  \frac{1+c_0^2}{\norm{c_0 \a + T \c}},
\label{eq:uniform_ansatz_bound}
\end{equation}
subject to the constraint $-1 \le c_0 \le +1$ and $\abs{\c}=1$.
\end{theorem}
\begin{proof}
By using any measure that satisfies the minimal requirement as an ansatz for the LHS ensemble, we obtain a lower bound for the critical radius. If we choose the uniform distribution supported on the Bloch sphere as the ansatz, then we can evaluate the numerator $\frac{1}{4 \pi}\int \d S(\n) \abs{c_0 + \c^T \n}= \frac{1+c_0^2}{2}$ exactly. Using this result, we obtain~\eqref{eq:uniform_ansatz_bound}.  
\end{proof}

The uniform distribution on the Bloch sphere has been used as an ansatz to prove unsteerability of two-qubit states~\cite{Bowles2016a,chaupra}. Here we used it to get a quantitative bound for the critical radius. 
%====================================================================================
\section{Computation of the critical radius} 
\label{sec:computation}
%-----------------------------------------------------------------------------------
\subsection{Bringing the state to the canonical form}
If the state is abnormal, we can compute the critical radius directly via formula~\eqref{eq:critical_non_normal}. If the state is normal, the very first step is to bring it to the canonical form. 

\new{This can be done using the following procedure. Starting with a state $\varrho$, one obtains $\varrho_1= \I_A \otimes V_B \varrho \I_A \otimes V_B/\Tr(\I_A \otimes V_B \varrho \I_A \otimes V_B)$ with $V_B=\sqrt{\varrho_B^{-1}}$. One then derives the Bloch tensor $\Theta_1$ for $\varrho_1$,
\begin{equation}
\Theta_1 = 
\begin{pmatrix}
1 & \pmb{0}^T \\
\a_1 & T_1
\end{pmatrix}.
\end{equation}
Now note that local unitary transformations are implemented by local rotations of the Bloch tensor. Utilising the invariance of the critical radius under local time reversals, we can also extend from the local rotations of the Bloch tensor to the general local orthogonal transformations, including the improper rotations. To this end, we find the singular value decomposition of $T_1$ as $T_1= L_1 \operatorname{diag} (\s) R_1$, where $L_1$ and $R_1$ are orthogonal matrices and $\s$ are singular values of $T_1$. We then apply local rotations $L^T_1$ and $R^T_1$ on $\Theta_1$ to obtain $\Theta_2$,
\begin{align}
\Theta_2 &= 
\begin{pmatrix}
1 & \pmb{0}^T \\
\pmb{0} & L^T_1 
\end{pmatrix}
\begin{pmatrix}
1 & \pmb{0}^T \\
\a_1 & T_1
\end{pmatrix}
\begin{pmatrix}
1 & \pmb{0}^T \\
\pmb{0} & R^T_1
\end{pmatrix} \nonumber \\
&= 
\begin{pmatrix}
1 & \pmb{0}^T \\
L^T_1 \a_1 & \operatorname{diag}(\s)
\end{pmatrix}.
\end{align}
This is the Bloch tensor representation of the canonical form.}

In the following, states are assumed to be non-degenerate and in the canonical form. These would include all steerable states. \new{Although degenerate states are separable, and thus unsteerable, later we will also remark how one can compute the critical radii for degenerate states for completeness.}

%-----------------------------------------------------------------------------------
\subsection{Sandwiching the Bloch sphere between two polytopes}
We would like to approximate the Bloch sphere by a discrete set of points in order to carry out the computation. Note that the concepts of principal radius and critical radius apply naturally when $\mu$ is a probability measure on some arbitrary compact set $\S$, provided its convex hull contains Bob's reduced state (which is the center of the Bloch sphere, since the bipartite state is in the canonical form). The latter requirement is to make sure that the minimal requirement does not result in an empty set of measures. Indeed, for a compact subset $\S$ of the Bloch hyperplane, for which the convex hull contains the center of the Bloch sphere and a probability measure $\mu$ on $\S$ satisfying the minimal requirement, $\int_{\S} \d \mu (\sigma) \sigma = \frac{\II_B}{2}$, we can naturally define the fraction function
\begin{equation}
F^{\S}[\varrho,\mu,C]= \frac{\int_{\S} \d \mu (\sigma)|\dprod{C}{\sigma}|}{\norm{\Tr_B[\bar{\varrho} \II_A \otimes C]}},
\end{equation}
where $\bar{\varrho}= \varrho - \frac{\II_A}{2} \otimes \varrho_B$. 
The principal radius is defined by
\begin{equation}
r^{\S}[\varrho,\mu]= \inf_{C} F^{\S}[\varrho,\mu,C].
\label{eq:relative_r}
\end{equation}
It is again possible to show that $r^{\S}[\varrho,\mu]$ is upper-semicontinuous in $\mu$. We then define the critical radius to be
\begin{equation}
R^{\S}[\varrho]= \max_\mu r^{\S} [\varrho,\mu],
\label{eq:critical_r_polytope}
\end{equation}
where $\mu$ are Borel measures on $\S$ subjected to the minimal requirement. 

The following theorem then allows us to compare the critical radius defined on nesting convex sets.
\begin{theorem}
In the Bloch hyperplane, suppose a compact set $\S_1$ is contained in the convex hull of a compact set $\S_2$, then for a non-degenerate canonical state $\varrho$, we have $R^{\S_1}[\varrho] \le R^{\S_2}[\varrho]$.
\label{th:nested_bound}
\end{theorem}
\begin{proof}
For non-degenerate canonical states, $r^{\S_1}[\varrho,\mu]$ is continuous in $\mu$. This is obtained by adapting the proof of Proposition~\ref{th:principal_lower_semicontinuity}.
Our strategy is to show that $r^{\S_1} [\varrho,\mu] \le R^{\S_2}[\varrho]$ on the set of finitely-supported probability measures, which is dense in the set of all Borel probabilistic measures~\cite{Parthasarathy1967a}. In fact we show that, for all finitely-supported measures $\mu$ on $\S_1$ satisfying the minimal requirement constraint, there exists a measure $\nu$ satisfying the minimal requirement on $\S_2$ such that $r^{\S_1} [\varrho,\mu] \le r^{\S_2}[\varrho,\nu]$. The latter is established if we can show that $\K(\mu) \subseteq \K(\nu)$.

Indeed, suppose the measure $\mu$ on $\S_1$ is characterised by discrete weights $\{u_i\}_{i=1}^{N}$ at discrete $3$D vectors $\{\pmb{t}_i\}_{i=1}^N$ on the Bloch hyperplane. 
Because $\pmb{t}_i$ is in the convex hull of $\S_2$, there exists a convex decomposition of each $\pmb{t}_i$ into finite $M_i$ points $\{\r_j\}_{j=1}^{M_i}$ of $\S_2$ (Caratheodory's principle),
\begin{equation}
\pmb{t}_i = \sum_{j=1}^{M_i} q^{i}_{j} \r_j^i, 
\end{equation}
where $q^i_j \ge 0$ and $\sum_{j=1}^{M_i} q^i_j=1$. So far we ignore the zeroth coordinate of the Bloch vectors in the full operator space, which are simply $1$. Taken this zeroth coordinate into account, we can write
\begin{equation}
\begin{pmatrix}1 \\ \pmb{t}_i \end{pmatrix} = \sum_{j=1}^{M_i} q^{i}_{j} \begin{pmatrix} 1 \\\r_j^i \end{pmatrix}.
\end{equation}
The set $\cup_{i=1}^{N} \{\r_j^i\}_{j=1}^{M_i}$ thus contains at most finite number of elements, and is denoted by $\{\r_k\}_{k=1}^M$. The convex decomposition above can be extended to run over all $M$ vectors, with coefficient $q^{i}_k$ set to zero when not defined so that we can write
\begin{equation}
\begin{pmatrix}1 \\ \pmb{t}_i \end{pmatrix} = \sum_{k=1}^{M} q^{i}_{k} \begin{pmatrix} 1 \\\r_k \end{pmatrix},
\label{eq:t_decomposition}
\end{equation}
with $\sum_{k=1}^{M}q^i_k=1$.

Then we define the weights $v_k$ at $\r_k$ by
\begin{equation}
v_k= \sum_{i=1}^N q^i_k u_i. 
\end{equation}

We claim that these weights $\{v_k\}_{k=1}^{M}$ define a discrete measure $\nu$ on $\S_2$ that has the desired properties. 

Indeed, for the minimal requirement, it is easy to see that 
\begin{align}
\sum_{k=1}^M v_k \r_k &= \sum_{k=1}^M \sum_{i=1}^N q^i_k u_i \r_k \\
&= \sum_{i=1}^M u_i \pmb{t}_i.
\end{align}

To show that $\K(\mu) \subseteq \K(\nu)$, we pick up an element $K$ of $\K(\mu)$ and show that $K \in \K(\nu)$. By the definition of $\K(\mu)$, there exist coefficients $\{g_i\}_{i=1}^N$, $0 \le g_i \le 1$, such that
\begin{equation}
K= \sum_{i=1}^N g_i u_i \begin{pmatrix} 1 \\ \pmb{t}_i \end{pmatrix}.
\end{equation}
Therefore, using~\eqref{eq:t_decomposition}, 
\begin{align}
K &= \sum_{i=1}^N \sum_{k=1}^{M} g_i u_i q^i_k  \begin{pmatrix} 1 \\ \r_k \end{pmatrix} \\
&= \sum_{k=1}^{M} \sum_{i=1}^N g_i u_i q^i_k  \begin{pmatrix} 1 \\ \r_k \end{pmatrix}
\end{align}
Let us fix $k$. Because $0 \le g_i\le 1$, (due to the mean value theorem in the discrete form) there exist $0 \le f_k \le 1$ such that
\begin{equation}
\sum_{i=1}^N g_i u_i q^i_k = f_k \sum_{i=1}^N u_i q^i_k. 
\end{equation}
Thus we have
\begin{equation}
K= \sum_{k=1}^{M} f_k v_k \begin{pmatrix} 1 \\ \r_k \end{pmatrix},
\end{equation}
for $0\le f_k \le 1$, or $K \in \K(\nu)$.

\end{proof}

The following corollary is a direct consequence of the above theorem.
\begin{corollary}
For a compact convex set $\S$ on the Bloch hyperplane containing the center of the Bloch sphere and with compact set of extreme points $\operatorname{ext} (S)$, we have $R^{\S}[\varrho] = R^{\operatorname{ext} (\S)}[\varrho]$ for a non-degenerate canonical state $\varrho$.
\end{corollary}
\begin{proof}
Since $\S$ is convex and compact, $\operatorname{ext} (\S) \subseteq \S$. It follows that $R^{\S}[\varrho] \ge  R^{\operatorname{ext} (\S)}[\varrho]$. On the other hand, $\S$ is inside the convex hull of $\operatorname{ext} (\S)$, by the above theorem, we have $R^{\S}[\varrho] \le  R^{\operatorname{ext} (\S)}[\varrho]$. Therefore $R^{\S}[\varrho] = R^{\operatorname{ext} (\S)}[\varrho]$.
\end{proof}
\new{Applied to the Bob's Bloch ball, this corollary implies that the LHS ensemble in equation (4) in the main text can be assumed to be supported on the Bloch sphere (i.~e., the pure states), excluding the mixed states. This fact has been actually often assumed in the literature without a proper proof.}

Computationally, the above theorem allows us to lower-bound and upper-bound the critical radius of a non-degenerate canonical state by approximating the Bloch sphere by a finite number of points. To be specific, let $\S_B^-$ and $\S_B^+$ be the sets of vertices of two convex polytopes  such that $\operatorname{conv} \S_B^- \subseteq \B_B \subseteq \operatorname{conv} \S_B^+$. When the polytopes are fixed by context, we denote $R^{+}[\varrho]= R^{\S_B^+}[\varrho]$ and $R^{-}[\varrho]= R^{\S_B^-}[\varrho]$ for simplicity.  We then have $R^{-}[\varrho] \le R[\varrho] \le R^{+}[\varrho]$. In practice, we can choose $\S_B^-$ on Bob's Bloch sphere $\S_B$ itself, and $\S_B^+$ such that the surface of its convex hull circumscribes $\S_B$. With sufficient high numbers of vertices where both $\S_B^-$ and $\S_B^+$ are good approximations for $\S_B$, we can expect to have a good approximation for $R[\varrho]$. This is indeed the case due to the bound of errors discussed below. 

A note on convention: in the following, polytopes are always assumed to be convex. Here and in the following a polytope may mean the set of its vertices or the whole convex polytope itself. This ambiguity should not cause any confusion, since it should be clear from the context what is meant by a polytope.  
   
%-----------------------------------------------------------------------------------
\subsection{Universal bound of the relative error}
In practice, it is convenient to choose $\S_B^-$ as a discrete set on the Bloch sphere \emph{with the inversion symmetry}. Let $r_\mathrm{in}$ be the inscribed radius of the polytope $\S_B^-$. Note that due to the inversion symmetry, the center of the inscribed sphere of the polytope is at the origin. We then define the enlarged polytope $\S_B^+=\{\vv= \eta \n: \n \in \S_B^- \}$ with $\eta= 1/r_\mathrm{in}$. The enlarged polytope then contains the Bloch ball. We therefore have that $R^-[\varrho] \le R[\varrho] \le R^+[\varrho]$. More interestingly, we also have $R^{+}[\varrho] \le \eta R^-[\varrho]$, which leads to a universal bound of the relative error to be $1/\eta-1$, regardless of the details of the input state. 
\begin{theorem}
Consider a polytope $\S_B^-$ with inversion symmetry and $\S_B^+=\{\vv= \eta \n: \n \in \S_B^- \}$ with $\eta= 1/r_\mathrm{in}$. For a canonical state $\varrho$, we have $R^+[\varrho] \le \eta R^-[\varrho]$. 
\end{theorem}
\begin{proof}
%\red{[The proof seems to be rather subtle actually!]}
For simplicity, we denote the canonical state $\varrho$ by $(\a,T)$ and as before, we write $R^{\pm}[\varrho]=R^{\pm}[(\a,T)]$. Then we have $R^+[(\a,T)] = \eta R^-[(\eta \a,T)]$. To see this, we start 
\begin{align}
R^+[(a,T)] &= \max_\mu \inf_{c_0,\c} \frac{\int_{\S_B^+} \d \mu ({\vv}) \abs{c_0 + \c^T {\vv}} }{\norm{c_0 \a + T \c}} \nonumber \\
&= \max_\mu \inf_{c_0,\c} \frac{\int_{\S_B^-} \d \mu (\n) \abs{c_0 + \eta \c^T \n} }{\norm{c_0 \a + T \c}} \nonumber \\
&= \max_\mu \inf_{c_0,\c} \frac{\int_{\S_B^-} \d \mu (\n) \abs{\eta c_0 + \eta \c^T \n} }{\norm{\eta c_0 \a + T \c}} \nonumber \\
&= \eta \max_\mu \inf_{c_0,\c} \frac{\int_{\S_B^-} \d \mu (\n) \abs{c_0 + \c^T \n} }{\norm{\eta c_0 \a + T \c}} \nonumber \\
&= \eta R^{-}[(\eta \a,T)].
\end{align}
In the above manipulation, note that $\mu$ is just a discrete measure defined by a finite probability weights on $\S_B^-$ or $\S_B^+$ (the integral thus can be replaced by a discrete sum). 

We now only need to show that for $\eta \ge 1$, $R^-[(\eta \a, T)] \le R^-[(\a, T)]$. This is in fact the content of Corollary~\ref{th:depolarising_alice}, except that there $\S_B^-$ is replaced by the whole Bloch ball $\B_B$. We thus just need to investigate the validity of  the corollary when the Bloch sphere is approximated by $\S_B^-$. Corollary~\ref{th:depolarising_alice} is in turn based on Corollary~\ref{th:semi_concavity}, thus Proposition~\ref{th:level_set} and the upper bound in Theorem~\ref{th:analytical_bounds}. We now inspect their validity.

(i) Proposition~\ref{th:level_set} in fact makes no use of any property of the Bloch sphere and works just fine for $\S_B^-$. The level set $C^{-}_{t}= \{(\a,T):R^{\S_B^-}[(\a, T)] \ge t\}$ is thus indeed convex. Corollary~\ref{th:semi_concavity} is also valid.

(ii) The upper bound in Theorem~\ref{th:analytical_bounds} is obtained by restricting the domain of infimum to $c_0=0$. We then need to show that 
\begin{equation}
R^{\S_B^-} [(0,T)]=\inf_{\c} \frac{\int_{\S_B^-} \d \mu (\vv) \abs{\c^T \vv}}{\norm{T\c}}.
\end{equation} 
The latter means that $c_0$ can indeed be set to $0$ in the definition of $R^{\S_B^-}[(0,T)]$ for $T$-states. This is based on the fact that $T$-states $(0,T)$ have the time-reversal symmetry implemented by the inversion of the operator space, which implies that the LHS ensemble $\mu$ can be chosen to be central symmetric. One then simply notes that this whole procedure remains valid for the approximated Bloch sphere $\S_B^-$ provided $\S_B^-$ has the inversion symmetry. 
\end{proof}

While one somehow might have anticipated the bound $R^+[\rho] \le R^-[\rho]/r_{\rm in}$, in the above proof, the inversion symmetry of the polytope (which is nothing but time-reversal symmetry) enters in a rather subtle way. We see once again the fundamental role of time-reversal symmetry in quantum steering, which seemed to be overlooked. 
%-----------------------------------------------------------------------------------
\subsection{Optimise the principal radius over probability distributions on a polytope}
It is now left to describe an algorithm to compute $R^{\S}[\varrho]$ where $\S$ is the set of vertices of a polytope. As $\mu$ is finitely supported on $\S$, $\K(\mu)$ is in fact a polytope of finite vertices and faces. In this case, the minimisation to compute the principal radius~\eqref{eq:simple_r} can be limited to operators $C$ which are normal vectors of the proper faces of maximal dimension of $\K(\mu)$. In that way, to compute the critical radius~\eqref{eq:simple_r} we only need to solve a linear program of finite size. 

In Ref.~\cite{chaupra}, the characterisation of vertices of such a polytope $\K(\mu)$ in the $4$D space is worked out in details. The technique boils down to take a direction, dictated by an operator $C$ and find the maximisers $\arg \max_{K \in \K(\mu)} \dprod{C}{K}$, which is simply a linear maximisation. It was shown that the maximisers are of the form
\begin{equation}
K^{\ast}= \int_\S \d \mu (\sigma) [\chi_{\dprod{C}{\sigma} > 0} (\sigma) + \chi_{\dprod{C}{\sigma} = 0}(\sigma) g(\sigma)] \sigma, 
\label{eq:optimise_maximiser}
\end{equation}
where $\chi_X(\sigma)$ is the characteristic function of the set $X$ and $g(\sigma)$ is an arbitrary function with values between $0$ and $1$. This has a simple interpretation: the maximisers are the sum of all members of the local hidden states which have positive projections on $C$ indicated by $\chi_{\dprod{C}{\sigma} > 0}$, and  members that are orthogonal to $C$ indicated by  $\chi_{\dprod{C}{\sigma} = 0}$ does not change the maximum. 
The maximal value is
\begin{equation}
\dprod{C}{K^{\ast}} = \int_\S \d \mu (\sigma) \max \{ \dprod{C}{\sigma},0 \},
\label{eq:optimise_maximum}
\end{equation}
independent of $g (\sigma)$.
This independence of the maximal value upon $g(\sigma)$ in the maximiser~\eqref{eq:optimise_maximiser} tells that the maximisation problem $\max_{K \in \K(\mu)} \dprod{C}{K}$ may have multiple maximisers, characterised by $g(\sigma)$. These maximisers form faces of $\K(\mu)$. Certainly, if the plane $\dprod{C}{\sigma} = 0$ does not go through any vertex of $\K(\mu)$, $g(\sigma)$ does not contribute to the maximiser~\eqref{eq:optimise_maximiser} and the maximiser is unique. On the other hand, if the plane $\dprod{C}{\sigma} = 0$ goes though a vertex, the function $g(\sigma)$ can be adjusted such that this point is included in the whole integral~\eqref{eq:optimise_maximiser} or not, giving $2^1=2$ independent extreme points of $\K(\mu)$, which form a line segment of maximisers. We are interested in the case where the plane $\dprod{C}{\sigma} = 0$ goes through (at least) $3$ points, where~\eqref{eq:optimise_maximiser} gives $2^3=8$ different extreme points of $\K(\mu)$ forming a proper face of $\K(\mu)$ with 
maximal dimension (i.e., of dimension $3$). This argument leads to a correspondence 
between planes that go through $3$ points of $\S$ and proper faces of maximal dimension of $\K(\mu)$. The correspondence actually goes much further: suppose $c_0$ and $\c$ are the offset and the normal vector of a plane that goes though $3$ points of $\S$, then $(c_0,\c)$ is the $4$D normal vector of a maximal face of $\K(\mu)$. \new{Crucially, one also finds that the set of $4$D normal vectors of the maximal faces of $\K(\mu)$ depends only on the chosen polytope, and not on the probability measure $\mu$.}

If $\S$ consists of $N$ vertices, then the linear program~\eqref{eq:critical_r_polytope} has $N$ variables (apart from some slack variables). There are $N(N-1)(N-2)/6$ planes that go through three points in $\S$. Together with the constraints on the positivity of the probability weights, we have $N(N-1)(N-2)/6+N$ inequality constraints. In addition, the minimum requirement gives $4$ equality constraints. Over all, we have a linear program of $\mathcal{O} (N)$ variables and $\mathcal{O}(N^3)$ constraints. 

%-----------------------------------------------------------------------------------
\subsection{Implication of the symmetry of the state}
When the state has some symmetry, we can exploit the symmetry to simplify the optimisation as well. Theorem~\ref{th:symmetry_LHS} implies that if a state $\varrho$ has symmetry group $\G$, then one can assume that the optimal LHS ensemble $\mu$ is symmetric under $\G$. When the Bloch sphere is approximated by a polytope $\S$, not only the symmetry of the state $\varrho$ matters, but also does the symmetry of $\S$. We have the restricted symmetry theorem. 

\begin{proposition}[Restricted symmetry of LHS ensemble]
If for a compact group $\G$, $\varrho$ is $(\G,U,V)$ symmetric, the polytope $\S$ is $(\G,V)$ symmetric, then there exists an optimal ensemble $\mu^{\ast}$ on $\S$ which is $(\G,V)$-invariant, $R_V(g) [\mu^{\ast}]=\mu^{\ast}$.
\label{th:symmetry_LHS_restricted}
\end{proposition}
\begin{proof}
The proof is actually rather the same as that of Theorem~\ref{th:symmetry_LHS}. Here the fact that $\S$ is symmetric under $\G$ just ensures the consistency of the proof: the space of probability measures on $\S$ is also  symmetric under $\G$.  
\end{proof}
%-----------------------------------------------------------------------------------
%\subsection{Technical aspects}
%-----------------------------------------------------------------------------------
\subsection{Remark on the computation for degenerate states}
\label{sec:computation_degenerate}
Note that even if the (canonical) states are close to degenerate, our procedure is still valid. Truly degenerate canonical states are not of the main interest, we nevertheless sketch how to cope with them for completeness. 

For truly degenerate state, strictly Proposition~\ref{th:principal_lower_semicontinuity} on the continuity of the principal radius in principle may not apply, and thus neither does Theorem~\ref{th:nested_bound}. We are back at the tedious problem of demonstrating the continuity of the principal radius with respect to LHS ensemble $\mu$ for degenerate states. For the practical purpose, there is a work around, though. The idea is that for degenerate states, both $\mu$ and $C$ can be subjected to some restrictive constraints. Under these restrictive constraints, Proposition~\ref{th:principal_lower_semicontinuity} and Theorem~\ref{th:nested_bound} regain their validity. 

Suppose $T$ is degenerate and consider the faction function~\eqref{eq:fraction_function},
\begin{equation}
F[\varrho,\mu,C]=\frac{\int \d \mu (\vv)  |c_0 + \c^T \vv|}{\norm{c_0 \a +  T \c}}.
\end{equation}
To be concrete, we can assume $T=\operatorname{diag}(s_1,s_2,0)$ without loss of generality. Note that now the state is invariant under the time-reversal transformation implemented by the reflection along $z$ on Bob's space (by the way, this implies that they are separable). Therefore the LHS ensemble can be assumed to be symmetric under reflection along $z$. Further, $\c$ can then be limited to be orthogonal to the kernel of $T$, i.e., in the $xy$-plane. One can easily verify that Proposition~\ref{th:principal_lower_semicontinuity} is again valid if $\mu$ is limited to those that are symmetric under $z$-reflection and $\c$ is on the $xy$-plane. As a result, Theorem~\ref{th:nested_bound} applies and the numerical procedure is valid. In fact, one can go on to show that $\mu$ can be assumed to be supported on the $xy$-plane, which largely simplifies the practical computation. Thus, while degenerate states seem theoretically complicated, practically they are in fact easier to work with.  The case where $T$ is 
rank-
$1$ can be worked out similarly. 
%-----------------------------------------------------------------------------------
\subsection{Technical aspects and the EPR-package}
\label{sec:EPR_package}
On the technical side, for a generic state, the vertices of the polytopes used to approximate the Bloch sphere were  chosen to be the solutions of the so-called ``covering problem" taken from 
Ref.~\cite{icosapoints}. These arrangements of points $\pmb{t}_i$ on the unit sphere have 
icosahedral symmetry and are arranged such that 
$\max_{i} \min_{j\neq i} \norm{ \pmb{t}_i - \pmb{t}_j}$ is minimal. This gives directly
an inner polytope. For the outer polytope, we used a rescaled version of the 
inner polytope by the inverse of its inscribed radius. The 
corresponding upper bound on $R(\varrho_{AB})$ was calculated independently, i.e., without using the estimate $R_{\rm in} \leq R(\varrho_{AB}) \leq R_{\rm in}/r_{\rm in}$,
as this gives typically a better bound.

The linear system of equations with constraints derived from the minimal requirement, 
the probability bounds, and the target function in  equation~\eqref{eq:simple_r}
were  written to a file in a linear program format which was then read and solved employing IBM ILOG 
CPLEX Optimization Studio. With a $92$-vertex polytope, the computational time for a generic state is about $30$ 
seconds on average when running on an Intel Xeon X5650 ($2.67$GHz) processor with six
physical cores.  The inscribed radius of the $92$-vertex polytope is 
$r_{\mathrm{in}} \approx 0.972$, thus the relative error is at most $2.8\%$.
When highly accurate values of the critical radius are desired, one can use a $252$-vertex polytope with $r_{\mathrm{in}} \approx 0.990$, which requires about $40$ minutes computation and about $48$GB memory.

If the states have
axial symmetry, assumed to be in $z$-axis, the best choice for the inner polytope is to impose
some subgroup of the rotation group around a fixed axis. To this end, one can
choose polytopes with vertices formed by intersecting $p$ circles
of latitude with $q=2p+2$ uniformly distributed great circles that go
through the north and south poles. The circles of latitude can be arranged in
various schemes: (i) the  polar angles are uniformly distributed, (ii) the
intersections with the symmetry axis are uniformly distributed, (iii) the
intersections with the symmetry axis are identical to the Gauss-Legendre
abscissae order $p$. We found that when $p$ and $q$ are fixed, the last scheme
gives the largest value for the inscribed radius $r_{\mathrm{in}}$. Thus, we chose this scheme in our calculations.

In this case, the polytope has the axial rotation group of degree $q$. By
Theorem~\ref{th:symmetry_LHS_restricted} on the restricted symmetry of the
optimal LHS ensemble, we can assume that the LHS ensembles are the same for
points with the same latitude. This reduces the number of variables in the
linear program by a factor of $q$. Moreover, many directions normal to
polytope facets transform to each other under the axial rotation group of degree
$q$. Thus the number of constraints also decreases by certain factor of order
$q$. Consequently, the size of the linear program in computing $R$
decreases significantly and the upper bound and lower bound for $R$ can be
obtained using polytopes with much higher number of vertices. In our
computation, we were able to use $q=52$ and $p=25$ which results in a
polytope with $1032$ vertices and $r_{\mathrm{in}} \approx 0.996$.

%(about $5$ times of that for the case of $162$-vertex polytope). 

Our programs and examples are available at \texttt{https://gitlab.com/cn611340/epr-steering}.
%-----------------------------------------------------------------------------------
%\subsubsection{Asymptotic analysis and extrapolation}

%To study the convergence of the algorithm, we run the algorithm for $T$-states with $T=\operatorname{diag}(-s,-s,-t)$. The obtained boundary is presented in Figure~\ref{fig:tstate} together with the exact boundary known from \red{[chau \& sania stuffs]}. The obtained accuracy is clearly better than that of \red{[Fillettaz]}.
% 
%\begin{figure}
%\includegraphics[width=0.5\textwidth]{tstates.pdf}
%\caption{Phase diagram of $z$-states. yellow: steerable states; lines of different colors corresponding to boundary obtained with different number of discrete points.}
%\label{fig:tstate}
%\end{figure}

%====================================================================================
\section{Gradients of the critical radius}

It is generally tedious to show that the critical radius is differentiable. However, at certain points such as $T$-states, it is plausible that the critical radius function is reasonably smooth. At these points, we can assume that the gradient exists.  Gradients of the critical radius are normal vectors of the supporting hyperplanes of its level sets. They are therefore directly related to optimal steering inequalities. 

We recall that any normal state can be brought to the canonical form~\eqref{eq:canonical_theta} by the group action of $\U(2) \times \GL(2)$. The action also allows one to relate  the supporting hyperplane at a normal state with that at its canonical form and vice versa. Excluding abnormal states, we can therefore restrict ourselves to computing the gradients of the critical radius at canonical states.

Take a state $\varrho=(\a,\operatorname{diag}{\s})$ in the canonical form, and let us assume that the gradient exists. Again, because of the invariance of $R$ with respect to the action of $\U(2) \times \GL(2)$, we know that the gradient $\nabla R [\varrho]$ has to be orthogonal to all the flowing directions of the action of $\U(2) \times \GL(2)$. More precisely, we have the following lemma.
\begin{lemma}
For any state $\varrho$, we have
\begin{enumerate}
\item[(i)] For any traceless hermitian operator $H$, 
\begin{equation}
\dprod{\nabla R [\varrho]}{ [H \otimes \II_B,\varrho]} = 0.
\label{eq:H_orthogonal}
\end{equation}
\item[(ii)] For any (complex) operator $M$, if we denote
$K=(\II_A \otimes M) \varrho - \varrho (\II_A \otimes M^{\dagger})- \varrho \Tr(M \varrho_B - \varrho_B M^{\dagger})$,
then
\begin{equation}
\dprod{\nabla R [\varrho]}{K}=0.
\label{eq:M_orthogonal}
\end{equation}
\end{enumerate}
\label{lem:invariant_directions}
\end{lemma}
\begin{proof}
The equalities are obtained by considering the infinitesimal action (i.e., the Lie algebra action) associated to the action of $\U(2) \times \GL(2)$. 

(i) Recall that the Lie algebra of $\U(2)$ are traceless, skewed hermitian operators $\mathfrak{su}(2)$. For any traceless hermitian operator $H \in i \mathfrak{su}(2)$, $U(t)=e^{-itH}$ with $t$ in some neighbourhood of $0$ is an element of $\U(2)$ near the identity. Since $R$ is invariant under $\U(2)$, we have $\frac{\d}{\d t} R[U(t) \otimes \II_B \varrho U^{\dagger}(t) \otimes \II_B]=0$. Computing the derivative explicitly results in~\eqref{eq:H_orthogonal}. 

(ii) Similarly, for any operator $M \in \mathfrak{gl}(2,\CC)$, which consists of all complex matrices, we have $V(t)=e^{-itM}$ is an element of $\GL(2)$ near the identity. Since $R$ is invariant under $\GL(2)$, we have $\frac{\d}{\d t} R[\II_A \otimes V(t) \varrho \II_A \otimes V^{\dagger} (t)/\Tr(V(t) \varrho_B V^{\dagger}(t))]=0$. Computing the derivative explicitly results in~\eqref{eq:M_orthogonal}.
\end{proof}

%As we mentioned, because of the invariance of $R$ with respect to $\U(2) \times \GL(2)$, any state can be brought into the canonical form of~\eqref{eq:quotient_theta}. Reversely, we can start with a state in the canonical form, compute the gradient, then bring to any other state by the action $\U(2) \times \GL(2)$. To this end, consider a state $\varrho_0$ with
%\begin{equation}
%\Theta_{\varrho_0}= \begin{pmatrix}1 & 0 \\ \a & \operatorname{diag}(\s) \end{pmatrix}.
%\end{equation}
Now suppose we can compute the derivatives of $R$ with respect to the canonical parameters $\a$ and $\s$, or equivalently, we know $\dprod{\nabla R [\varrho]}{\sigma_i^A \otimes \II_B}$ and $\dprod{\nabla R [\varrho]}{\sigma_i^A \otimes \sigma_i^B}$ for $i=1,2,3$. We then should be able to incorporate the invariant directions in Lemma~\ref{lem:invariant_directions} to find $\nabla R[\varrho]$ explicitly. For $T$-states, all these can be computed explicitly.
\begin{lemma}
For a non-degenerate $T$-state $(0,T)$ with $T=\operatorname{diag} (\s)$, we have:
\begin{align}
\dprod{\nabla R [(0,T)]}{\sigma_i^A \otimes \II_B} & = 0, \label{eq:nabla_tstates_a} \\
\dprod{\nabla R [(0,T)]}{\sigma_i^A \otimes \sigma_i^B} & = F_i(\s),
\label{eq:nabla_tstates}
\end{align} 
for $i=1,2,3$, where
\begin{equation}
F_i (\s) = 2 \pi  \frac{\partial (\abs{s_1s_2s_3} N_T)}{\partial s_i}, 
\end{equation}
with $N_T$ defined in equation~\eqref{eq:nt_tstates}.
\label{lem:nabla_tstates}
\end{lemma}
\begin{proof}
(i) The first identity expresses the fact that at $T$-states, the derivative of the critical radius with respect to Alice's reduced state vanishes. Indeed, we know that for states in the canonical form $\varrho=(\a,T)$, due to the time-reversal symmetry, we have $R[(\a,T)]=R[(-\a,T)]$. Thus the derivative of $R$ with respect to $\a$ must vanish at $\a=0$.

(ii) The second identity is obtained by directly differentiating the critical radius of $T$ states in equation~\eqref{eq:t_states} with respect to $\s$.
 
\end{proof}
%\red{[Some normalisation factors may need to be added ($\sigma_i \otimes \sigma_i$ are not normalised). To be checked!]}

Let us reconstruct the gradient at the $T$-states explicitly. One can compute the invariance directions dictated by Lemma~\ref{lem:invariant_directions} directly. This computation is further simplified by noting that $\rho_B=\frac{\II_B}{2}$ for $T$-states.
To find out these directions, in equation~\eqref{eq:H_orthogonal}, we choose $H=\sigma_k^A$ with $k=1,2,3$ and in equation~\eqref{eq:M_orthogonal}, we choose $M=\sigma_k^B$ and $M=i \sigma_k^B$ with $k=1,2,3$. This gives us $9$ directions in which the gradient vanishes. Further more, when \emph{incorporating} the fact that $\dprod{\nabla R [(\a,T)]}{\sigma_i^A \otimes \II_B} = 0 $ for all $i$, equation~\eqref{eq:nabla_tstates_a}, we come to the conclusion that $\nabla R [(\a,T)]$ vanishes in all directions $\sigma_i^A \otimes \sigma_j^B$ for $i \ne j$. Therefore we have, for $T$-states,
\begin{equation}
\nabla R [(0,T)]= \frac{1}{16}\sum_{i=1}^{3} F_i (\s) \sigma_i^A \otimes \sigma_i^B.
\end{equation}
Here the prefactor $\frac{1}{16}$ is due to the normalisation for the vectors $\sigma_i^A \otimes \sigma_i^B$. 
Note that the expression is symmetric in two parties, as a result, the gradients for the critical radius of steering from $A$ to $B$ and from $B$ to $A$ share the same gradients at $T$-states. This is rather surprising given the asymmetry in the definition of quantum steering with respect to the two parties. This surprising fact is again deeply rooted in the hidden symmetry of the critical radius under time-reversal transformation, which results in the first condition in~\eqref{eq:nabla_tstates_a}. 

Under the light of the relationship between gradients and supporting hyperplanes of level sets, these gradients of the critical radius certainly result in optimal steering inequalities. Despite the fact that we can actually compute the critical radius $R$ and gives various bounds for it, these steering inequalities may still be useful for proving steerability in experiments when the full tomography of the state is not available. 

%====================================================================================
\section{Details of the examples}
%----------------------------------------------------------------------------------
\subsection{Random cross-sections}
To construct a $2$D random cross-section, we choose two random states and 
construct the $2$D plane that goes through the maximally mixed state,
$\frac{\II_A}{2} \otimes \frac{\II_B}{2}$, and these 
two random states. The boundary of the set of unsteerable states was obtained
by solving the equation $R[\varrho]=1$ along $200$ rays going through
$\frac{\II_A}{2} \otimes \frac{\II_B}{2}$.  The equation was solved numerically employing the standard
bisection method. The polytopes that are used to approximate the Bloch sphere and the detailed implementation are described in Section~\ref{sec:EPR_package}. 

%In the computation, vertices of inner polytopes used to approximate the Bloch
%sphere were chosen to be the solutions of the so-called ``covering problem''
%which can be stated as: \textit{how one should place $n$ points on the unit
%sphere so as to minimize the maximal distance of any point on the sphere to its
%nearest neighbour?}. Arrangements of points with icosahedral symmetry 
%employed in our calculations were taken from Ref.~\cite{icosapoints}. The linear
%system with constraints derived from the minimal requirement, the probability
%bounds, and the definition of the critical radius were 
%written to a file in LP format which was then read and solved employing IBM ILOG 
%CPLEX Optimization Studio. With a $162$-vertex polytope, the computational 
%time for a generic state takes about $500$ seconds on average when running 
%on an Intel Xeon X5650 (2.67GHz) processor with $6$ physical cores and $48$GB
%RAM. The inscribed radius of the $162$-vertex polytope is $r_{\mathrm{in.}}
%\approx 0.984$, thus the relative error is at most $1.6\%$. This upper
%bound of the relative error can be reduced to $1\%$ by using $252$-vertex
%polytope that requires a favourable increase in computational time (about $5$
%times of that for the case of $162$-vertex polytope). It is worth to note that
%in fact the upper bounds obtained using rescaled afore-mentioned polytopes in
%our calculations are lower than these universal upper bounds. 

To illustrate the results of our
calculations, we selected, in an arbitrary manner, three examples of $2$D random cross-sections
as presented in Figure~\ref{fig:sm_cross_sections} where the steerability of
states in the narrow gray regions is uncertain due to the numerical accuracy. \new{Here the computation was performed using the $162$-vertex polytope and the sets of unsteerable states are extended beyond proper states. The first two examples in this figure have been presented in Figure 5 of the main text, where the $252$-vertex polytope was used, and only proper states were considered.}

%\begin{figure}[t!!]
%\includegraphics[width=0.20\textwidth]{{rd_boundary_20_new_crop}.pdf}
%\hspace{10pt}
%\includegraphics[width=0.20\textwidth]{{rd_boundary_100_new_crop}.pdf}
%\hspace{10pt}
%\includegraphics[width=0.40\textwidth]{{dpi600}.png}
%\includegraphics[width=0.40\textwidth]{{dpi600}.pdf}
%\begin{tabular}{cc}
%\begin{minipage}[c]{0.13\textwidth}
%\begin{center}
%\includegraphics[width=\textwidth]{{legends}.pdf}
%\end{center}
%\end{minipage}&
%\begin{minipage}[c]{0.30\textwidth}
%\begin{center}
%\includegraphics[width=0.48\textwidth]{{rd_boundary_20}.pdf}
%\includegraphics[width=0.48\textwidth]{{sections_rd_020}.pdf}
% \hspace{0.002\textwidth}
%\includegraphics[width=0.48\textwidth]{{rd_boundary_100}.pdf}
%\includegraphics[width=0.48\textwidth]{{sections_rd_100}.pdf}
%\end{center}
%\end{minipage}\\
%{ } & { } \\
%\begin{minipage}[c]{0.16\textwidth}
%\begin{center}
%\includegraphics[width=\textwidth]{{sections_symmetric}.pdf}
%\end{center}
%\end{minipage}&
%\begin{minipage}[c]{0.28\textwidth}
%\begin{center}
%\includegraphics[width=\textwidth]{{theta_state}.pdf}
%\end{center}
%\end{minipage}\\
%\end{tabular}
%\end{figure}
\begin{figure}
\begin{center}
\includegraphics[width=0.21\textwidth]{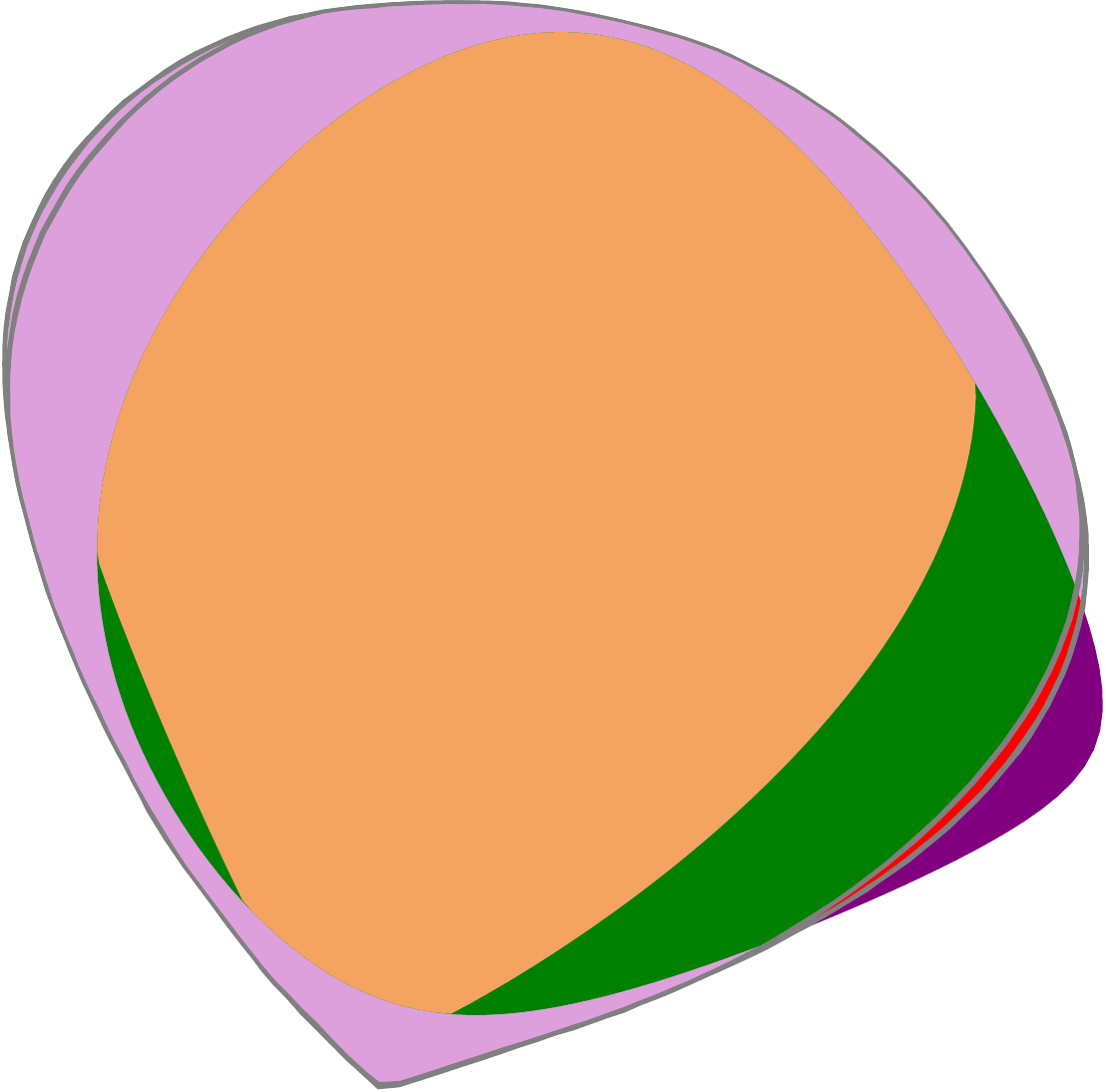}
\hspace{0.3cm}
\includegraphics[width=0.21\textwidth]{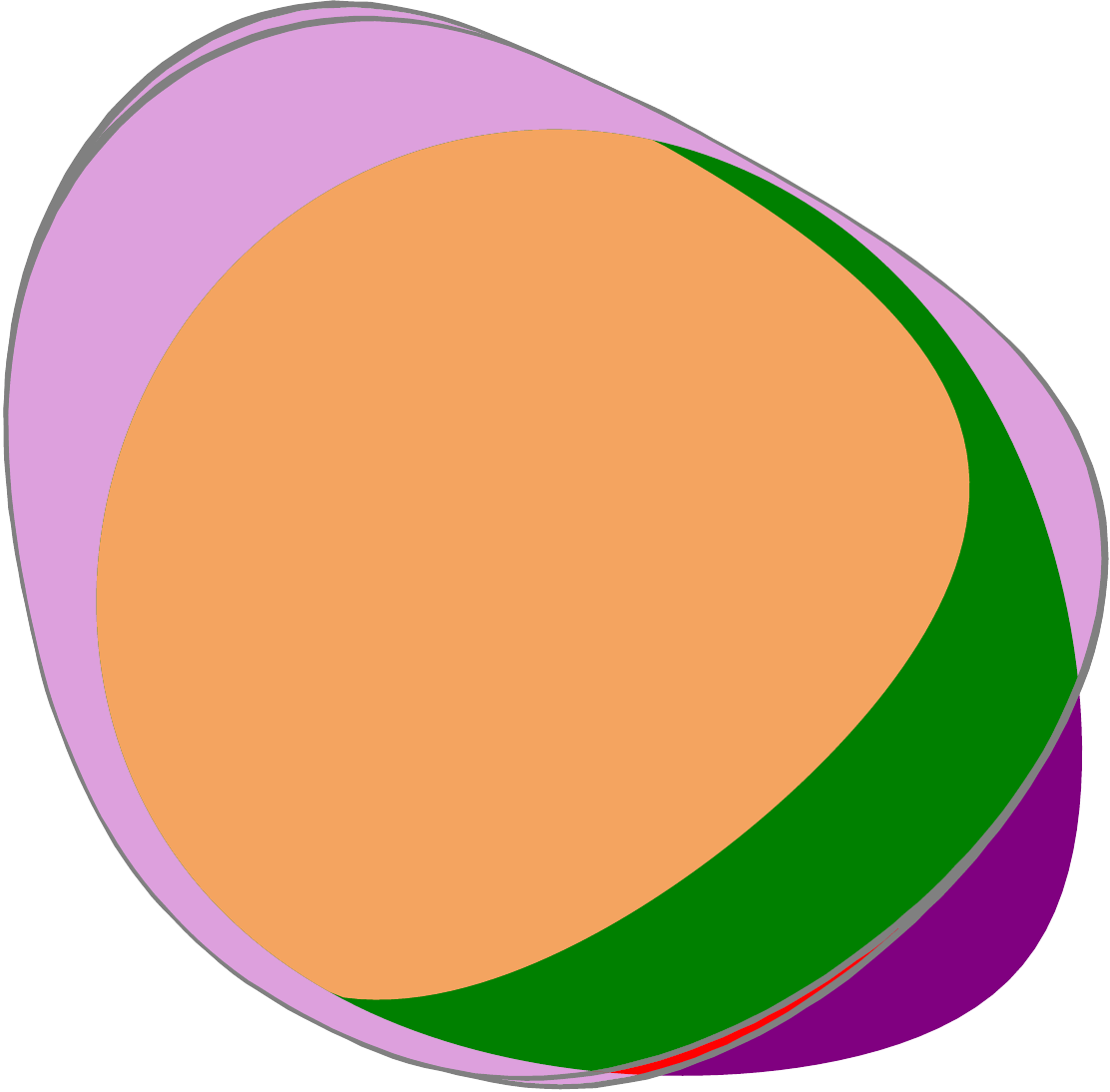}
\hspace{0.3cm}
\includegraphics[width=0.21\textwidth]{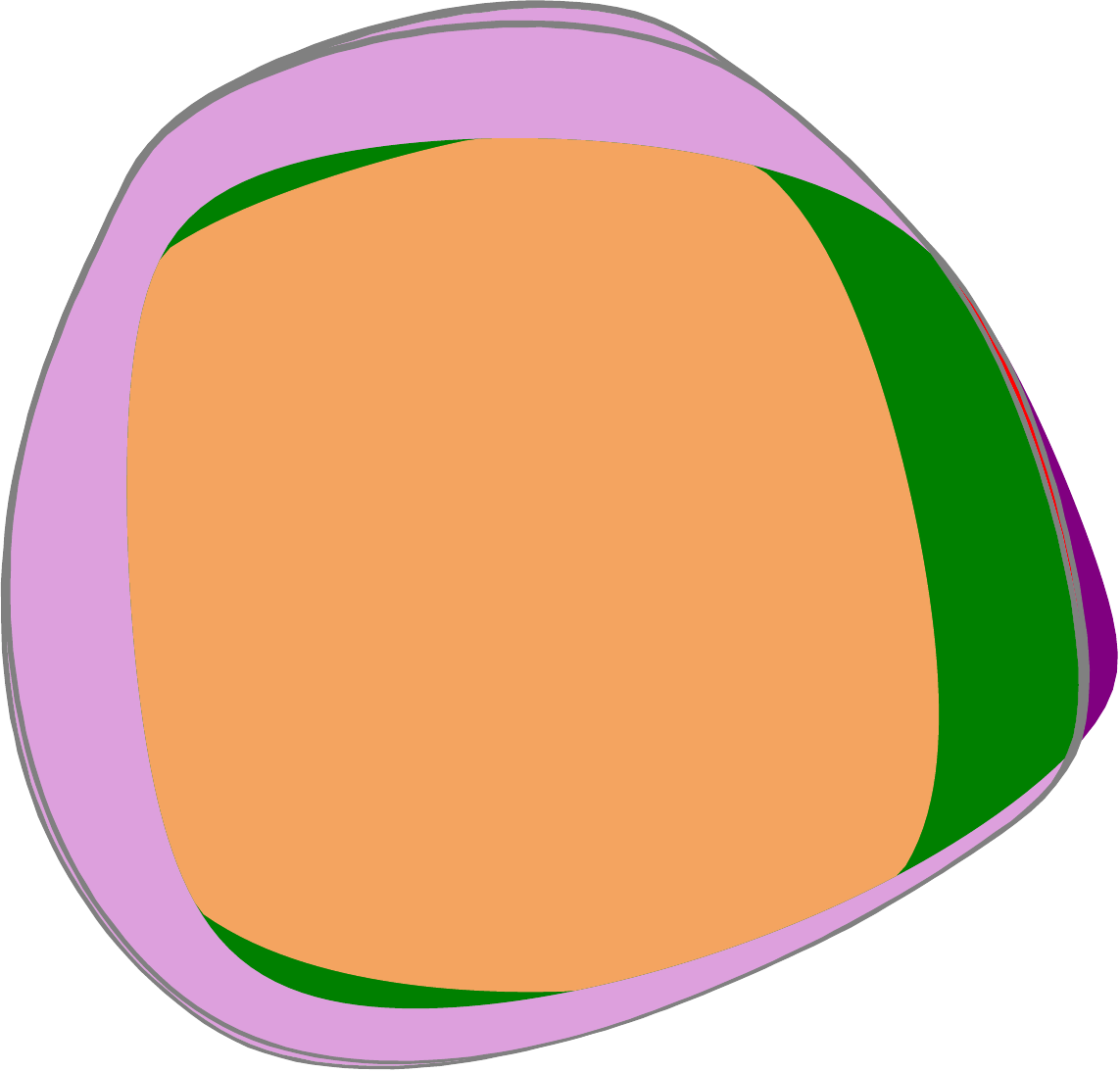}
\includegraphics[width=0.21\textwidth]{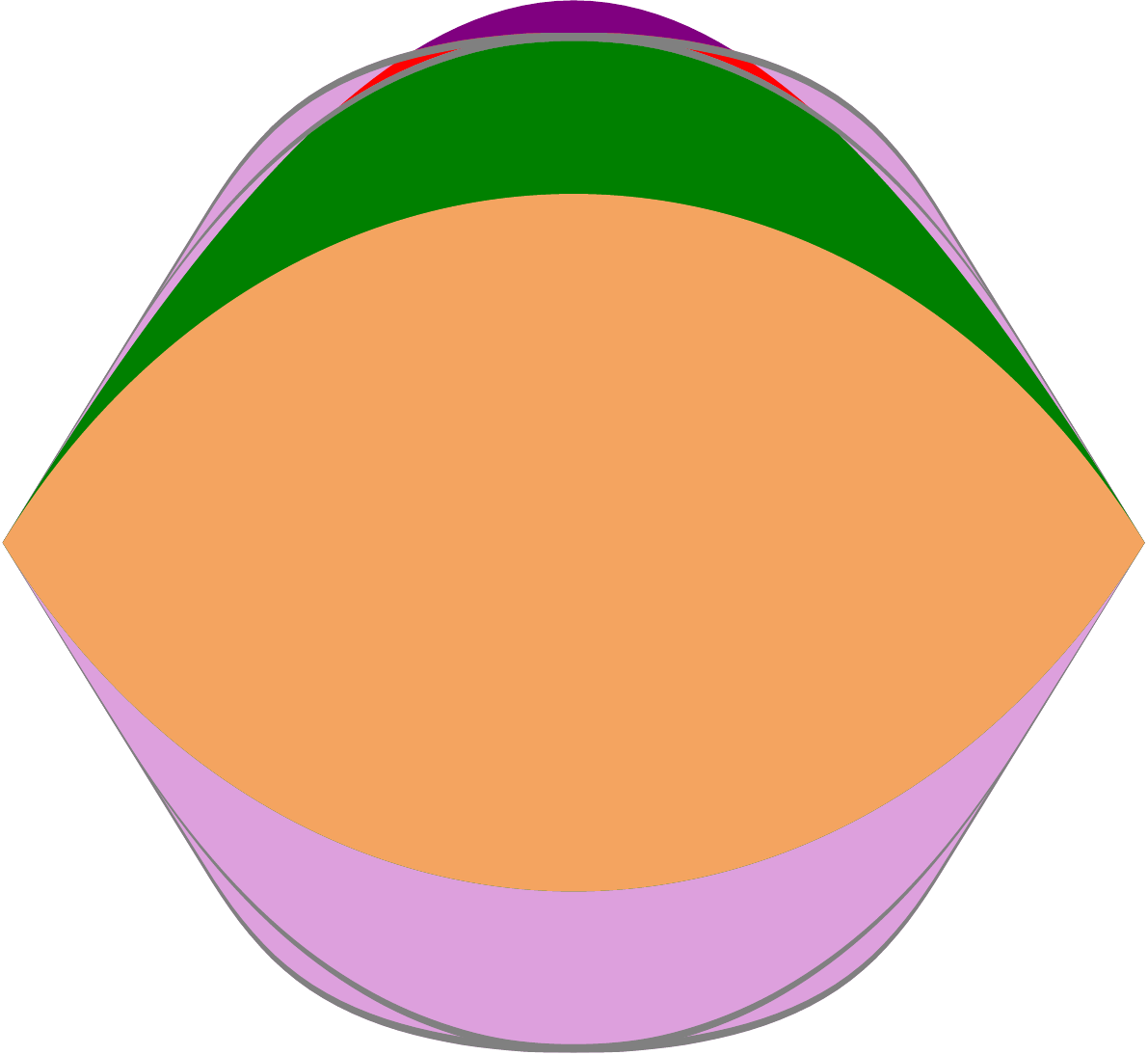}
\\{\ }\\
\includegraphics[width=0.5\textwidth]{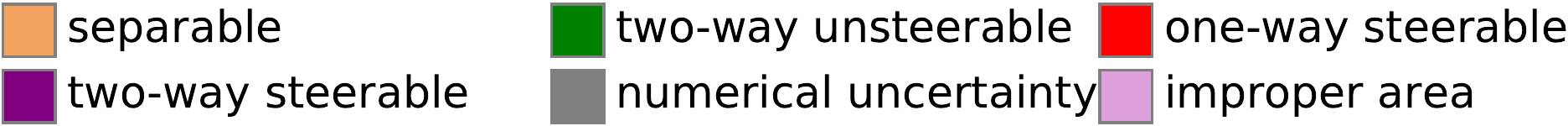}
\end{center}
\caption{ Top row and lower left: random $2$D cross-sections of the set 
of two-way steerable states, one-way steerable states, two-way unsteerable states, 
and separable states. Differently from the main text, here the set of
states with $R[\varrho] \le 1$ is extended beyond the set of proper states (to
include improper states). Lower right: a $2$D symmetric cross-sections showing 
the set of unsteerable states, including improper ones, that is symmetric under
the the time-reversal transformations implemented by the reflections and
inversions as described in the text.}
\label{fig:sm_cross_sections}
\end{figure}
%----------------------------------------------------------------------------------
\subsection{Symmetric cross-sections}
\new{To illustrate the symmetry of the critical radius under the local time-reversal transformations, we choose a cross-section cut by a $2$D plane which is invariant under the local time-reversal transformations.} Such a symmetric cross-section is illustrated with states in the canonical form for steering from $A$ to $B$. Two random states are chosen to be $(\a,0)$ and $(0,\operatorname{diag}(\s))$. All states in the cross-section are of the form $x(\a,0)+y(0,\operatorname{diag} (\s))=(x \a, y \operatorname{diag} (\s))$. On this plane, the time-reversal transformation on Alice's side (upto local unitary transformations) is implemented by inversion of $(x,y)$, while the time-reversal transformation on Bob's side is implemented by inversion of $y$.

\new{
Note that this cross-section contains the whole scaling lines, namely,
$R[(\a,\operatorname{\s})]=\lambda R[(\lambda \a, \lambda \operatorname{\s})]$.
This scaling relation of the critical radius provides a powerful
tool for determining the boundary of the set of unsteerable states. One no longer needs to solve
equation $R[\varrho]=1$ by the bisection
method. Instead, one simply computes the upper bound and lower bound for
$R[(\a,\operatorname{\s})]$ on a closed loop---here chosen to be the unit circle---and then uses the scaling relation to locate the boundary of the
set of unsteerable states.

For steering from $B$ to $A$, we note that the special structure
of the canonical form for steering from $A$ to $B$ allows for a slightly
different scaling relation for steering from $B$ to $A$, namely $R^{B \to A}
[(\a,\operatorname{diag} (\s))]= \lambda R^{B \to A} [(\a, \lambda
\operatorname{diag} (\s))]$ with $R^{B \to A}[\varrho]$ being the critical
radius for steering from $B$ to $A$. This scaling can also be employed to locate
the boundary of unsteerable states starting from the values of
$R$ on a closed loop as for the case of steering $A$ to $B$. Here the loop was
chosen to be the boundary of the set of separable states. The upper
bound and lower bound for $R^{B \to A}$ of states on this boundary were
determined by bringing the states to its canonical form for steering from $B$ to
$A$ and applying the same computation procedure as for steering from $A$ to $B$. An example of symmetric 2D cross-sections is also shown in the
bottom-right panel of Figure~\ref{fig:sm_cross_sections}.
}
%Note that this cross-section contains the whole scaling lines, namely, $R[(\a,\operatorname{\s})]=\lambda R[(\lambda \a, \lambda \operatorname{\s})]$. Therefore the bisection method is not necessary: here we see the power of the scaling of the critical radius. One simply computes  the upper bound and lower bound for $R[(\a,\operatorname{\s})]$ on the unit circle, then use this scaling relation to allocate the boundary of the set of unsteerable states. 

%For the critical radius for steering from $B$ to $A$, we denote $R^{B \to A} [\varrho]$. Note that we are still using the canonical form for steering from $A$ to $B$. In this parameterisation, we have a slightly different scaling relation for steering from $B$ to $A$, $R^{B \to A} [(\a,\operatorname{diag} (\s))]= \lambda R^{B \to A} [(\a, \lambda \operatorname{diag} (\s))]$. This scaling also allows one to allocate the boundary of unsteerable states starting from the values of $R$ on certain boundary--here chosen to be the boundary of unsteerable states. We compute the upper bound and lower bound for $R^{B \to A}$ on this curve by bringing the state to its canonical form for steering from $B$ to $A$ and apply the standard procedure as above. 
%\begin{figure}[hbt!]
%\begin{center}
%\includegraphics[width=0.3\textwidth]{{sm_sections_symmetric}.pdf}
%\end{center}
%\caption{}
%\label{fig:sm_symmetric_cross_sections}
%\end{figure}

%----------------------------------------------------------------------------------
\subsection{A family of one-way unsteerable states}
In this section we consider the state
\begin{equation}
\varrho= \alpha \ketbra{\theta}{\theta} + (1-\alpha) \varrho_A \otimes \frac{\II_B}{2}, 
\label{eq:theta_state}
\end{equation}
where $\ket{\theta}= \cos \frac{\theta}{2} \ket{00} + \sin \frac{\theta}{2}
\ket{11}$ with  $0\le \theta \le \frac{\pi}{4}$ and $0 \le \alpha \le 1$. This
state is important for demonstrating the one-way steering
phenomenon~\cite{Bowles2016a}. %Since $\ket{\theta}$ 
%is pure, the critical radii $R[\ketbra{\theta}{\theta}]=1/2$ for all $\theta\in (0,\frac{\pi}{4})$. 
As $\ket{\theta}$ is pure, it is easy to  show via the scaling relation
(\ref{eq:scaling}) that the state is steerable from $B$ to $A$ for $\alpha >
\frac{1}{2}$ and $\theta>0$. However, determining the boundary of unsteerable
states from $A$ to $B$ has  been proven to be difficult~\cite{brunnerlhsrecent}.
Here we show how this boundary can be obtained with high accuracy in
our approach. 

\begin{figure}[t!]
\includegraphics[width=0.45\textwidth]{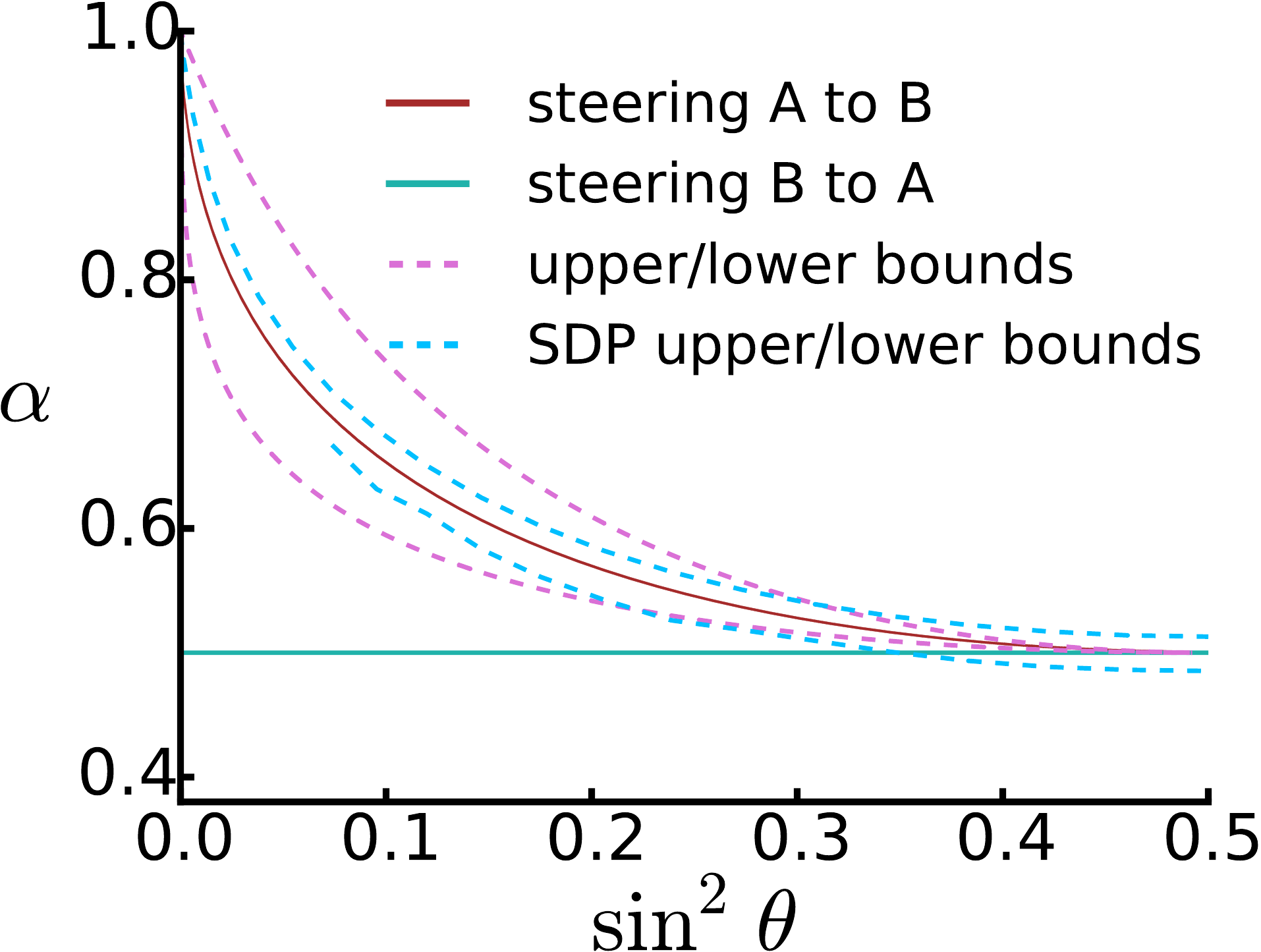}
\caption{{{Boundary of the set of one-way steerable states for the
family of states described in equation~\eqref{eq:theta_state}}. The thickness of
the line presenting the boundary between steerable/unsteerable states from $A$
to $B$ indicates the uncertainty due to numerical accuracy. The analytical upper
bound is obtained using equation~\eqref{eq:analytical_bounds}}. The analytical
lower bound is obtained using equation~\eqref{eq:uniform_ansatz_bound}, which gives the same lower bound as Figure 3 in Ref.~\cite{Bowles2016a}. The SDP
data are provided by the authors of Ref.~\cite{brunnerlhsrecent}.}
\label{fig:theta_state}
\end{figure}

For states of the form~\eqref{eq:theta_state}, the boundary for unsteerable
states is also obtained by solving the equation $R[\varrho]=1$
numerically using the bisection method, similar to the
case of random cross-sections. Note that in this case the state is axially symmetric around the $z$-axis. We therefore can use the symmetry to reduce the size of the linear program  and conveniently work with polytopes of $1032$ vertices; see Section~\ref{sec:EPR_package}. 

\new{In Figure~\ref{fig:theta_state}, we present the obtained border between unsteerable/steerable states together with certain analytical bounds and the known data from SDP~\cite{brunnerlhsrecent} for states of the form~\eqref{eq:theta_state}. One observes that with $q=52$ and $p=25$, we can obtain rather accurate description of the border. Note that the regime of one-way unsteerable states looks significantly exaggerated in comparison to Figure~\ref{fig:sm_cross_sections}; however here the parametrisation does not faithfully represent the Hilbert--Schmidt metric of the state space.}  

%====================================================================================
\blue{
\section{Generalisation to higher dimensional systems}
\label{sec:generalisation}
In this section, we assume that Alice and Bob share a state $\varrho$ of dimension $d_A \times d_B$.
An $n$-POVM implemented by Alice is  $E=\oplus_{i=1}^{n} E_i$, where $0 \le E_i \le \II_A$, $\sum_{i=1}^{n} E_i=\II_A$. Following Ref.~\cite{chaujpa}, a bipartite state $\varrho$ is unsteerable (from $A$ to $B$) with respect to $n$-POVMs if and only if there exists a LHS ensemble $\mu$ such that
\begin{equation}
\int \d \mu (\sigma) \max_i \{\dprod{Z_i}{\sigma}\} \ge \sum_{i=1}^{n} \Tr [\varrho (E_i \otimes Z_i)],
\label{eq:general_steering_inequality}
\end{equation}
for all $Z= \oplus_{i=1}^{n} Z_i$ and all POVMs $E=\oplus_{i=1}^{n} E_i$. This inequality has a simple interpretation. {Upon making a measurement $E$ on her side, Alice decomposes Bob's state into $n$ conditional states. Bob then makes $n$ different measurements to determine the expectation values of $n$ arbitrary observables $Z_i$, each for a conditional state, and then average them out over all conditional states. If the conditional states are simulated from an LHS ensemble $\mu$, this average clearly cannot exceed the left hand side of~\eqref{eq:general_steering_inequality}, where each of the state in the LHS ensemble is associated to the operator $Z_i$ that has the maximal mean value.}   

We then define the \emph{inverse} fraction function $F^{-1}[\varrho,\mu,Z,E]$ to be
\begin{equation}
\frac{\sum_{i=1}^{n} \Tr [\varrho (E_i \otimes Z_i) ] - \frac{1}{d_A}  \sum_{i=1}^{n} \Tr(E_i) \Tr (\varrho_B Z_i)}{\int \d \mu (\sigma) \max_i \{\dprod{Z_i}{\sigma}\} - \frac{1}{d_A} \sum_{i=1}^{n}  \Tr(E_i) \Tr (\varrho_B Z_i)},
\label{eq:fraction_function_general}
\end{equation}
with the \emph{numerator-dominated convention}, meaning, if the numerator vanishes, the function vanishes regardless of the denominator. The reason we define the inverse of the fraction function, instead of the function itself, is because the numerator of the inverse fraction function can be negative, while the denominator is non-negative. That the denominator is non-negative ensures that inequality~\eqref{eq:general_steering_inequality} holds if and only if $F^{-1}[\varrho,\mu,Z,E] \le 1$. Moreover, the offset subtracted from both the numerator and the denominator was chosen to enforce the scaling of the critical radius; see Section~\ref{sec:general_scaling} below.

The inverse principal radius $r^{-1}_n [\varrho,\mu]$ is defined as
\begin{equation}
r^{-1}_n [\varrho,\mu]= \sup_{Z,E} F^{-1}[\varrho,\mu,Z,E].
\label{eq:general_principal_radius}
\end{equation}
In difference from PVMs, the fraction function for POVMs can be negative. Yet, one can easily show that $r^{-1}_n [\varrho,\mu] \ge 0$. Then one can also write
\begin{equation}
r^{-1}_n [\varrho,\mu]= \sup_{Z,E} \max\{ F^{-1}[\varrho,\mu,Z,E],0 \}.
\end{equation}

Similar to Lemma~\ref{th:principal_concavity}, we can easily show that the inverse critical radius $r^{-1}_n [\varrho,\mu]$ is convex in $\mu$, since so is $\max\{ F^{-1}[\varrho,\mu,Z,E],0 \}$. Also, similar to Lemma~\ref{th:principal_upper_semiconitinuity}, for a fixed $\varrho$, $r^{-1}_n [\varrho,\mu]$ is weakly lower-semicontinuous with respect to $\mu$. Therefore $r^{-1}[\varrho,\mu]$ attains the minimum value for some $\mu^\ast$ (an optimal LHS ensemble). We define the inverse \emph{critical radius} to be
\begin{equation}
R^{-1}_n[\varrho]= \min_{\mu} r^{-1}_n[\varrho,\mu], 
\label{eq:general_crictical_radius}
\end{equation}
where $\mu$ is subject to \emph{minimal requirement},
\begin{equation}
\int \d \mu (\sigma) \sigma = \varrho_B.
\end{equation}
Then the state $\varrho$ is unsteerable if and only if $R_n[\varrho] \ge 1$.

%----------------------------------------------------------------------------------
\subsection{Reducing to the formula for two-qubit states}
We first note that both the numerator and the denominator are invariant under transformation $Z_i \to Z_i - Y$ for arbitrary hermitian operator $Y$. Thus we can assume $\sum_{i=1}^{n} Z_i= 0$.  When restricted to $2$-POVMs, we can set $C=Z_1=-Z_2$. Further, for two-qubit systems, we can restrict from $2$-POVMs to PVMs, thus $E_1=Q$ with $E_2=\II_A-Q$ for some projection $Q$. We then have
\begin{equation}
r_2^{-1}[\varrho,\mu]= \sup_C \frac{2 \max_Q \Tr[(\varrho-\frac{\II_A}{2} \otimes \varrho_B) (Q \otimes C)]}{ \int \mu (\sigma) \abs{\dprod{C}{\sigma}}}. 
\label{eq:reduced_to_two_qubit}
\end{equation} 
Now note that $\Tr[(\varrho-\frac{\II_A}{2} \otimes \varrho_B) (Q \otimes C)]=\Tr\{\Tr_B[(\varrho-\frac{\II_A}{2} \otimes \varrho_B) (\II_A \otimes C)] Q \} $. Since $\Tr_B[(\varrho-\frac{\II_A}{2} \otimes \varrho_B) (\II_A \otimes C)]$ is a traceless operator, we have
\begin{align}
\max_Q\Tr\{\Tr_B[(\varrho-\frac{\II_A}{2} \otimes \varrho_B) (\II_A \otimes C)] Q \}= \nonumber \\ \frac{1}{\sqrt{2}}\norm{\Tr_B[(\varrho-\frac{\II_A}{2} \otimes \varrho_B) (\II_A \otimes C)]}.
\end{align}
This identifies~\eqref{eq:reduced_to_two_qubit} with the previous definition of the principal radius for two-qubit states~\eqref{eq:simple_r}.

%----------------------------------------------------------------------------------
\subsection{Remarks on other properties}
Many properties of the critical radius can be obtained easily by adapting the proofs for $2$-POVMs and the two-qubit system. This includes the scaling and the symmetry of the critical radius. As examples, we repeat these two statements and proofs.   
%==================================================================================== 
%\subsection{Implication of symmetry on the optimal LHS ensemble.}
%We say a state $\varrho$ is $(\G,U,V)$-symmetric with a compact group $G$ with its two actions $\mu$ on $\H_A$ and $V$ on $\H_B$ if $U^{\dagger}(g) \otimes V^{\dagger}(g) \varrho U(g) \otimes V(g)$ for all $g \in \G$. Recall that the action $V$ on $\H_B$ induces an action on the measures on $\S_B$, defined by $R_V(g)[\mu](X)= \mu[V(g) X V^{\dagger}(g)]$ for all measurable subset $X$ of $\S_B$. 
%\begin{theorem}[Symmetry of LHS ensemble]
%If $\varrho$ is $(\G,U,V)$-symmetric, then there exists an optimal ensemble \mu^{\ast}$ which is $(\G,V)$-invariant, $R_V(g) [\mu^{\ast}] = \mu^{\ast}$.
%\end{theorem}
%\begin{proof}
%\blue{This should be a simple consequence of the concavity of $r_\mu(\varrho)$ in $\mu$.}
%\end{proof}

%====================================================================================
\subsubsection{Scaling of the critical radius}
\label{sec:general_scaling}
\begin{theorem}[Scaling of the critical radius]
For any state $\varrho$ and any $\lambda \ge 0$, we have
\begin{equation}
R^{-1}_n[\varrho]= \frac{1}{\lambda} R^{-1}_n[\lambda \varrho + (1- \lambda) \frac{\II_A}{d_A} \otimes \varrho_B].
\label{eq:scaling_general}
\end{equation} 
\label{th:scaling_general}
\end{theorem}
\begin{proof}
The proof is very simple. We first note that the numerator in the definition of $r^{-1}_n[\varrho,\mu]$ can be rewritten as $\sum_{i=1}^n \Tr [\varrho E_i \otimes Z_i] - \frac{1}{d_A} \sum_{i=1}^n E_i = \sum_{i=1}^n \Tr [(\varrho - \frac{\II_A}{d_A} \otimes \varrho_B) E_i \otimes Z_i]$. Then upon transforming $\varrho \to \lambda \varrho + (1- \lambda) \frac{\II_A}{d_A} \otimes \varrho_B$, this numerator gets a factor of $\lambda$ while the denominator is invariant.
\end{proof}

%\blue{We have to be careful at what happens at $\lambda=0$. Probably this tells $R[\frac{\I}{\d_A} \otimes \varrho_B]=\infty$. This is expected to be true for all product states.}

%==================================================================================== 
\subsubsection{Continuous symmetry of the critical radius}
The Bloch hyperplane $\P$ is the linear manifold of hermitian trace-$1$ operators acting on $\CC^{d_A} \otimes \CC^{d_B}$. For $U \in \U(d_A) $, $V \in \GL(d_B)$, consider the affine transformation from the Bloch hyperplane of the joint system into itself $\varphi_{(U,V)}:\P \to \P$, defined by 
\begin{equation}
\varphi_{(U,V)} (X) = \frac{(U \otimes V) X (U^\dagger \otimes V^\dagger)}{\Tr [(U \otimes V) X (U^\dagger \otimes V^\dagger)]}.
\end{equation}
for $X \in \P$. This is a group action of $\U(d_A) \times \GL(d_B)$ on $\P$. Note that $\varphi_{(U,V)}$ conserves the positivity, thus also maps the set of (bipartite) proper states into itself. 

\begin{theorem}[Continuous symmetry of the critical radius]
For any state $\varrho$ and $U \in \operatorname{U} (d_A)$, $V \in \operatorname{GL} (d_B)$, we have $R_n[\varrho] = R_n [\varphi_{(U,V)} \varrho]$. 
\label{th:invariant_general}
\end{theorem}

The proof of this theorem then goes very similarly to the proof of Theorem~\ref{th:invariant}, provided the following lemma is used instead of Lemma~\ref{lem:principal_invariance}. 

\begin{lemma}
Consider a given state $\varrho$, a given probability measure (LHS ensemble) $\mu$ satisfying the minimal requirement $\int \d \mu (\sigma) \sigma = \varrho_B$. For $U \in \U(d_A)$ and $V \in \operatorname{GL} (d_B)$, we denote $\tilde{\varrho}=\varphi_{(U,V)} (\varrho)$. Note that there exists a unique probability measure $\tilde{\mu}$ on $\B_B$ defined by  
\begin{equation}
\int \d \tilde{\mu} (\sigma) f(\sigma) = \frac{1}{\Tr (V \varrho_B V^{\dagger})} \int  \frac{\d \mu \circ \varphi_{V^{-1}} (\sigma) }{ \Tr [V^{-1} \sigma (V^{-1})^{\dagger}]} f(\sigma)
\end{equation}
for all continuous functions $f$. Then $\tilde{\mu}$ satisfies the minimal requirement for $\tilde{\varrho}$ and $r_n [\varrho,\mu] = r_n [\tilde{\varrho},\tilde{\mu}]$.
\label{lem:general_principal_invariance}
\end{lemma}

\begin{proof}
(i) The proof that  $\tilde{\mu}$ satisfies the minimal requirement for $\tilde{\varrho}$ goes exactly as the proof of Lemma~\ref{lem:principal_invariance}.

(ii) Now we prove that $r_n[\varrho,\mu] = r_n[\tilde{\varrho},\tilde{\mu}]$. Using the definition~\eqref{eq:general_principal_radius}, we have $r^{-1}_n[\tilde{\varrho},\tilde{\mu}]$ as
\begin{equation}
\sup_{Z,E} \frac{\sum_{i=1}^{n} \Tr [\tilde{\varrho} (E_i \otimes Z_i) ] - \frac{1}{d_A} \sum_{i=1}^{n}  \Tr(E_i) \Tr (\tilde{\varrho}_B Z_i)}{\int \d \tilde{\mu} (\sigma) \max_i \{\dprod{Z_i}{P}\} - \frac{1}{d_A} \sum_{i=1}^{n} \Tr(E_i)  \Tr (\tilde{\varrho}_B Z_i)}.
\label{eq:def_of_r_repeated}
\end{equation}
Now using the definition of $\tilde{\mu}$, we find $\int \d \tilde{\mu} (\sigma) \max_i \{\dprod{Z_i}{\sigma}\}$ to be
\begin{align}
\frac{1}{\Tr (V \varrho_B V^{\dagger})} \int  \frac{\d \mu \circ \varphi_{V^{-1}} (\sigma) }{ \Tr (V^{-1} \sigma (V^{-1})^{\dagger})} \max_i \{\dprod{Z_i}{\sigma}\}.
\end{align}
Upon making the transformation of variable $\sigma= \varphi_V(\tau)$, this becomes
\begin{align}
\frac{1}{\Tr (V \varrho_B V^{\dagger})} \int  \d \mu (\tau) \max_i \Tr (V^\dagger Z_i V \tau).
\end{align}
The denominator of the expression under the supremum in~\eqref{eq:def_of_r_repeated} can then be written as 
\begin{equation}
\int \d \mu (\sigma) \max_i \langle \tilde{Z}_i, \sigma \rangle - \frac{1}{d_A}\sum_{i=1}^{d_A}  \Tr(E_i) \Tr (\varrho_B \tilde{Z}_i),
\end{equation}
where $\tilde{Z}_i=V^\dagger Z_i V$.
Now using the definition of $\tilde{\varrho}$, the numerator can be written as
\begin{equation}
\sum_{i=1}^{n} \Tr [\varrho (\tilde{E}_i \otimes  \tilde{Z_i}) ] - \frac{1}{d_A}  \sum_{i=1}^{d_A} \Tr (\tilde{E}_i) \Tr (\varrho_B \tilde{Z}_i).
\end{equation}
where  $\tilde{E}_i = U^\dagger E_i U$. So the principal radius~\eqref{eq:def_of_r_repeated} can be written as
\begin{equation}
\sup_{\tilde{Z},\tilde{E}} \frac{\sum_{i=1}^{n} \Tr [\varrho (\tilde{E}_i \otimes  \tilde{Z_i}) ] - \frac{1}{d_A}  \sum_{i=1}^{d_A} \Tr (\tilde{E}_i) \Tr (\varrho_B \tilde{Z}_i)}{\int \d \mu (\sigma) \max_i \langle \tilde{Z}_i, \sigma \rangle - \frac{1}{d_A}\sum_{i=1}^{d_A}  \Tr(\tilde{E}_i) \Tr (\varrho_B \tilde{Z}_i)},
\end{equation}
where we have used $\Tr(E_i)=\Tr(\tilde{E}_i)$.
Since the set of $(\tilde{Z},\tilde{E})$ are the same as that of $(Z,E)$, this expression in fact coincides with the definition of $r^{-1}_n[\varrho,\mu]$. 
\end{proof}

\begin{figure}[t!]
\includegraphics[width=0.5\textwidth]{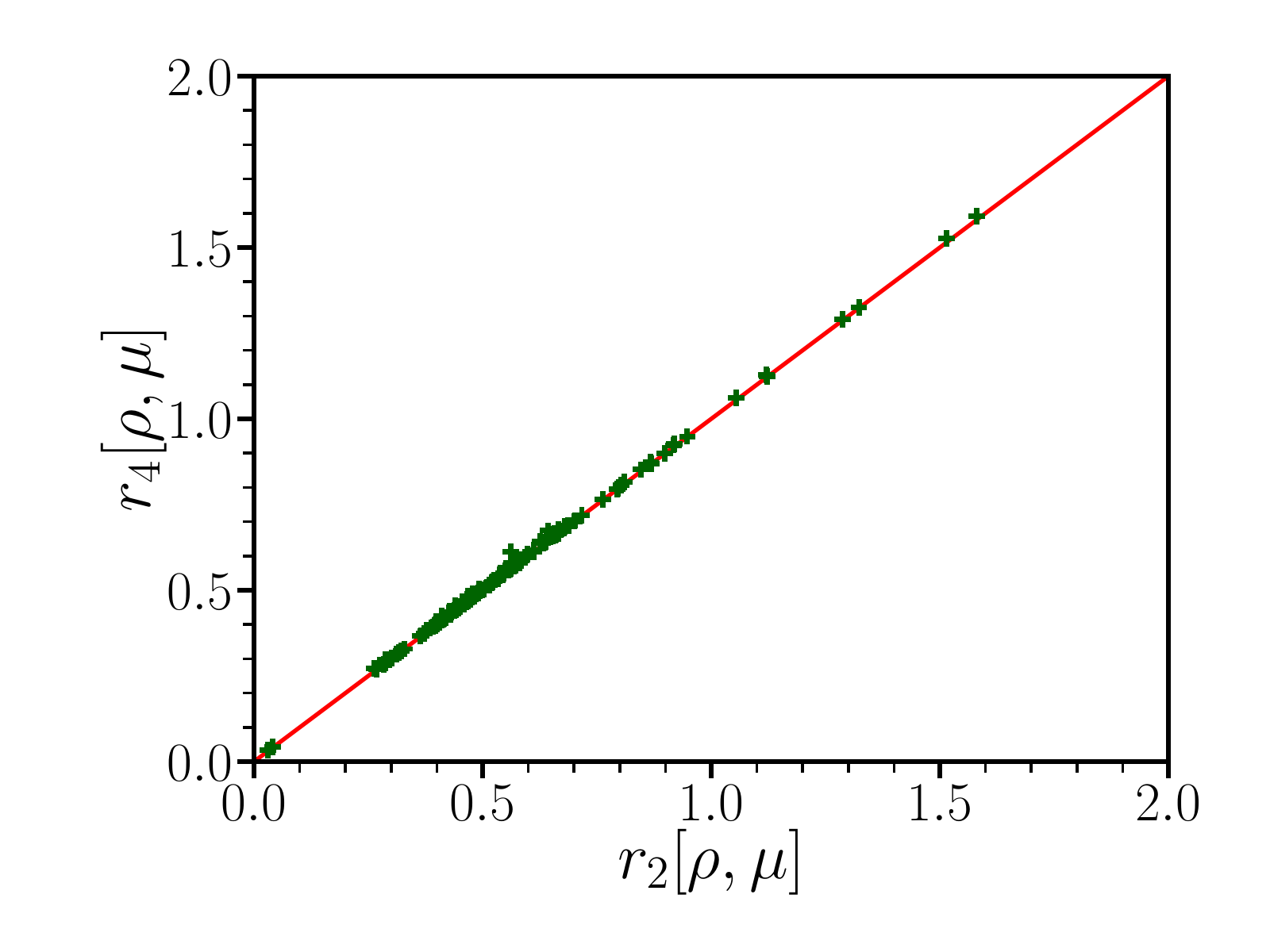}
\caption{The principal radii of random states and random LHS ensembles. \new{While the values of $r_2[\varrho,\mu]$ are computed exactly, $r_4[\varrho,\mu]$ is estimated by a simulated annealing algorithm. If POVMs and PVMs were inequivalent, one would expect $r_4[\varrho,\mu]$ to be strictly smaller $r_2[\varrho,\mu]$ for certain $\varrho$ and $\mu$. Surprisingly, upto the numerical accuracy, we find that the estimated $r_4[\varrho,\mu]$ is equal to $r_2[\varrho,\mu]$, which implies that POVMs and PVMs are equivalent in a strong sense.}}
\label{fig:POVM}
\end{figure}

%----------------------------------------------------------------------------------
\section{On the relation between PVMs and POVMs}
Armed with the newly defined concepts, we now discuss the question of the equivalence of different classes of measurements in quantum steering, in particular of PVMs and POVMs. Because POVMs of $n$ outcomes constitute a subset of POVMs of $n+1$ outcomes, we have a decreasing chain $r_2 [\varrho,\mu] \ge r_3   [\varrho,\mu] \ge \cdots$. As a consequence, the critical radii also form a decreasing chain $R_2[\varrho] \ge R_3[\varrho]\ge \cdots$. Since the extreme POVMs have at most $d_A^2$ non-empty outcomes, both of these two chains turn into equalities at $n=d_A^2$. We denote $R_{\mathrm{POVM}}[\varrho]=R_{d_A^2}[\varrho]$. Where does the critical radius for PVMs, here denoted  $R_{\mathrm{PVM}}[\varrho]$, fit into this chain? There has been a suspicion that $R_{d_A^2} [\varrho]= R_{\mathrm{PVM}}[\varrho]$, or POVMs and PVMs are equivalent in quantum steering. Until now, there has been no concrete evidence whether this conjecture is true except for certain special states~\cite{chaujpa,
Werner2014a}. 
%In fact, from private communications, we find that opinions vary widely among researchers. 

Here restricted to two-qubit states, we investigate a stronger hypothesis: $r_{\mathrm{POVM}} [\varrho,\mu]= r_{\mathrm{PVM}}[\varrho,\mu]$ \emph{for all} $\mu$, which certainly implies that $R_{\mathrm{POVM}} [\varrho]= R_{\mathrm{PVM}}[\varrho]$. For steering in two-qubit systems, since PVMs are equivalent to $2$-POVMs (see, e.g., Ref.~\cite{chaupra}), the above hypothesis amounts to ask if $r_2[\varrho,\mu]=r_4[\varrho,\mu]$. We test this hypothesis by sampling random states, constructing random LHS ensembles for each state. We then compute $r_2[\varrho,\mu]$ exactly. The computation of $r_4[\varrho,\mu]$ is performed by the simulated annealing algorithm (see below). Although the algorithm in principle only provides an upper bound of $r_4[\varrho,\mu]$, repeated runs indicate that it is close to the exact value of $r_4[\varrho,\mu]$. To our surprise, we find that in any single case, the obtained upper bound of $r_4[\varrho,\mu]$ approaches $r_2[\varrho,\mu]$ from above; see Figure~\ref{fig:POVM}. This 
strongly 
supports the hypothesis that $r_4[\varrho,\mu]=r_2[\varrho,\mu]$ at least for generic states $\rho$ and generic LHS ensembles $\mu$. This in turn supports the conjecture that for two-qubit systems, 
POVMs are equivalent to PVMs. 

%We hope that future improvement could turn this piece of evidence into a numerical proofs, or gives a hint to a full proof of the conjecture. 

%\begin{proposition}
%For a state $\varrho$, if $R_{POVM}[\varrho]=R_{PVM}[\varrho]$, if a LHS ensemble $\mu$ is optimal for steering with PVMs, it must be optimal for POVMs.
%\end{proposition}
%\begin{proof}
%Suppose $\mu$ is optimal for steering with PVMs, namely $R_{PVM}[\varrho]=r_{PVM}[\varrho,\mu]$ but not optimal for steering with POVMs, namely $R_{POVM}[\varrho]  > r_{POVM}[\varrho,\mu]$. However PVM is a subset of POVM, we have $r_{PVM}[\varrho,\mu] \le r_{POVM}[\varrho,\mu]$.
%\end{proof}
%Of course, the string equivalence between POVMs and PVMs implies that the two sets just equal. This proposition is only meaningful if the strong conjecture is wrong (for example in high dim).

%----------------------------------------------------------------------------------
\begin{remark}
Let us make some remarks on the computation of the principal radius.
From the previous section, it is clear that in actual computation, we are principally interested in the case $n=d_A^2$. While it is not obvious from the first look, the optimisation in the computation of $r_n[\varrho,\mu]$ can be limited to some simple subset of POVMs, namely rank-$1$ POVMs. We first note that $r^{-1}_n[\varrho,\mu]$ can be written as
\begin{equation}
\inf \left\{y: y \ge 0, y \ge F^{-1}[\varrho,\mu,Z,E] \right\}.
\end{equation}
Then since the denominator of $F^{-1}[\varrho,\mu,Z,E]$ is positive, we can write $r_n[\varrho,\mu]$ as
\begin{equation}
\sup \{x: x \ge 0, \int \d \mu (\sigma) \max_i \dprod{Z_i}{\sigma} \ge \sum_{i=1}^{n} \Tr [\varrho_x E_i \otimes Z_i]\},
\end{equation}
where $\varrho_x= x \varrho + (1-x) ({\II_A}/{d_A}) \otimes \varrho_B$. The second inequality is required to hold for all POVMs $E$ and arbitrary composite operators $Z$. Note that this inequality is precisely the condition for $\varrho_x$ to be unsteerable with LHS ensemble $\mu$, c.~f. equation~\eqref{eq:general_steering_inequality}. Then we know that it holds for all POVMs $E$ if it holds for all rank-$1$ POVMs; see, e.g., Ref.~\cite{Barret2002a}. 

Now the optimisation~\eqref{eq:general_principal_radius} to compute $r_n[\varrho,\mu]$ when limiting $E$ to rank-$1$ POVMs is completely similar to the computation of the gap function in Ref.~\cite{chaujpa}. We refer to Ref.~\cite{chaujpa} for the detailed description of the simulated annealing algorithm. 
\end{remark}
}

\end{document}